%% file: diss_arxiv.tex
\documentclass[12pt, a4paper, twoside]{report}
\usepackage[margin=3cm]{geometry}
\usepackage[backend=bibtex,style=numeric,sorting=none]{biblatex}

\usepackage{docmute}   % only needed to allow inclusion of proposal.tex

\usepackage{graphicx}
\graphicspath{ {./figures/} }
\usepackage[utf8]{inputenc}
\usepackage{csquotes}
\usepackage{proof}

\usepackage{lscape}         % Makes landscape easier
\usepackage{graphics}       % Graphics commands
\usepackage{multicol}       % Allows table cells to span cols

\usepackage{listings}       % Source code listings
\usepackage{array}          % Array environment
\usepackage[pdfborder={0 0 0}]{hyperref}    % turns references into hyperlinks
\usepackage{amsmath}        % American Mathematical Society
\usepackage{amssymb}        % Maths symbols
\usepackage{amsthm}         % Theorems
\usepackage{verbatim}       % Verbatim blocks
\usepackage{ifthen}         % Conditional processing in tex
\usepackage{xcolor}         % X11 colour names

\newtheorem{theorem}{Theorem}
\newtheorem{corollary}{Corollary}
\newtheorem*{problem*}{Problem}

\usepackage{centernot}

\usepackage{pmboxdraw}
\usepackage{bm}
\usepackage{mathrsfs} 
\usepackage{adjustbox} 
\usepackage{tabu}

\usepackage{fancyhdr}
\usepackage{tikzfig}
\input{qcs.tikzstyles}

\usepackage{scalerel}
\DeclareMathOperator*{\xor}{\scalerel*{\oplus}{\sum}}
\usepackage{stmaryrd}

\usepackage{soul}
\usepackage[toc,page]{appendix}

\newcommand{\bra}[1]{\langle#1|}
\newcommand{\ket}[1]{|#1\rangle}

\DeclareMathOperator*{\Motimes}{\text{\raisebox{0.25ex}{\scalebox{0.9}{$\bigotimes$}}}}

\usepackage{scalefnt}
\usepackage{tikz-cd}
\usepackage{epigraph}
\usepackage{parskip}

\pgfplotsset{compat=1.15}

\definecolor{zxgreen}{RGB}{230,254,230}
\definecolor{zxred}{RGB}{255,135,136}
\definecolor{zxblue}{RGB}{116,116,235} 
\definecolor{zxdgreen}{RGB}{91,107,91}
\definecolor{zxdred}{RGB}{142,94,94}
\definecolor{zxdblue}{RGB}{61,61,77}

\newcommand{\eq}[1]{\mathrel{\overset{\makebox[0pt]{\mbox{\normalfont\tiny\sffamily #1}}}{=}}}
\theoremstyle{definition}
\newtheorem{definition}{Definition}[section]
\usepackage{csquotes}
\usepackage{algorithm}
\usepackage[noend]{algpseudocode}

\usepackage{setspace}
\begin{document}
\def\YZ{\textcolor{zxblue}{-}\nodepart{two}\textcolor{zxgreen}{-}}
\def\YX{\textcolor{zxblue}{-}\nodepart{two}\textcolor{zxred}{-}}

\def\tket{\textsf{t$|$ket$\rangle$}}
\pagenumbering{gobble}
% cover
\vspace*{\fill}
\begin{center}
\LARGE{\textbf{Diagrammatic Design and Study of Ans\"{a}tze for Quantum Machine Learning}} \\

\vspace{4cm}
\includegraphics[height=3cm]{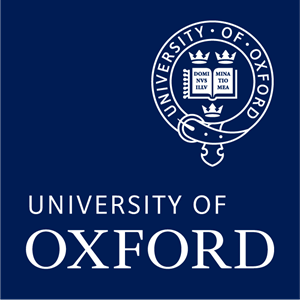}
\vspace{3cm}

\large{Richie Yeung}\\
\normalsize St.\ Cross College\\
University of Oxford 
%\large{Candidate Number: 1040476}

\vspace{3cm}

A thesis submited for the degree of \\
\textit{MSc Computer Science}

Trinity 2020
\end{center}
\vspace*{\fill}
\newpage 
\vspace{1cm}
\section*{Proforma}

{\large
\begin{tabular}{ll}
Candidate Number:   & \bf 1040476 \\
Project Title:      & \bf {Diagrammatic Design and Study of Ans\"{a}tze} \\
		    & \bf {for Quantum Machine Learning} \\
Examination:        & \bf MSc Computer Science 2020  \\
Word Count:         & \bf 11767\footnotemark[1]  \\
Diagram Count:      & \bf 435\footnotemark[2]  \\
\end{tabular}
}

\footnotetext[1]{This word count was computed
by \texttt{detex diss.tex | tr -cd '0-9A-Za-z $\tt\backslash$n' | wc -w}.
}
\footnotetext[2]{This diagram count was computed
by \texttt{find . -name *.tikz | wc -l}.
}
\stepcounter{footnote}

\newpage
\vspace*{\fill}
\vspace{-10mm}
\begin{center}
\subsection*{Acknowledgements}
\parbox{10.5cm}{
Thank you Bob Coecke for your lectures on QCS. 

Thank you Stefano Gogioso for your patient supervision. 

Thank you Aleks Kissinger for your extensive help throughout the project, and for developing PyZX and TikZiT. 

Thank you Ian Fan for listening to my ramblings on ZX calculus and for proofreading my dissertation. 

Finally, thank you to all of my friends and family for supporting me through this challenging and unusual year.
}
\end{center}
\vspace*{\fill}
\newpage
\vspace*{\fill}
\begin{center}\subsection*{Abstract}\end{center}
Given the rising popularity of quantum machine learning (QML), it is important to develop techniques that effectively simplify commonly adopted families of parameterised quantum circuits (commonly known as ans\"{a}tze).
This project pioneers the use of diagrammatic techniques to reason with QML ans\"{a}tze. We take commonly used QML ans\"{a}tze and convert them to diagrammatic form and give a full description of how these gates commute, making the circuits much easier to analyse and simplify. Furthermore, we leverage a combinatorial description of the interaction between CNOTs and phase gadgets to analyse a periodicity phenomenon in layered ans\"{a}tze and also to simplify a class of circuits commonly used in QML.

\vspace*{\fill}
\newpage

\pagenumbering{arabic}

\tableofcontents

\pagestyle{headings}
\pagenumbering{arabic}

\chapter{Introduction}

\section{Background}

Quantum computing recently received mainstream attention when Google demonstrated quantum supremacy \autocite{supreme} using a 53-qubit superconducting quantum computer, but when can quantum supremacy actually be applied in the real world?  While quantum algorithms such as Shor's algorithm and Grover's algorithm appear to have a ``quantum advantage" over their classical counterparts, current quantum computers are of insufficient size and noise tolerance to apply these algorithms to real world scenarios.

Seeing that quantum hardware is still at the ``mainframe" stage, why are we already investigating quantum machine learning? While traditional quantum algorithms require more memory to store its input data, in quantum machine learning one could first ``compress" the input using classical computation before feeding it to the quantum computer; the output of the quantum computer can then be interpreted and used to make a prediction, again using classical computation. The techniques for dimensionality reduction, compression and reconstruction of data have already been extensively studied in the fields of machine learning and statistics. Since these learning techniques are designed to be noise tolerant, they can cope with the noise introduced by a quantum computer.
It is for these reasons we will see quantum machine learning being applied to real world scenarios before 
traditional quantum algorithms.  Furthermore, experimental evidence shows that even relatively small \textit{parameterised quantum circuits} can learn vastly complicated functions that would require a 
much larger classical neural network to learn \autocite{expressive, IQP}.

Quantum computers are a natural tool of choice for simulating processes that are inherently quantum, 
and using quantum machine learning to tackle problems with quantum data may allow us 
to better understand the structure of atoms and molecules and discover new drugs. 
In the field of quantum chemistry, the wavefunction of a water molecule can be naturally encoded using 
50 qubits, whilst on a classical computer the wavefunction is obtained by diagonalising a $\approx10^{11}*10^{11}$ 
Hamiltonian matrix. Recently, Google performed an experiment which models the isomerisation of diazene \autocite{google2020hartree}.

Another application of quantum machine learning is the emerging field of quantum natural language processing (QNLP): the compositional structure of categorical quantum mechanics coincide with the structures used to describe the categorical grammar developed by Lambek \autocite{zeng2016quantum, lambek2008word, coecke2013alternative}, making quantum computers a naturally suited to implement such ``meaning-aware" models of natural language. Quantum machine learning can be used to produce a quantum word embedding, 
which encodes the meaning of a word into a quantum state. Currently the technology of quantum RAM has not fully developed, so it is necessary to learn the word meanings on the fly instead of loading it from memory \autocite{meichanetzidis2020quantum, mediumpost}. This highlights the importance of designing good parameterised quantum circuits for quantum machine learning.

This project develops and applies diagrammatic techniques to analyse parameterised quantum circuits, in order to gain a understanding and intuition behind what makes a good variational quantum circuit.
\begin{itemize}
	\item \textbf{Chapter 2} introduces diagrammatic notation to describe quantum machine learning, and converts commonly used gates from quantum machine learning into ZX calculus.
	\item \textbf{Chapter 3} builds upon the work in chapter 2 and develops more techniques for manipulating parameterised quantum circuits, in particular using phase gadgets and Pauli gadgets.
	\item \textbf{Chapter 4} acts as a reference for diagrammatic quantum machine learning by applying the diagrammatic techniques from chapters 2 and 3 to nineteen quantum circuits taken from quantum machine learning literature. 
	\item \textbf{Chapter 5} describes an optimisation routine that can be applied to a class of circuits commonly used in quantum machine learning.
\end{itemize}

\input{chapter1.tex}

\input{chapter2.tex}

\input{zoo.tex}
\input{chapter3.tex}

\chapter{Conclusion}

\section{Summary of Ideas}
\subsection*{New Notation}
Due to the prevalence of Y rotation gates in QML literature, we introduced a trichromatic notation that allows a compact representation for gates in the Y basis: the graphical depiction of XY and YZ Hadamard gates are designed to illustrate the effect when taking the transpose, the adjoint and the conjugate. The Y spider and the Y Hadamards were used again when providing a graphical representation of the Pauli gadget.

To achieve diagrammatic quantum machine learning, we also introduced notation for differentiating diagrams in ZX calculus. This allowed us to present a typical framework for quantum machine learning and also find the Hamiltonian generator of a one parameter unitary group. In particular, we were able to first present the phase and Pauli gadgets diagrammatically, then use diagrammatic differentiation to express them as matrix exponentials.

\subsection*{QML Reference}

Besides presenting a typical framework for quantum machine learning, we also converted commonly used gates into ZX calculus. We advocate the use of phase gadgets and Pauli gadgets to represent and analyse QML ans\"{a}tze due to their intuitive commutation properties, their elegant symmetric representation, and their susceptibility to optimisation. To demonstrate the generality of this approach, we dedicated a chapter to converting nineteen circuits from QML literature to Pauli gadgets and applying simplifications to them. A particularly successful example is circuit 5, which reduced the number of parameters from 28 to 14.

\subsection*{Abstraction of Phase Gadgets}

The commutation properties of the CNOT and rotation gates allow us to write phase gadgets and CNOTs in terms of binary vectors and binary matrices in $GL(n, 2)$. This allowed us to compute the manner in which phase gadgets transform through a block of CNOTs with ease. Since a monotonic block of CNOTs correspond to a triangular matrix, by proving the group order of such matrices we were able to demonstrate that ans\"{a}tze that arrange their CNOTs monotonically have a linear periodicity rather than the super-polynomial periodicity achievable.

We also developed a simplification algorithm for a special class of circuits constructed using only Z and X phase gadgets. This algorithm utilises the isomorphism between the group actions of CNOT blocks on Z and X phase gadgets, both representable in $GL(n, 2)$. Using simulated annealing, we achieved up to a 60\% of CNOT count and depth in this class of commonly used circuits. This algorithm will be presented at the 4th International Workshop of Quantum Compilation.

%We first developed a trichromatic graphical language that allows us to compactly represent operations in the X, Y and Z basis. We also developed notation for performing diagrammatic differentiation; this allowed us to present the phase gadget diagrammatically, then use differentiation to show it is equivalent to the usual matrix exponential definition. We then moved to phase polynomials --- by using commutative properties of phase gadgets, we show how to obtain the sum-of-paths normal form diagrammatically. Using parameterised phase polynomials as an example, we illustrate the importance of adding ``inverted CNOTs" in the repeating layers of a QML ans\"{a}tze as it generates more unique parameters. 

%The second theme of our work is the idea that Pauli gadgets are the most natural way of reasoning with QML ans\"{a}tze. 

\section{Future Work}

The work in chapter 4 can be developed into a full-fledged quantum machine learning tool that allows the user to design, optimise and analyse ans\"{a}tze. The dependency graph produced during the optimisation progress allows the algorithm to perform automatic differentiation and perform QML on the same landscape but using a simpler circuit. Furthermore, the Jacobian between the old and new parameters will shed light on which parameters affect the circuit output the most.

The optimisation algorithm on Z and X phase gadgets has room for refinement and adaptation. It may be possible to estimate the cost of synthesising the binary matrix use in the optimisation and including that cost in the objective function. More experimentation will reveal whether the savings of this algorithm will carry through to more restrictive machine architectures. Due to the cost of the naive Gaussian elimination algorithm used to perform the matrix inverse, computing the objective function is time consuming; a more optimised implementation is necessary before it can be used in a production quantum compiler.

\printbibliography

\newpage
\begin{appendices}
\chapter{Supplementary Calculations}
\subsection*{ZX Calculus}
This is a justification for the rules given by the ZX calculus. The result that diagrams are equal up to spatial isotopy is proved in \textit{Categories for Quantum Theory} \autocite{heunen2019categories}. \\

\textbf{Fusion Rule}: In the case where the number of output wires in the first spider matches the number of output wires in the second spider:
\[
(e^{-i\frac{\beta}{2}}\ket{0\ldots 0}\bra{0\ldots 0} + e^{i\frac{\beta}{2}}\ket{1\ldots 1}\bra{1\ldots 1}) \circ
(e^{-i\frac{\alpha}{2}}\ket{0\ldots 0}\bra{0\ldots 0} + e^{i\frac{\alpha}{2}}\ket{1\ldots 1}\bra{1\ldots 1})
\]
\[=(e^{-i\frac{\alpha+\beta}{2}}\ket{0\ldots 0}\bra{0\ldots 0} + e^{i\frac{\alpha+\beta}{2}}\ket{1\ldots 1}\bra{1\ldots 1})\]
The general case is difficult to articulate directly using Dirac notation. Instead, we can verify the identity by comparing the result of basis states $\ket{0\ldots 0}$ and $\ket{1\ldots 1}$, as all other basis states map to the zero vector:

\[ \ket{0\ldots 0} \mapsto e^{-i\frac{\alpha}{2}}\ket{0\ldots 0} \mapsto  e^{-i\frac{\alpha+\beta}{2}}\ket{0\ldots 0} \] 
\[ \ket{1\ldots 1} \mapsto e^{i\frac{\alpha}{2}}\ket{1\ldots 1} \mapsto  e^{i\frac{\alpha+\beta}{2}}\ket{1\ldots 1} \] 

\textbf{Hadamard Rule}: 
\[ H^{\otimes n} (e^{-i\frac{\alpha}{2}}\ket{0\ldots 0}\bra{0\ldots 0} + e^{i\frac{\alpha}{2}}\ket{1\ldots 1}\bra{1\ldots 1}) H^{\otimes m} \]
\[= e^{-i\frac{\alpha}{2}}\ket{{+}\ldots {+}}\bra{{+}\ldots {+}} + e^{i\frac{\alpha}{2}}\ket{{-}\ldots {-}}\bra{{-}\ldots {-}} \]

\textbf{Identity Rules}: follow easily from the definitions
$\pi$-copy rule: The X $\pi$ rotation acts as the NOT gate in the Z basis.
\[ X^{\otimes n} (e^{-i\frac{\alpha}{2}}\ket{0\ldots 0}\bra{0\ldots 0} + e^{i\frac{\alpha}{2}}\ket{1\ldots 1}\bra{1\ldots 1}) X^{\otimes m} \]
\[= e^{-i\frac{\alpha}{2}}\ket{1\ldots 1}\bra{1\ldots 1} + e^{i\frac{\alpha}{2}}\ket{0\ldots 0}\bra{0\ldots 0} \]

\textbf{Copy Rule}: The 1-legged X spider is just the pure state $\ket{0}$, and $(\ket{0\ldots 0}\bra{0} + \ket{1\ldots 1}\bra{1}) \ket{0} = \ket{0\ldots 0}$.

\textbf{Bialgebra Rules}: The 0 phase Z spider acts as a ``COPY" gate and the 0 phase X spider acts as a ``XOR" gate --- this can be verified by checking the computational basis. The bialgebra rule is the statement ``Copying the $m$ input bits to be XORed on the $n$ outputs is the same as XORing the $m$ inputs and copying it to the $n$ outputs."

\subsection*{The Rotation Matrices are Similar}

According to the PennyLane documentation, the rotation matrices are defined through the matrix exponential.
\begin{align*}
R_{x}( \phi ) &= e^{-i\phi X /2}\\
& =\begin{pmatrix}
\cos(\phi/2) & -i\sin(\phi/2)\\
-i\sin(\phi/2) &  \cos(\phi/2)\\
\end{pmatrix} \\
R_{y}( \phi ) &= e^{-i\phi Y /2}\\
& =\begin{pmatrix}
\cos(\phi/2) & -\sin(\phi/2)\\
\sin(\phi/2) &  \cos(\phi/2)\\
\end{pmatrix} \\
R_{z}( \phi ) &= e^{-i\phi Z /2}\\
& =\begin{pmatrix}
e^{-i\phi /2} & 0\\
0 & e^{i\phi /2}
\end{pmatrix}
\end{align*}
In the same way that the Pauli matrices are similar matrices, the X, Y, Z rotation gates are 
related through a change of basis:
\begin{align*}
R_{x}( \phi ) & =\begin{pmatrix}
\cos( \phi /2) & -i\sin( \phi /2)\\
-i\sin( \phi /2) & \cos( \phi /2)
\end{pmatrix}\\
 & =\frac{1}{2}\begin{pmatrix}
1 & 1\\
1 & -1
\end{pmatrix}\begin{pmatrix}
e^{-i\phi /2} & 0\\
0 & e^{i\phi /2}
\end{pmatrix}\begin{pmatrix}
1 & 1\\
1 & -1
\end{pmatrix}\\
 & =HR_{z}( \phi ) H \\
R_{y}( \phi ) & =\begin{pmatrix}
\cos( \phi /2) & -\sin( \phi /2)\\
\sin( \phi /2) & \cos( \phi /2)
\end{pmatrix}\\
 & =\frac{1}{2}\begin{pmatrix}
1 & i\\
i & 1
\end{pmatrix}\begin{pmatrix}
e^{-i\phi /2} & 0\\
0 & e^{i\phi /2}
\end{pmatrix}\begin{pmatrix}
1 & -i\\
-i & 1
\end{pmatrix}\\
 & =R_{x}( -\pi /2) R_{z}( \phi ) R_{x}( \pi /2)\\
 & =\frac{1}{2}
{\footnotesize\begin{pmatrix}
1-i & 0\\
0 & 1+i
\end{pmatrix}\begin{pmatrix}
\cos(\phi/2) & -i\sin(\phi/2)\\
-i\sin(\phi/2) & \cos(\phi/2)
\end{pmatrix}\begin{pmatrix}
1+i & 0\\
0 & 1-i
\end{pmatrix}}\\
 & =R_{z}( \pi /2) R_{x}( \phi ) R_{z}( -\pi /2)
\end{align*}

\section*{Y Basis}
\begin{align*}
H_{XY} &= \begin{pmatrix}1&i\\i&1\end{pmatrix} \begin{pmatrix}1&1\\1&-1\end{pmatrix} \begin{pmatrix}1&-i\\-i&1\end{pmatrix}\\
	   &= \begin{pmatrix}1&i\\i&1\end{pmatrix} \begin{pmatrix}1-i&1-i\\1+i&-1-i\end{pmatrix} = \begin{pmatrix}&1-i\\1+i&\end{pmatrix} \\ 
H_{YZ} &= \begin{pmatrix}1&\\&i\end{pmatrix} \begin{pmatrix}1&1\\1&-1\end{pmatrix} \begin{pmatrix}1&\\&-i\end{pmatrix}\\
	   &= \begin{pmatrix}1&\\&i\end{pmatrix} \begin{pmatrix}1&-i\\1&i\end{pmatrix} = \begin{pmatrix}1&-i\\i&-1\end{pmatrix} \\
\end{align*}

\section*{Proof of Euler Decomposition} \label{euler}
\setcounter{theorem}{3}
\begin{theorem}
For $a_1, a_2, a_3 \in (-\pi, \pi]$, there exists $b_1, b_2, b_3 \in (-\pi, \pi]$ such that
$$R_{x}(a_{3}) R_{z}(a_{2}) R_{x}(a_{1}) = R_{z}(b_{3}) R_{x}(b_{2}) R_{z}(b_{1})$$
\begin{center} and \end{center}
$$R_{z}(b_{3}) R_{x}(b_{2}) R_{z}(b_{1}) = R_{x}(a_{3}) R_{z}(a_{2}) R_{x}(a_{1}),$$
where $b_1, b_2, b_3$ is defined using the auxiliary variables $z_1, z_2$: 
$$z_{1} =\cos \frac{a_2}{2}\cos\left(\frac{a_{1}+a_{3}}{2}\right) +i\sin \frac{a_2}{2}\cos\left(\frac{a_{1}-a_{3}}{2}\right)$$
$$z_{2} =\cos \frac{a_2}{2}\sin\left(\frac{a_{1}+a_{3}}{2}\right) -i\sin \frac{a_2}{2}\sin\left(\frac{a_{1}-a_{3}}{2}\right)$$
\begin{align*}
    b_{1} &=\arg z_{1} +\arg z_{2}\\
    b_{2} &=\arctan\left(\frac{|z_{2} |}{|z_{1} |}\right)\\
          &=\arg( |z_{1} |+i|z_{2} |)\\
    b_{3} &=\arg z_{1} -\arg z_{2}
\end{align*}
\end{theorem}

\begin{proof}
    Define $\displaystyle \alpha _{i} =\frac{1}{2} a_{i}$ and $\displaystyle \beta _{i} =\frac{1}{2} b_{i}$ for typesetting convenience. Expanding the definitions of $R_z(\theta)$ and $R_x(\theta)$ gives us the two matrices in explicit form.
\begin{gather*}
R_{XZX}(a_{1} ,a_{2} ,a_{3}) =R_{x}(a_{3}) R_{z}(a_{2}) R_{x}(a_{1})\\
=\begin{pmatrix}
\cos \alpha _{1} & -i\sin \alpha _{1}\\
-i\sin \alpha _{1} & \cos \alpha _{1}
\end{pmatrix}\begin{pmatrix}
e^{-i\alpha _{2}} & 0\\
0 & e^{i\alpha _{2}}
\end{pmatrix}\begin{pmatrix}
\cos \alpha _{3} & -i\sin \alpha _{3}\\
-i\sin \alpha _{3} & \cos \alpha _{3}
\end{pmatrix}\\
=\footnotesize{\begin{pmatrix}
e^{-i\alpha _{2}}\cos \alpha _{1}\cos \alpha _{3} -e^{i\alpha _{2}}\sin \alpha _{1}\sin \alpha _{3} & -i\left[ e^{-i\alpha _{2}}\cos \alpha _{1}\sin \alpha _{3} +e^{i\alpha _{2}}\sin \alpha _{1}\cos \alpha _{3}\right]\\
-i\left[ e^{i\alpha _{2}}\cos \alpha _{1}\sin \alpha _{3} +e^{-i\alpha _{2}}\sin \alpha _{1}\cos \alpha _{3}\right] & e^{i\alpha _{2}}\cos \alpha _{1}\cos \alpha _{3} -e^{-i\alpha _{2}}\sin \alpha _{1}\sin \alpha _{3}
\end{pmatrix}}\\
=\begin{pmatrix}
z^{*}_{1} & -iz^{*}_{2}\\
-iz_{2} & z_{1}
\end{pmatrix}
%z_{1} =\cos \alpha _{2}\cos( \alpha _{1} +\alpha _{3}) +i\sin \alpha _{2}\cos( \alpha _{1} -\alpha _{3})\\
%z_{2} =\cos \alpha _{2}\sin( \alpha _{1} +\alpha _{3}) -i\sin \alpha _{2}\sin( \alpha _{1} -\alpha _{3})
\end{gather*}

\begin{gather*}
R_{ZXZ}( b_{1} ,b_{2} ,b_{3}) =R_{z}( b_{3}) R_{x}( b_{2}) R_{z}( b_{1})\\
=\begin{pmatrix}
e^{-i\beta _{1}} & 0\\
0 & e^{i\beta _{1}}
\end{pmatrix}\begin{pmatrix}
\cos( \beta _{2}) & -i\sin( \beta _{2})\\
-i\sin( \beta _{2}) & \cos( \beta _{2})
\end{pmatrix}\begin{pmatrix}
e^{-i\beta _{3}} & 0\\
0 & e^{i\beta _{3}}
\end{pmatrix}\\
=\begin{pmatrix}
e^{-i( \beta _{1} +\beta _{3})}\cos \beta _{2} & e^{-i( \beta _{1} -\beta _{3})}\sin \beta _{2}\\
e^{i( \beta _{1} -\beta _{3})}\sin \beta _{2} & e^{i( \beta _{1} +\beta _{3})}\cos \beta _{2}
\end{pmatrix}
\end{gather*}
By comparing magnitude and modulus of the matrix entries, we obtain the following equations which can be solved to find $b_1, b_2, b_3$.
\begin{align*}
    \alpha _{1} +\alpha _{3} &=\arg z_{1} &\sin \alpha _{2} &=|z_{2} |\\
    \alpha _{1} -\alpha _{3} &=\arg z_{2} &\cos \alpha _{2} &=|z_{1} |
\end{align*}
To check that the second equation holds, we use the fact that $H^2 = I$ and $R_x(\theta) = HR_z(\theta)H$ to apply a Hadamard change of basis to both sides. In ZX calculus, this corresponds to conjugating the two colours.
$$R_{x}(a_{3}) R_{z}(a_{2}) R_{x}(a_{1}) = R_{z}(b_{3}) R_{x}(b_{2}) R_{z}(b_{1})$$
$$HR_{x}(a_{3}) R_{z}(a_{2}) R_{x}(a_{1})H = HR_{z}(b_{3}) R_{x}(b_{2}) R_{z}(b_{1})H$$
$$R_{z}(a_{3})HR_{z}(a_{2})HR_{z}(a_{1}) = R_{x}(b_{3})HR_{x}(b_{2})HR_{x}(b_{1})$$
$$R_{z}(a_{3})R_{x}(a_{2})R_{z}(a_{1}) = R_{x}(b_{3})R_{z}(b_{2})R_{x}(b_{1})$$
\end{proof}
\end{appendices}
\end{document}

%% file: qcs.tikzstyles
\tikzstyle{node}=[minimum size=0.3cm]
\tikzstyle{Z}=[fill={rgb,255:red,230; green,254; blue,230}, draw={rgb,255: red,61; green,77; blue,61}, shape=circle]
\tikzstyle{X}=[fill={rgb,255:red,255; green,135; blue,136}, draw={rgb,255: red,102; green,54; blue,54}, shape=circle]
\tikzstyle{Y}=[fill=zxblue, draw=zxdblue, shape=circle]
\tikzstyle{Z_big}=[fill={rgb,255:red,230; green,254; blue,230}, draw={rgb,255: red,61; green,77; blue,61}, shape=circle, minimum width=1.6em, font={\small}]
\tikzstyle{X_big}=[fill={rgb,255:red,255; green,135; blue,136}, draw={rgb,255: red,102; green,54; blue,54}, shape=circle, minimum width=1.6em, font={\small}]
\tikzstyle{Y_big}=[fill=zxblue, draw=zxdblue, shape=circle, minimum width=1.6em, font={\small}]
\tikzstyle{Z_tri}=[fill={rgb,255:red,230; green,254; blue,230}, draw={rgb,255: red,61; green,77; blue,61}, regular polygon, regular polygon sides=3, draw, shape border rotate=0, inner sep=0pt, minimum width=15pt, line width=0.75]
\tikzstyle{X_tri}=[fill={rgb,255:red,255; green,135; blue,136}, draw={rgb,255: red,102; green,54; blue,54}, regular polygon, regular polygon sides=3, draw, shape border rotate=0, inner sep=0pt, minimum width=15pt, line width=0.75]
\tikzstyle{Z_tri_inv}=[fill={rgb,255:red,230; green,254; blue,230}, draw={rgb,255: red,61; green,77; blue,61}, regular polygon, regular polygon sides=3, draw, shape border rotate=180, inner sep=0pt, minimum width=15pt, line width=0.75, font={\footnotesize}]
\tikzstyle{X_tri_inv}=[fill={rgb,255:red,255; green,135; blue,136}, draw={rgb,255: red,102; green,54; blue,54}, regular polygon, regular polygon sides=3, draw, shape border rotate=180, inner sep=0pt, minimum width=15pt, line width=0.75]
\tikzstyle{Z_tri_l}=[fill={rgb,255:red,230; green,254; blue,230}, draw={rgb,255: red,61; green,77; blue,61}, regular polygon, regular polygon sides=3, draw, shape border rotate=90, inner sep=0pt, minimum width=15pt, line width=0.75]
\tikzstyle{X_tri_l}=[fill={rgb,255:red,255; green,135; blue,136}, draw={rgb,255: red,102; green,54; blue,54}, regular polygon, regular polygon sides=3, draw, shape border rotate=90, inner sep=0pt, minimum width=15pt, line width=0.75]
\tikzstyle{H}=[fill=yellow, draw=black, shape=rectangle, minimum width=2mm, minimum height=2mm]
\tikzstyle{Y_box}=[draw=black, shape=rectangle, minimum width=2mm, minimum height=2mm, tikzit fill={rgb,255: red,255; green,128; blue,0}]
\tikzstyle{Z_med}=[fill={rgb,255:red,230; green,254; blue,230}, draw={rgb,255: red,61; green,77; blue,61}, shape=circle, minimum width=1.1em, font={\footnotesize}]
\tikzstyle{X_med}=[fill={rgb,255:red,255; green,135; blue,136}, draw={rgb,255: red,102; green,54; blue,54}, shape=circle, minimum width=1.1em, font={\footnotesize}]
\tikzstyle{Y_med}=[fill=zxblue, draw=zxdblue, font={\footnotesize}, minimum width=1.1em, font={\footnotesize}, shape=circle]
\tikzstyle{ZP}=[fill={rgb,255:red,230; green,254; blue,230}, draw={rgb,255: red,61; green,77; blue,61}, regular polygon, regular polygon sides=3, shape border rotate=30]
\tikzstyle{ZM}=[fill={rgb,255:red,230; green,254; blue,230}, draw={rgb,255: red,61; green,77; blue,61}, regular polygon, regular polygon sides=3, shape border rotate=-30]
\tikzstyle{XP}=[fill={rgb,255:red,255; green,135; blue,136}, draw={rgb,255: red,102; green,54; blue,54}, regular polygon, regular polygon sides=3, shape border rotate=30]
\tikzstyle{XM}=[fill={rgb,255:red,255; green,135; blue,136}, draw={rgb,255: red,102; green,54; blue,54}, regular polygon, regular polygon sides=3, shape border rotate=-30]
\tikzstyle{small_box}=[fill=white, draw=black, shape=rectangle, minimum width=1.5cm, minimum height=1.5cm, font={\footnotesize}]
\tikzstyle{med_rectangle}=[fill=white, draw=black, shape=rectangle, minimum width=2.8cm, minimum height=1.5cm, font={\footnotesize}]
\tikzstyle{big_rectangle}=[fill=white, draw=black, shape=rectangle, minimum width=4.5cm, minimum height=1.5cm, font={\footnotesize}]
\tikzstyle{smol_box}=[fill=white, draw=black, shape=rectangle, minimum width=1cm, minimum height=1cm]
\tikzstyle{poo}=[minimum height=0.7cm, minimum width=0.7cm, path picture={\node at (path picture bounding box.center) {\includegraphics[width=0.7cm] {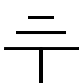}};}]
\tikzstyle{Z_long}=[fill={rgb,255:red,230; green,254; blue,230}, draw={rgb,255: red,61; green,77; blue,61}, shape=rectangle, rounded corners=0.25cm, minimum height=0.5cm, inner sep=0.25em]
\tikzstyle{X_long}=[fill={rgb,255:red,255; green,135; blue,136}, draw={rgb,255: red,102; green,54; blue,54}, shape=rectangle, rounded corners=0.25cm, minimum height=0.5cm, inner sep=0.25em]
\tikzstyle{med_rectv}=[fill=white, draw=black, shape=rectangle, minimum width=1.1cm, minimum height=2.4cm, font={\scriptsize}, inner sep=0.2cm]
\tikzstyle{Z dot}=[fill={rgb,255: red,0; green,127; blue,0}, draw=black, shape=circle, minimum width=1.5mm]
\tikzstyle{X dot}=[fill={rgb,255:red,255; green,21; blue,0}, draw=black, shape=circle, minimum width=1.5mm]
\tikzstyle{Z phase dot}=[fill={rgb,255: red,0; green,127; blue,0}, draw=black, shape=circle, font={\footnotesize}]
\tikzstyle{X phase dot}=[fill={rgb,255:red,255; green,21; blue,0}, draw=black, shape=circle, font={\footnotesize}]
\tikzstyle{ZYa}=[draw=black, shape=rectangle, rectangle split, rectangle split parts=2, rectangle split horizontal, rectangle split part fill={zxgreen, zxblue}, rectangle split draw splits=false, minimum height=2mm, font={\tiny}]
\tikzstyle{YZ}=[draw=black, shape=rectangle, rectangle split, rectangle split parts=2, rectangle split horizontal, rectangle split part fill={zxblue, zxgreen}, rectangle split draw splits=false, minimum height=2mm, font={\tiny}]
\tikzstyle{XYa}=[draw=black, shape=rectangle, rectangle split, rectangle split parts=2, rectangle split horizontal, rectangle split part fill={zxred, zxblue}, rectangle split draw splits=false, minimum height=2mm, font={\tiny}]
\tikzstyle{YX}=[draw=black, shape=rectangle, rectangle split, rectangle split parts=2, rectangle split horizontal, rectangle split part fill={zxblue, zxred}, rectangle split draw splits=false, minimum height=2mm, font={\tiny}]

\tikzstyle{dashs}=[-, dashed, line width=0.15mm]
\tikzstyle{thick}=[-, line width=0.5mm]
\tikzstyle{arrow}=[->]

%% file: chapter1.tex
\chapter{Diagrammatic QML}
Before we apply diagrammatic techniques to analyse quantum machine learning, we first provide a diagrammatic presentation of quantum machine learning. 
We convert the elementary gates used in quantum circuits to diagrammatic form, first using ZX calculus in section \ref{zx def}, then later using the trichromatic language developed in section \ref{y}. To work entirely within the framework of diagrammatic calculus, section \ref{diadiff} introduces a novel method to differentiate circuits diagrammatically, which can be used to obtain the Hamiltonian of a system entirely diagrammatically. The method also allows us to diagrammatically describe a typical framework for performing quantum machine learning in section \ref{qnndiff}.

\section{Dirac notation}\label{dirac}
This section defines the quantum gates used by popular quantum computing libraries such as PennyLane 
\autocite{pennylane}, Cirq \autocite{cirq} and Qiskit \autocite{Qiskit-Textbook}, using matrix and Dirac notation.

Define $|x\rangle$ as the column vector with 1 in the $x$th coordinate and 0 everywhere else, and $\langle x|$ as its adjoint. As a special case, define 
$\ket{+}=\frac{1}{\sqrt{2}} (\ket{0}+\ket{1})$ and $\ket{-}=\frac{1}{\sqrt{2}} (\ket{0}-\ket{1})$, which represents the Hadamard basis. The Hadamard matrix $H$ serves as an involutive transformation from the standard computational basis to the Hadamard basis: 
\begin{align*}
H &= \ket{0}\bra{+} + \ket{1}\bra{-} \\
  &= \ket{+}\bra{0} + \ket{-}\bra{1} \\
  &= \frac{1}{\sqrt{2}} \begin{pmatrix}1&1\\1&-1\end{pmatrix} 
\end{align*} 

The $Z, X, Y$ Pauli matrices are unitary and Hermitian, and when combined with identity matrix $I$, form 
a basis for $2\times 2$ Hermitian matrices. It is for this reason that the identity matrix is considered 
the fourth Pauli matrix.

$$
Z=\begin{pmatrix}1&0\\0&-1\end{pmatrix} \qquad
X=\begin{pmatrix}0&1\\1&0\end{pmatrix} \qquad
Y=\begin{pmatrix}0&-i\\i&0\end{pmatrix} \qquad
I=\begin{pmatrix}1&0\\0&1\end{pmatrix}
$$
	
The rotation matrices in the Z, X, Y basis are similar, and hence are related through a change of basis. They are related to the Pauli matrices through Stone's theorem 
(See section \ref{stones}). 

\begin{align*}
R_{z}(\phi) &= e^{-i\phi /2}\ket{0}\bra{0} + e^{i\phi /2}\ket{1}\bra{1} \\
		    &= \begin{pmatrix}
e^{-i\phi /2} & 0\\
0 & e^{i\phi /2}
\end{pmatrix} \\ \\
R_{x}( \phi ) &= HR_{z}(\phi)H \\
	      &= e^{-i\phi /2}\ket{+}\bra{+} + e^{i\phi /2}\ket{-}\bra{-} \\
	      &= \begin{pmatrix}
\cos( \phi /2) & -i\sin( \phi /2) \\
-i\sin( \phi /2) & \cos( \phi /2)
\end{pmatrix} \\ \\
R_{y}( \phi ) &= R_{x}(-\pi /2) R_{z}(\phi) R_{x}(\pi /2) \\
	      &= R_{z}(\pi /2) R_{x}(\phi) R_{z}(-\pi /2) \\
	      &= \begin{pmatrix}
\cos( \phi /2) & -\sin( \phi /2)\\
\sin( \phi /2) & \cos( \phi /2)
\end{pmatrix}.
\end{align*}

Controlled gates, such as the controlled X and controlled Z gates, apply their operation on the target qubit when the control qubit is set to 1. They can be written in matrix form by expanding their action on the computational basis:
\begin{center}
$\begin{aligned}
\mathrm{CX} & =\ket{00}\bra{00} +\ket{01}\bra{01} & \qquad \mathrm{CZ} & =\ket{00}\bra{00} +\ket{01}\bra{01}\\
 & + \ket{11}\bra{10} +\ket{10}\bra{11} &  & + \ket{10}\bra{10} -\ket{11}\bra{11}\\
 & =\begin{pmatrix}
1 & 0 & 0 & 0\\
0 & 1 & 0 & 0\\
0 & 0 & 0 & 1\\
0 & 0 & 1 & 0
\end{pmatrix} &  & =\begin{pmatrix}
1 & 0 & 0 & 0\\
0 & 1 & 0 & 0\\
0 & 0 & 1 & 0\\
0 & 0 & 0 & -1
\end{pmatrix}
\end{aligned}$
\end{center}

\section{ZX calculus}\label{zx def}
This section introduces the ZX calculus developed by Coecke and Duncan \autocite{zx}, which is an extension of the categorical quantum mechanics developed by Abramsky and Coecke \autocite{cqm}. Categorical quantum mechanics utilises dagger 7ymmetric monoidal categories as a rigorous framework for reasoning with quantum processes.

\subsection{ZX Spiders} \label{zx}

\begin{align*}
    \left\llbracket\input{zx/z_spider.tikz}\right\rrbracket &\eq{z def} e^{-i \alpha/2}|0 \ldots 0\rangle\langle 0 \ldots 0|+e^{i \alpha/2}|1 \ldots 1\rangle\langle 1 \ldots 1| \\
    \left\llbracket\input{zx/x_spider.tikz}\right\rrbracket &\eq{x def} e^{-i \alpha/2}|{+}\ldots{+}\rangle\langle{+}\ldots{+}|+e^{i \alpha/2}|{-}\ldots{-}\rangle\langle{-}\ldots{-}| 
\end{align*}
ZX diagrams are constructed using the green Z spiders and the red X spiders defined above. Unless specified, diagrams in this work are read from left to right, with the left wires as the input wires and the right wires as the output wires. Spiders with $\alpha = 0$ have their phases omitted for brevity.
The double square brackets represents the \textit{interpretation functor}, 
which converts the diagrams back to the usual Hilbert space notation. 
The interpretation functor is contravariant, so $\llbracket D_1 \cdot D_2 \rrbracket = \llbracket D_2 \rrbracket \circ \llbracket D_1 \rrbracket$.

Z and X spiders are the universal building blocks for quantum computation ---
all quantum operations can be constructed using these two spiders. 
For starters, we have the RZ and RX rotation gates:
\begin{align*}
	\left\llbracket\input{zx/rz.tikz}\right\rrbracket &= e^{-i \alpha/2}\ket{0}\bra{0} + e^{i \alpha/2}\ket{1}\bra{1} \\
	\left\llbracket\input{zx/rx.tikz}\right\rrbracket &= e^{-i \alpha/2}\ket{{+}}\bra{{+}} + e^{i \alpha/2}\ket{{-}}\bra{{-}} 
\end{align*}
The Z and X Pauli matrices can be viewed as $\pi$ rotation gates, up to a phase:
\begin{align*}
	\llbracket\input{zx/pz.tikz}\rrbracket &= -i\ket{0}\bra{0} + i\ket{1}\bra{1} = -iZ \\
	\llbracket\input{zx/px.tikz}\rrbracket &= -i\ket{{+}}\bra{{+}} + i\ket{{-}}\bra{{-}} \\
					    &= -i\ket{0}\bra{1} - i\ket{1}\bra{0} = -iX
\end{align*}
The CNOT and CZ gates are written as a composition of multiple spiders. By interpreting the components of the gates in Dirac notation, one can verify that they do combine to match the original definition in section \ref{dirac}.
\begin{align*}
\mathrm{CNOT} = \left\llbracket\input{zx/check/cx.tikz}\right\rrbracket &= 
\left\llbracket\input{zx/check/cx_1.tikz} \bullet \input{zx/check/cx_2.tikz} \right\rrbracket\\
				      &= \left\llbracket\input{zx/check/cx_2.tikz}\right\rrbracket \circ \left\llbracket\input{zx/check/cx_1.tikz}\right\rrbracket 
\end{align*}
\begin{align*}
\mathrm{CZ} = \left\llbracket\input{zx/check/cz.tikz}\right\rrbracket &= 
\left\llbracket\input{zx/check/cx_1.tikz} \bullet \input{zx/check/cz_1.tikz} \right\rrbracket\\
				      &= \left\llbracket\input{zx/check/cz_1.tikz}\right\rrbracket \circ \left\llbracket\input{zx/check/cx_1.tikz}\right\rrbracket 
\end{align*}
\begin{align*}
    \left\llbracket\input{zx/check/cx_1.tikz}\right\rrbracket
                            &= (\ket{00}\bra{0} + \ket{11}\bra{1})\otimes(\ket{0}\bra{0}+\ket{1}\bra{1})     \\ 
			    &= \ket{000}\bra{00} + \ket{011}\bra{01} + \ket{100}\bra{10} + \ket{111}\bra{11}     \\\\
    \left\llbracket\input{zx/check/cz_1.tikz}\right\rrbracket
			    &= (\ket{0}\bra{0}+\ket{1}\bra{1})\otimes (\ket{0}\bra{00} + \ket{1}\bra{11})
    				(I\otimes H\otimes I)\\
			    &= \ket{00}\bra{000} + \ket{00}\bra{010} + \ket{01}\bra{001} + \ket{01}\bra{011} \\
		            & \:\quad + \ket{10}\bra{100} - \ket{10}\bra{110} + \ket{11}\bra{101} - \ket{11}\bra{111} 
\end{align*}
\begin{align*}
    \left\llbracket\input{zx/check/cx_2.tikz}\right\rrbracket
			    &= (\ket{0}\bra{0}+\ket{1}\bra{1})\otimes(\ket{{+}}\bra{{+}{+}} + \ket{{-}}\bra{{-}{-}})     \\ 
			    &= (\ket{0}\bra{0}+\ket{1}\bra{1})\otimes
			    (\ket{0}\bra{00} + \ket{1}\bra{01} + \ket{1}\bra{10} + \ket{0}\bra{11})     \\ 
			    &= \ket{00}\bra{000} + \ket{01}\bra{001} + \ket{01}\bra{010} + \ket{00}\bra{011}  \\
			    & \:\quad + \ket{10}\bra{100} + \ket{11}\bra{101} + \ket{11}\bra{110} + \ket{10}\bra{111} 
\end{align*}
\begin{center}
$\begin{aligned}
\mathrm{CX} & =\ket{00}\bra{00} +\ket{01}\bra{01} & \qquad \mathrm{CZ} & =\ket{00}\bra{00} +\ket{01}\bra{01}\\
 & + \ket{11}\bra{10} +\ket{10}\bra{11} &  & + \ket{10}\bra{10} -\ket{11}\bra{11}\\
\end{aligned}$
\end{center}

ZX diagrams are invariant under spatial isotopy: if one diagram can be deformed into another, then they are equal.
By straightening the wires we get a representation that more resembles a gate in conventional circuit notation.
\begin{align*}
\mathrm{CNOT} &= \input{zx/check/cx.tikz} = \input{zx/cnot.tikz} \\
\mathrm{CZ} &= \input{zx/check/cz.tikz} = \input{convert_all/cz1.tikz}
\end{align*}
The Hadamard gate can be written equivalently in the following four ways; they can be verified by expanding the definitions of the spiders.
\begin{align*}
	\input{zx/h.tikz}\eq{h def}\input{zx/h_def_1.tikz} &= \input{zx/h_def_2.tikz} \\
                  =\input{zx/h_def_3.tikz} &= \input{zx/h_def_4.tikz}
\end{align*}
Even the Z and X basis states can be expressed directly using spiders:
the Z basis states can be expressed using X spiders and the X basis states can be expressed using Z spiders. 
\begin{align*}
	\ket{+} &= \llbracket\input{zx/X0.tikz}\rrbracket = \ket{0} + \ket{1} = \llbracket\input{zx/X0S.tikz}\rrbracket \\
	\ket{0} &= \llbracket\input{zx/Z0.tikz}\rrbracket = \ket{+} + \ket{-} = \llbracket\input{zx/Z0S.tikz}\rrbracket \\
	\ket{-} &= \llbracket\input{zx/X1.tikz}\rrbracket = \ket{0} - \ket{1} = -i\llbracket\input{zx/X1S.tikz}\rrbracket \\
	\ket{1} &= \llbracket\input{zx/Z1.tikz}\rrbracket = \ket{+} - \ket{-} = -i\llbracket\input{zx/Z1S.tikz}\rrbracket
\end{align*}
All of the above definitions can be checked to match the matrix unitary definitions by expanding the bra-ket notation. This has been done to ensure the work done in this project matches to the quantum gates defined in PennyLane, Cirq and Qiskit.

Not only can all unitary functions can be expressed in this diagrammatic notation, but the diagrams can be manipulated
using the rules for ZX calculus.
Here we have presented the rules for a scalar-free ZX calculus: since a quantum circuit must be unitary, working with a scalar-free ZX calculus allows us to preserve the circuit up to a global non-zero scalar, which does not alter the distribution of observed states.

\begin{center}
	\subsection*{The ZX Calculus}
\begin{tabular}{ cc }
    Fusion Rule & Hadamard Rule \\
    $ \input{zx/rules/z_fuse1.tikz} \eq{f} \input{zx/rules/z_fuse2.tikz} $ &
$ \input{zx/rules/h1.tikz} \eq{h} \input{zx/rules/h2.tikz} $ \\ \\
    Identity Rules & $\pi$-copy Rule \\
$\begin{aligned} \input{zx/rules/i11.tikz} \eq{i1} \input{zx/rules/i12.tikz} \\
     \input{zx/rules/i21.tikz} \eq{i2} \input{zx/rules/i22.tikz} \end{aligned}$ &
    $ \input{zx/rules/pi1.tikz} \eq{$\pi$} \input{zx/rules/pi2.tikz} $ \\ \\
    Copy Rule & Bialgebra Rule \\
    $ \input{zx/rules/c1.tikz} \eq{c} \input{zx/rules/c2.tikz} $ &
$ \input{zx/rules/b1.tikz} \eq{b} \input{zx/rules/b2.tikz} $ \\
\end{tabular}
\end{center}

The ZX calculus remains true under colour-swapping: by applying the Hadamard gate to all input and output legs, we see that the rules above are invariant under a Hadamard change of basis. The ZX calculus is also invariant under spatial isotopy: if two ZX diagrams can be deformed into one another, then they are equal. 

We have presented a sound but incomplete version of ZX calculus as it is easier to work with; a complete version is presented by Wang \autocite{wang2018completeness}. Regardless, we can show many useful identities, such as the complementarity rule: 
\vspace{-0.75cm}
\begin{center}
    $$ \input{zx/rules/comp0.tikz} \eq{comp} \input{zx/rules/comp4.tikz} $$
    $$ \input{zx/rules/comp0.tikz} \eq{f, id1} 
       \input{zx/rules/comp1.tikz} \eq{b} \input{zx/rules/comp2.tikz} \eq{c}
       \input{zx/rules/comp3.tikz} \eq{f, id1} \input{zx/rules/comp4.tikz} $$
\end{center}
The complementarity rule can be used to prove that a SWAP gate can be implemented by three CNOT gates: 

\ctikzfig{zx/swap}

By writing the basis states in terms of spiders, it is clear that the X Pauli gate acts as the NOT gate for the Z basis and the Z Pauli gate acts as the NOT gate for the X basis. In the general case, writing the basis states in terms of spiders shows that the 1-to-$n$ Z spider acts as a COPY gate and the $n$-to-1 X spider acts as a XOR gate to the computational Z basis.

\[ \input{zx/copy1.tikz} = \input{zx/copy2.tikz} \eq{$\pi$, c} \input{zx/copy3.tikz} = \input{zx/copy4.tikz} \]
\[ \input{zx/xor1.tikz} = \input{zx/xor2.tikz} \eq{f} \input{zx/xor3.tikz} = \input{zx/xor4.tikz} \]
\[ \sigma = x \oplus y \oplus z \]

\subsection{Y Spiders}\label{y}
Although ZX calculus, a graphical language describing the interaction of two complementary observables, 
and is sufficient to describe any quantum process, it is possible and perhaps more elegant to work with three 
complementary observables. In fact, it has been shown that one can fit three complementary 1-qubit observables and no more. \autocite{wocjan2004new}
We start by introducing a third basis, the Y basis. Unlike the Z and X bases, the Y basis is not Hermitian.
\begin{align*}
	\ket{R} &= \frac{1}{\sqrt{2}}(\ket{0} +i \ket{1}) &
	\bra{R} &= \frac{1}{\sqrt{2}}(\bra{0} -i \bra{1}) \\
	\ket{L} &= \frac{1}{\sqrt{2}}(\ket{0} -i \ket{1}) &
	\bra{L} &= \frac{1}{\sqrt{2}}(\bra{0} +i \bra{1})
\end{align*}
The isomorphism between $SO(3)$ and the quotient group $SU(2)/\{\pm I\}$ \autocite{su2so3} allows us to visualise the qubit in the Bloch sphere. 
Here we adopt a right-handed convention by choosing $\ket{R}$ to be the first Y basis --- this results in the 
principal basis vectors of the XYZ bases to be arranged in a right-handed manner.
The Y spider is thus given as: 
\begin{align*}
    \left\llbracket\input{yrules/y_spider.tikz}\right\rrbracket &\eq{y def} e^{-i \alpha/2}|R \ldots R\rangle\langle R \ldots R|+e^{i \alpha/2}|L \ldots L\rangle\langle L \ldots L|
\end{align*}
The right-handed change of basis maps are the given by the $H_{XY}$, $H_{YZ}$ and $H_{ZX}$ Hadamards, and 
the left-handed change of basis maps are given by the $H_{ZY}$, $H_{YX}$ and $H_{XZ}$ Hadamards. 
\begin{align*}
	H_{XY} &= R_x\left(\frac{-\pi}{2}\right)HR_x\left(\frac{\pi}{2}\right) \\
	H_{YZ} &= R_z\left(\frac{\pi}{2}\right)HR_z\left(\frac{-\pi}{2}\right) \\\\
	H_{YX} &= R_x\left(\frac{\pi}{2}\right)HR_x\left(\frac{-\pi}{2}\right) \\
	H_{ZY} &= R_z\left(\frac{-\pi}{2}\right)HR_z\left(\frac{\pi}{2}\right) 
\end{align*}
\[ H_{XY} = H_{YX}^T \qquad H_{YZ} = H_{ZY}^T \]
The $H_{ZX}$ and $H_{XZ}$ Hadamards coincide and correspond to the usual Hadamard matrix $H$, which is symmetric.
The other Hadamards are not symmetric and taking the transpose switches the chirality of the map. 
These properties, combined with the fact that all the Hadamards gates are Hermitian lead us to the following 
diagrammatic representation of the Hadamards:
\begin{figure}
\begin{center}\includegraphics[width=7.5cm]{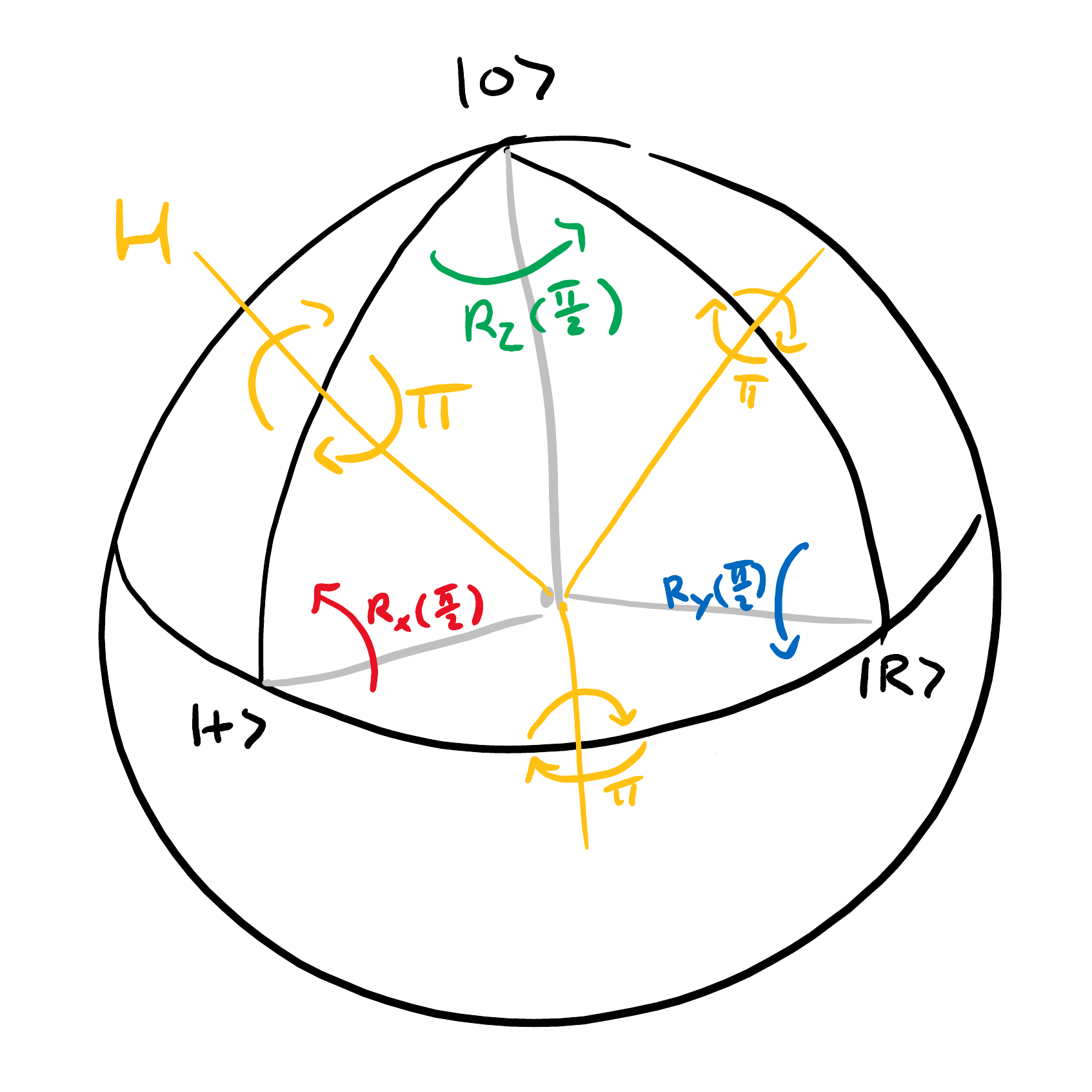}\end{center}
\caption{The Bloch sphere annotated with the 3 principal axes, the basis rotations and Hadamards (yellow) which act as basis transformations. Note that the X, Y, Z axes are arranged in a right-handed manner. (Diagram by Stefano Gogioso)}
\end{figure}

\begin{align*}
	\text{Right Handed } \quad H_{XY}&=\input{yrules/yx.tikz} \\
				   H_{YZ}&=\input{yrules/yz.tikz} \\\\
	\text{Left Handed } \quad  H_{YX}&=\input{yrules/xy.tikz} \\
				   H_{ZY}&=\input{yrules/zy.tikz} 
\end{align*}

\begin{figure}
\begin{center}\includegraphics{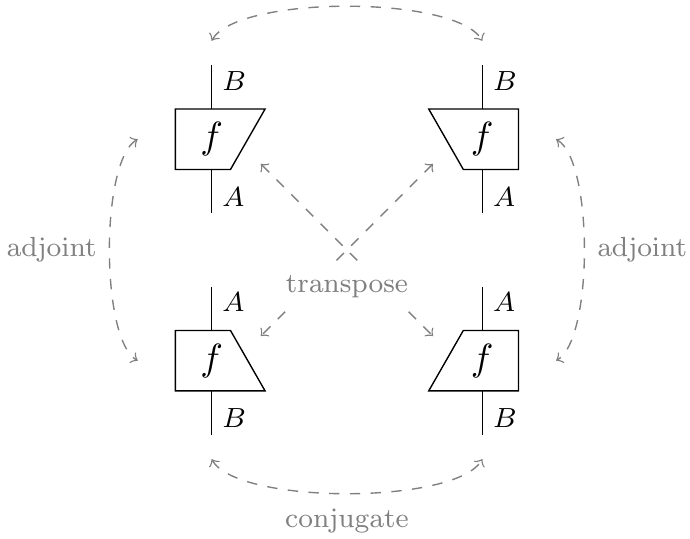}\end{center}
\caption{The geometric transforms to obtain the adjoint, transpose and conjugate of a linear map. Diagram taken from \textit{Picturing Quantum Processes} \autocite{coecke_kissinger}.}
\end{figure}

One can deduce the name of the gate by reading the colours of the gate from top to bottom. 
Due to the vertical line of symmetry, these gates remain unchanged when reflected vertically, 
the geometric operation for the adjoint, which reflects their Hermitian properties. 
Furthermore, the chirality is switched when the gate is reflected horizontally, which is 
the geometric operation for the transpose.
\begin{align*}
    \input{yrules/zy_yank1.tikz} &= \input{yrules/zy_yank2.tikz} 
                              &\input{yrules/yz_yank1.tikz}
                              &= \input{yrules/yz_yank2.tikz} \\
    \input{yrules/xy_yank1.tikz} &= \input{yrules/yx_yank2.tikz} 
                              &\input{yrules/yx_yank1.tikz} 
                              &= \input{yrules/xy_yank2.tikz} 
\end{align*}
It is also worth pointing out that the cups and caps of the Y basis are different from the usual cups and caps 
used in ZX calculus. In ZX calculus, the Z and X cups (and caps) coincide so it is acceptable to omit the spider. 
Pushing a gate through a cup (or a cap) applies the transpose in that basis to the gate, on the other qubit.
The fact that the Y cups (and caps) are different from the Z/X cups (and caps) tells us that the transpose is a basis-dependent operation.
Unless specified, the cups and caps are the usual ones in the Z and X basis.
$$ \left\llbracket\input{yrules/zcup.tikz}\right\rrbracket = \ket{00} + \ket{11} $$
$$ \left\llbracket\input{yrules/xcup.tikz}\right\rrbracket = |{+}{+}\rangle + |{-}{-}\rangle $$
$$ \input{yrules/zcup.tikz} = \input{yrules/cup.tikz} = \input{yrules/xcup.tikz} $$
$$ \input{yrules/ycup.tikz} \neq \input{yrules/cup.tikz} $$

The Y spider can be decomposed using the new Hadamard gates:
$$ \input{yrules/y_spider1.tikz} = \input{yrules/y_spider.tikz}  = \input{yrules/y_spider2.tikz} $$
and the Y rotation can be deomposed in the same way.
$$\input{yrules/ydef3.tikz} = \input{zx/ry.tikz} = \input{yrules/ydef4.tikz}$$
By Euler decomposition of the Y rotation, we have two ways of expressing the RY gate using RZ and RX gates. A proof of this can be found in the appendix.
$$\input{zx/ry_def1.tikz} = \input{zx/ry_def2.tikz}$$
The $\frac{\pi}{2}$ and $-\frac{\pi}{2}$ rotation gates are abbreviated to 
the ``$+$" and ``$-$" gates respectively.
$$\input{yrules/ydef1.tikz} \eq{y def} \input{zx/ry.tikz} \eq{y def} \input{yrules/ydef2.tikz}$$

By expanding the definition, the Y rotation interacts with the $\pm\frac{\pi}{2}$ Z and X 
rotation gates in the following way:
\begin{align*}
	\input{yrules/zpy1.tikz} &= \input{yrules/zpy2.tikz} &
	\input{yrules/xmy1.tikz} &= \input{yrules/xmy2.tikz} \\
	\input{yrules/zmy1.tikz} &= \input{yrules/zmy2.tikz} &
	\input{yrules/xpy1.tikz} &= \input{yrules/xpy2.tikz} \\
	\input{yrules/zpiy1.tikz} &= \input{yrules/zpiy2.tikz} &
	\input{yrules/xpiy1.tikz} &= \input{yrules/xpiy2.tikz}.
\end{align*}

Finally, we stress that the Hadamard gates are not rotations about the X,Y, or Z axes, so these two methods of decomposition for the Y rotation really are different. Furthermore, the Y spider \textbf{cannot} be decomposed into Z or X spiders using only the $+/-$ rotations in the general case. As a general rule, one can replace an equal number of input and output wires with rotation gates in place of Hadamard gates: this is done in section \ref{pauli_gadget}. In practice, it is most convenient to use the Y rotation for Y rotations and use the Hadamard gates for the more general Y spiders. 
\begin{align*}
\input{yrules/zp.tikz} &\neq \input{yrules/yx.tikz} &\input{yrules/xp.tikz} &\neq \input{yrules/yz.tikz} \\
\input{yrules/zm.tikz} &\neq \input{yrules/xy.tikz} &\input{yrules/xm.tikz} &\neq \input{yrules/zy.tikz} 
\end{align*}

\section{Diagrammatic Differentiation} \label{diadiff}
%[TODO intro this section]
%This section introduces a novel method to differentiate circuits diagrammatically. In particular, 
%this method can obtain the Hamiltonian of a system entirely diagrammatically.

\subsection{Stone's Theorem} \label{stones}
\textit{Stone's Theorem on one parameter unitary groups} \autocite{stones}: Let $\{ U(t) \}_{t\in\mathbb{R}}$ be a strongly continuous one-parameter unitary group. Then $\{ U(t) \}_{t\in\mathbb{R}}$ is generated by a Hermitian operator $H$: 
$$ U(t) = e^{itH} $$
and conversely every Hermitian matrix generates a unitary one parameter group. The matrix exponential is given by the Taylor series:

$$ e^X = 1 + X + \frac{X^2}{2} + \frac{X^3}{6} + \ldots $$

In other words, there is a one-to-one correspondence between the Hermitian operators and one parameter unitary groups. Given a quantum mechanical system $U(t)$, Stone's theorem tells us that the time evolution of the system over time is described entirely by its infinitesimal generator $H$, known as the Hamiltonian. Thus, a quantum mechanical system can be understood by study its Hamiltonian. It is a remarkable result that the generator of the unitary group is not parameterised. For example, the RZ gates form a one parameter group $\{R_Z(\theta) \}$ and is generated by the Pauli matrix $Z$: $R_Z(\theta) = \exp\left(i\frac{\theta}{2}Z\right)$. To verify this, one can differentiate the unitary map to obtain $iU(\theta)H$. 

\begin{align*}
	\frac{\mathrm{d}}{\mathrm{d}\theta}R_Z(\theta)  &= \frac{i}{2} e^{-i\theta/2} \ket{0}\bra{0} - \frac{i}{2} e^{i\theta/2} \ket{1}\bra{1}\\
						  &= -\frac{i}{2}(e^{-i\theta/2} \ket{0}\bra{0} + e^{i\theta/2} \ket{1}\bra{1})(\ket{0}\bra{0} - \ket{1}\bra{1})\\
						  &= -\frac{i}{2} R_Z(\theta) Z \\\\
	\frac{\mathrm{d}}{\mathrm{d}\theta}R_X(\theta) &= \frac{\mathrm{d}}{\mathrm{d}\theta}HR_Z(\theta)H    \\
						 &= -\frac{i}{2} HR_Z(\theta)ZH \\
						 &= -\frac{i}{2} HR_Z(\theta)HHZH \\
						 &= -\frac{i}{2}R_X(\theta)X \\\\
	\frac{\mathrm{d}}{\mathrm{d}\theta}R_Y(\theta) &= \frac{\mathrm{d}}{\mathrm{d}\theta}R_X(-\pi/2)R_Z(\theta)R_X(\pi/2)   \\
	 &= -\frac{i}{2} R_X(-\pi/2)R_Z(\theta)ZR_X(\pi/2)   \\
	&= -\frac{i}{2} R_X(-\pi/2)R_Z(\theta)R_X(\pi/2)R_X(-\pi/2)ZR_X(\pi/2)   \\
	&= -\frac{i}{2} R_Y(\theta)Y
\end{align*}

So indeed

$$
R_Z(\theta) = \exp\left(-i\frac{\theta}{2}Z\right) \quad
R_X(\theta) = \exp\left(-i\frac{\theta}{2}X\right) \quad
R_Y(\theta) = \exp\left(-i\frac{\theta}{2}Y\right).
$$

This proves that the Hamiltonian of the rotation gates are proportional to their respective Paulis. To analyse the derivatives of parameterised quantum diagrams, we add the following 3 rules to our calculus:
$$ \frac{\mathrm{d}}{\mathrm{d}\theta}\left[ \input{zx/rz.tikz} \right] \eq{d} \input{zx/dz.tikz} $$
$$ \frac{\mathrm{d}}{\mathrm{d}\theta}\left[ \input{zx/rx.tikz} \right] \eq{d} \input{zx/dx.tikz} $$
$$ \frac{\mathrm{d}}{\mathrm{d}\theta}\left[ \input{zx/ry.tikz} \right] \eq{d} \input{zx/dy.tikz} $$

For convenience, we drop the $\mathrm{d}/\mathrm{d}\theta$ and only use the square brackets to denote the derivative of a diagram when the variable is unambiguous. To differentiate more complex diagrams, three more rules must be introduced: the chain rule, the product rule, and the linearity of derivatives. The linearity rule uses the observation that 

$$\frac{\mathrm{d} MU(\theta)N}{\mathrm{d} \theta} = M\frac{\mathrm{d} U(\theta)}{\mathrm{d} \theta}N$$
\begin{align*}
       &\input{deriv/linear1.tikz} \\
    \eq{lin}\quad &\input{deriv/linear2.tikz}.
\end{align*}

The product rule is used when a circuit has two gates that depend on the same parameter.
$$\frac{\mathrm{d} U_1(\theta)U_2(\theta)}{\mathrm{d} \theta} = \frac{\mathrm{d} U_1(\theta)}{\mathrm{d} \theta}U_2(\theta) + U_1(\theta)\frac{\mathrm{d} U_2(\theta)}{\mathrm{d} \theta}$$
       $$\qquad\input{deriv/product1.tikz}$$
       $$\eq{pr}\quad \input{deriv/product2.tikz} + \input{deriv/product3.tikz}$$

Finally, we also have the chain rule, which is used when the parameter of the gate is a function of the 
variable we are differentiating with respect to.
$$\frac{\mathrm{d} U(f(\theta))}{\mathrm{d} \theta} = \frac{\mathrm{d} U(f(\theta))}{\mathrm{d} f(\theta)}\frac{\mathrm{d} f(\theta)}{\mathrm{d} \theta}$$
$$\input{deriv/chain1.tikz} = \input{deriv/chain2.tikz}$$

As an example, the Z and X spiders can be differentiated as follows:
$$
\input{deriv/x11.tikz} \eq{f, lin} 
\input{deriv/x12.tikz} \eq{d}
\input{deriv/x13.tikz} \eq{f}
\input{deriv/x14.tikz}
$$
$$
\input{deriv/zn1.tikz} \eq{f, lin} 
\input{deriv/zn2.tikz} \eq{d}
\input{deriv/zn4.tikz} \eq{f}
\input{deriv/zn5.tikz}
$$

\subsection{Computing the Hamiltonian Diagrammatically}

These rules for differentiation can be used to find the Hamiltonian suggested by Stone's theorem. Since any element in a one parameter unitary group can written as $e^{itV}$, its derivative should be $iVe^{itV} = ie^{itV}V$. If a ZX circuit does indeed form a one parameter unitary group, then its derivative should be the original circuit composed with the Hamiltonian. For example, the tensor product of the Z rotation and the X rotation groups is also a one parameter unitary group, so we can find its Hamiltonian by differentiating with respect to $\theta$.
\begin{align*}
\input{stones/g1.tikz} &\eq{pr, d} \input{stones/g2.tikz} + \input{stones/g3.tikz} \\
                    &= \left(\input{stones/g4.tikz} + \input{stones/g5.tikz} \right) \input{stones/g6.tikz}
\end{align*}
We conclude that:
$$ \input{stones/g6.tikz} =  \exp\left(-\frac{i}{2}\theta\left(\input{stones/g4.tikz} + \input{stones/g5.tikz} \right)\right) $$
where the $\frac{i}{2}$ factor is reintroduced after it was discarded from the 
graphical calculus in the $\mathsf{(d)}$ rule. Although we generally work in a scalar-free ZX calculus, when working with matrix exponentials the scalars must be taken into account.
\pagebreak
\section{Variational Quantum Circuits} \label{qnndiff}
The idea of applying quantum computers to compute neural-network-like structures was first introduced in 1994 by Lewenstein \autocite{qnn}, with the quantum perceptron. \textit{Variational Quantum Circuits} (VQC), also known as \textit{Parameterised Quantum Circuits} (PQC), are conceptually similar to a traditional neural network, with tunable parameters arranged in layers. However, an $n$-qubit quantum neural network represents a  $2^n\times 2^n$ unitary mapping, so the number of input qubits is equal to the number of output qubits, and there are no non-linear activation gates. Instead, small local gates are used as building blocks of the overall unitary. Typically, CNOT, CZ and rotation gates are used, as they can be simulated by a quantum computer relatively efficiently. The phases of the rotation gates are determined by the parameters of the circuit. 

To match the convention used by Coecke and Kissinger in \textit{Picturing Quantum Processes} \autocite{coecke_kissinger}, this subsection displays its diagrams vertically from bottom to top.
\ctikzfig{qml/pqc}
$$ U(\bm \theta) = \prod_{i=1}^n U_i(\theta_i) $$

In the diagram above, we give a compact representation of the circuit where each horizontal strip of blocks consists only of one-parameter gates $U_k(\theta_k)$, and parameter-free gates $U$. We preclude the possibility that the circuit uses a parameter more than once.
The big matrix product operator is non-commutative, applying the unitaries from right-to-left in the equations and bottom-to-top in the diagram.

Typically the inputs are encoded as pure states or parameters of rotation gates. 
Both methods correspond to a layer of X spiders followed by a layer of RX gates. The input parameters are not tuned during training and can be absorbed into $U_1(\theta_1)$.
\[\input{qml/prepare1.tikz} \quad = \quad \input{qml/prepare2.tikz} \]
\[\input{qml/prepare3.tikz} \quad = \quad \input{qml/prepare4.tikz} \quad = \quad \input{qml/prepare5.tikz} \]

The output of the PQC is an entangled quantum state, which collapses to a particular value when measured. Measurable physical quantities such as momentum, spin and angular momentum can be associated with a Hermitian operator $H$, known as an observable. The expected value of the physical quantity's measurement $H$ over the mixed state $U(\theta)\ket{0}$ can be calculated as

\ctikzfig{qml/double_form_straight}
$$
E(\boldsymbol{\theta})=\left\langle 0\left|U(\boldsymbol{\theta})^{\dagger} H U(\boldsymbol{\theta})\right| 0\right\rangle
$$

Since the parameter $\theta_i$ appears twice in the expectation, the partial derivative of the expected value $\frac{\partial E}{\partial \theta_i}$ is calculated using the product rule:

\ctikzfig{qml/top}
$$\scalebox{1.5}{$i$}\left(\input{qml/normal_bracket.tikz}\right)$$
\ctikzfig{qml/bottom}

The two terms can be grouped together using the commutator $[\cdot, \cdot]$ which is defined as 
$ [A,B] = AB - BA $.
\ctikzfig{qml/top}
$$\scalebox{1.5}{$i$}\left[\input{qml/lie_bracket.tikz}\right]$$
\ctikzfig{qml/bottom}
$$
\frac{\partial E(\boldsymbol{\theta})}{\partial \theta_{k}}=i\left\langle 0\left|U_{-}^{\dagger}\left[V_{k}\otimes I, U_{+}^{\dagger} H U_{+}\right] U_{-}\right| 0\right\rangle
$$
\[ U_{-} = \prod_{i=1}^{k}U_i(\theta_i) \qquad U_{+} = \prod_{i=k+1}^nU_i(\theta_i)\] 
These gradients can then be used to perform backpropagation. 

\section{Quantum Machine Learning}

We can now to apply quantum machine learning to solve 
real-world problems. As described in the previous section, the learning procedure using a parameterised quantum circuit can be summarised as follows:

\begin{enumerate}
	\item The input $\bm x$ of the data $(\bm x, y)$ is encoded into the input state of the network either as a product state $\ket{\bm x}$ or as $\ket{0}$ followed by a set of rotations $\Motimes_{i=1}^nR(x_i)$.
	\item The parameterised unitary map $U(\bm \theta)$ is applied to the input state, producing a new, entangled state.
	\item The entangled state is observed using a measurement operator. The result of the measurement is either used as the predicted label $\hat{y}$, or later used to produce the predicted label using classical computation.
\end{enumerate}

These three steps describe how a parameterised quantum circuit can be used to make predictions. To improve upon these predictions, gradient-based methods are deployed to minimise a loss function, just like in classical machine learning.

\begin{enumerate}
	\setcounter{enumi}{3}
\item The predicted label $\hat{y}$ and the true label $y$ are used to compute a loss function via classical computation, e.g.\  the squared loss function $ \mathscr{L}(\hat{y}) = (y - \hat{y})^2$. 
	\item Backpropagation is used to obtain $\frac{\partial \mathscr{L}}{\partial \theta_i}$.
	\item $\frac{\partial \mathscr{L}}{\partial \theta_i}$ is used to perform gradient descent:
		$$ \bm\theta_{k+1} = \bm\theta_k - \gamma \left.\frac{\partial \mathscr{L}}{\partial \bm\theta}\right|_{\bm\theta=\bm\theta_k}.$$
\item Repeat this procedure until $\bm \theta$ converges.
\end{enumerate}

As seen above, the only difference between quantum machine learning and classical machine learning is the manner in which the circuit and its gradients are computed; both of these computations can be managed by the PennyLane library.

\section{PennyLane} \label{pennylane}

PennyLane \autocite{pennylane} is a Python library for quantum machine learning and quantum-classical operations. Similar to PyTorch, TensorFlow and other machine learning libraries, PennyLane performs automatic differentiation of circuits by building a computation graph and computing the overall gradient using the chain rule. Furthermore, it provides plugins to either simulate a quantum device, or directly interface with actual quantum devices.

Due to the inherent non-determinism of the PQC's output, a key difference between PennyLane and classical machine learning frameworks is that the expected value of the output is used instead; this expected value is estimated by averaging over multiple measurements. It is important to point out that the expected value of the estimator may not ever be the circuit output.
%For example, the expected value of observing the state $\ket{00} + \ket{10}$ in the Z basis is $\frac{00+10}{2} = 01$, but it is clear that the probability of the state collapsing to $\ket{01}$ is 0.
\subsection{Gradient Recipes}

The symbolic differentiation of a circuit with respect to one parameter has been shown in \ref{qnndiff}. 
However, this expression cannot be directly computed efficiently on a classical computer for the same reason 
the original circuit cannot be classically simulated (unless $\mathsf{BQP}=\mathsf{BPP}$). One could try to estimate the 
derivative by performing numerical differentiation using quantum devices

$$\frac{\partial f(x)}{\partial \theta_i} \approx \frac{f(x+h) - f(x-h)}{2h}$$

but choosing the most numerically satisfactory $h$ is difficult and the non-determinism of the circuit further complicates the estimation. 
Fortunately, there is a better way. The exact, analytical derivative of a quantum gate $f(\bm{\mu})$ is given by 

$$
\frac{\partial f(x)}{\partial\theta_i}=c(f(x+s)-f(x-s))
$$

where $c, s \in \mathbb{R}$ differ for each gate and depend on the eigenvalues of the quantum gate's Hamiltonian \autocite{recipes}. 
To emphasise, this is \textbf{not} a numerical approximation, but an algebraic trick to rewrite the derivative analytically. 
Therefore the accuracy of the gradient estimate solely depends on our estimates of $f(\mu+s)$ and $f(\mu-s)$, which 
improve from repeated sampling the circuit by the law of large numbers. Alternatively, one could perform ``doubly stochastic gradient" descent by averaging the gradient over different data samples \autocite{sweke2019stochastic}.

%By the law of large numbers, the estimates for $f(\mu \pm s)$ improve as the number of samples increase and the rate of convergence 
%with respect to the number of samples is given by Chebyshev's inequality.
%[TODO Is this relevant?]
%Schuld et al.\ \autocite{hilbert} presented how a parameterised quantum network can be used
%as a kernel to a high dimensional Hilbert space for a support vector machine (SVM). Even though 
%parameterised quantum networks represent $2^n\times 2^n$ unitary linear maps, it is later embedded in a $n$ dimensional space with a non-linear transformation.

\section{Barren Plateaus in VQCs} \label{barren}
It is a mystery why gradient descent empirically performs well on deep neural networks from a learning theory perspective. 
While it may be visually comforting to think about descending towards the bottom of a 3 dimensional convex bowl, deep neural 
networks generally have a non-convex landscape so travelling in the direction of steepest descent may take one to a saddle point 
or a local minima instead. Chroromanska et al.\ \autocite{loss} conjectured that the unexpected success of applying gradient descent to deep learning is due to 
the geometry of the loss \textit{landscape}:

\begin{displayquote}
    ``... We show that for large-sized decoupled networks the lowest critical values of the random loss function form
a layered structure and they are located in a well-defined band lower-bounded by the
global minimum. The number of local minima outside that band diminishes exponentially with the size of the network. 
... We conjecture that both simulated annealing and stochastic gradient descent converge to the band of low critical points, and that all
critical points found there are local minima of high quality measured by the test error."
\end{displayquote}

In short, there are very few bad local minimas. This, combined with an argument with the Hessian matrix that most critical points 
are saddles rather than maximas and minimas could be the explanation to why gradient descent seems to yield good results even 
though the global minimum is not always reached.

The application of variational quantum circuits to machine learning via gradient descent gives rise to questions of a 
similar nature: 
\begin{center}
	\textit{What do the landscapes of VQCs look like? \\Which architectures of VQCs, if any, provide `nice' 
landscapes for learning?} 
\end{center}
Although random circuits are often proposed as initial guesses
for exploring the space of quantum states, McClean et al.\ \autocite{barren} showed in 2018 that wiring the circuits randomly 
does not yield good landscapes: 

\begin{displayquote}
	``Specifically, we show that for a wide class of reasonable
parameterized quantum circuits, the probability that the gradient along any reasonable direction is
non-zero to some fixed precision is exponentially small as a function of the number of qubits."
\end{displayquote}

where the ``wide class of reasonable parameterized quantum circuits" is a technical formalisation of 
``almost uniformly sampled unitaries". The geometric interpretation of this result is that the landscapes that result 
from randomly wired quantum circuits are wide and flat with gradients numerically close to zero, making gradient descent 
computationally infeasible or even numerically impossible. This means it is necessary to find a better way to design the quantum circuits 
used in machine learning. Besides analysing specific structures in (Section \ref{init_strat}), we also discuss a general initialisation strategy 
proposed by Grant et al.\ \autocite{init} and provide a diagrammatic interpretation of it.

\section{Types of Ans\"{a}tze} \label{PQC architectures}
Since parameterised quantum circuits cannot be wired randomly, the quantum machine learning community has proposed other structures that can be used instead, known as \textit{ans\"{a}tz}. The term ``ans\"{a}tz" comes from physics, where variational methods are used to find approximate solutions to the lowest energy ground state of a quantum system. In both situations, the quality of the approximation depends heavily on the choice of ans\"{a}tze. Generally, circuit ans\"{a}tze fall under three basic classes: \textit{layered gate ans\"{a}tze}, \textit{alternating operator ans\"{a}tze} \autocite{hadfield2019quantum}, and \textit{tensor network ans\"{a}tze} \autocite{huggins2019towards, expressive}
\begin{center}
\begin{tabular}{cc}
	\input{ansatz/aoa.tikz} & \input{ansatz/tna.tikz} \\
	Alternating Operator Ansatz & Tensor Network Ansatz
\end{tabular}
\end{center}
Alternating operator ans\"{a}tz can be thought of as a special case of layered gate ans\"{a}tz. 
Furthermore, layered gate ans\"{a}tz have theoretically been shown to be more expressive than tensor network ans\"{a}tz. \autocite{expressive} Therefore we will direct our attention to analysing circuit ans\"{a}tz with repeating blocks.

%\subsection{Usage}
%\begin{enumerate}
%    \item Direct use for hybrid classical-quantum computation
%    \item Blackbox Kernel
%\end{enumerate}

%\subsection{BPP vs BQP}
%[Talk about https://dl.acm.org/doi/10.1145/167088.167097, Bernstein and Varziani 1993]

\section{Converting Ans\"{a}tze to ZX Calculus} \label{conv}
Now we can convert ans\"{a}tz from traditional circuit notation to ZX-calculus:

\vspace{0.5cm}
\begin{center}
\includegraphics[scale=1]{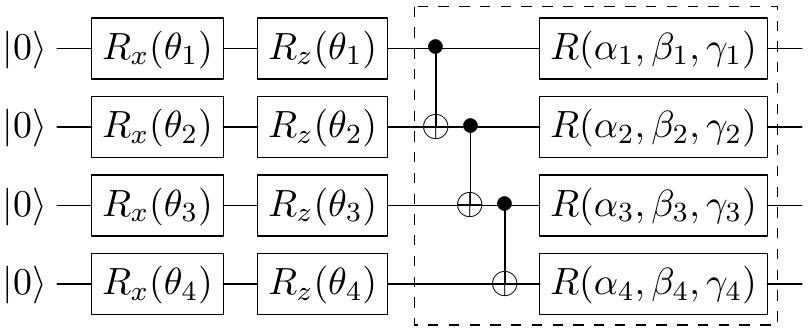} \\
From \textit{S YC Chen et al.\ } \autocite{reinforcement}, used for Deep Reinforcement Learning. \\
\vspace{3mm}

\input{ansatz/stairs.tikz}

Ans\"{a}tz converted to ZX Calculus.

\input{ansatz/stairsy.tikz}

Ans\"{a}tz converted to ZX Calculus + Y Spiders.
\end{center}
\begin{center}
\includegraphics[scale=1]{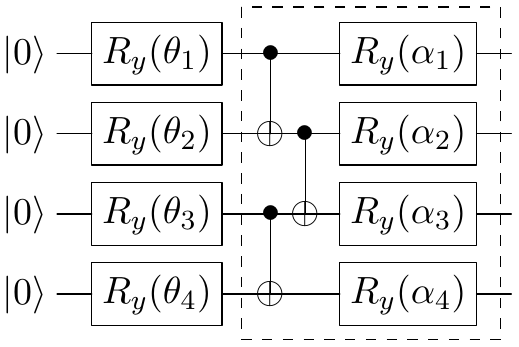} \\
From \textit{A Mari et al.\ } \autocite{transfer}, used for Hybrid Classical-Quantum Transfer Learning. \\
\vspace{3mm}

\input{ansatz/bricks.tikz}

Ans\"{a}tz converted to ZX Calculus.

\input{ansatz/bricksy.tikz}

Ans\"{a}tz converted to ZX Calculus + Y Spiders.
\end{center}

We refer to the former CNOT layout as ``staircase'' and the latter one as ``brick wall''. Chapter 4 contains more ans\"{a}tze conversions and simplifications.

%% file: figures/zx/rz.tikz
\begin{tikzpicture}[scale=0.4, baseline={([yshift=-0.5ex]current bounding box.center)}]
	\begin{pgfonlayer}{nodelayer}
		\node [style=none] (2) at (0, 0) {};
		\node [style=none] (10) at (3, 0) {};
		\node [style={Z_med}] (11) at (1.5, 0) {$\alpha$};
	\end{pgfonlayer}
	\begin{pgfonlayer}{edgelayer}
		\draw (2.center) to (11);
		\draw (11) to (10.center);
	\end{pgfonlayer}
\end{tikzpicture}

%% file: figures/zx/rx.tikz
\begin{tikzpicture}[scale=0.4, baseline={([yshift=-0.5ex]current bounding box.center)}]
	\begin{pgfonlayer}{nodelayer}
		\node [style=none] (2) at (0, 0) {};
		\node [style=none] (10) at (3, 0) {};
		\node [style={X_med}] (11) at (1.5, 0) {$\alpha$};
	\end{pgfonlayer}
	\begin{pgfonlayer}{edgelayer}
		\draw (2.center) to (11);
		\draw (11) to (10.center);
	\end{pgfonlayer}
\end{tikzpicture}

%% file: figures/zx/pz.tikz
\begin{tikzpicture}[scale=0.4, baseline={([yshift=-0.6ex]current bounding box.center)}]
	\begin{pgfonlayer}{nodelayer}
		\node [style=none] (2) at (0, 0) {};
		\node [style=none] (10) at (3, 0) {};
		\node [style={Z_med}] (11) at (1.5, 0) {$\pi$};
	\end{pgfonlayer}
	\begin{pgfonlayer}{edgelayer}
		\draw (2.center) to (11);
		\draw (11) to (10.center);
	\end{pgfonlayer}
\end{tikzpicture}

%% file: figures/zx/px.tikz
\begin{tikzpicture}[scale=0.4, baseline={([yshift=-0.6ex]current bounding box.center)}]
	\begin{pgfonlayer}{nodelayer}
		\node [style=none] (2) at (0, 0) {};
		\node [style=none] (10) at (3, 0) {};
		\node [style={X_med}] (11) at (1.5, 0) {$\pi$};
	\end{pgfonlayer}
	\begin{pgfonlayer}{edgelayer}
		\draw (2.center) to (11);
		\draw (11) to (10.center);
	\end{pgfonlayer}
\end{tikzpicture}

%% file: figures/zx/check/cx.tikz
\begin{tikzpicture}[scale=0.6, baseline={([yshift=-0.5ex]current bounding box.center)}]
	\begin{pgfonlayer}{nodelayer}
		\node [style=none] (2) at (-2, 0) {};
		\node [style=none] (5) at (2, 0.5) {};
		\node [style=X] (6) at (1, 0.5) {};
		\node [style=none] (7) at (0.5, 0) {};
		\node [style=none] (8) at (0.5, 1) {};
		\node [style=none] (10) at (2, 2) {};
		\node [style=none] (11) at (-2, 1.5) {};
		\node [style=Z] (12) at (-1, 1.5) {};
		\node [style=none] (13) at (-0.5, 1) {};
		\node [style=none] (14) at (-0.5, 2) {};
	\end{pgfonlayer}
	\begin{pgfonlayer}{edgelayer}
		\draw (10.center)
			 to (14.center)
			 to [bend right=45] (12.center)
			 to [bend right=45] (13.center)
			 to (8.center)
			 to [bend left=45] (6.center)
			 to [bend left=45] (7.center)
			 to (2.center);
		\draw (6) to (5.center);
		\draw (12) to (11.center);
	\end{pgfonlayer}
\end{tikzpicture}

%% file: figures/zx/check/cx_1.tikz
\begin{tikzpicture}[scale=0.6, baseline={([yshift=-0.5ex]current bounding box.center)}]
	\begin{pgfonlayer}{nodelayer}
		\node [style=none] (8) at (1, -0.5) {};
		\node [style=none] (10) at (1, 0.5) {};
		\node [style=none] (11) at (-1, 0) {};
		\node [style=Z] (12) at (0, 0) {};
		\node [style=none] (13) at (0.5, -0.5) {};
		\node [style=none] (14) at (0.5, 0.5) {};
		\node [style=none] (15) at (-1, -1.5) {};
		\node [style=none] (16) at (1, -1.5) {};
	\end{pgfonlayer}
	\begin{pgfonlayer}{edgelayer}
		\draw (8.center)
			 to (13.center)
			 to [bend left=45] (12.center)
			 to [bend left=45] (14.center)
			 to (10.center);
		\draw (12) to (11.center);
		\draw (16.center) to (15.center);
	\end{pgfonlayer}
\end{tikzpicture}

%% file: figures/zx/check/cx_2.tikz
\begin{tikzpicture}[scale=0.6, baseline={([yshift=-0.5ex]current bounding box.center)}]
	\begin{pgfonlayer}{nodelayer}
		\node [style=none] (8) at (-1.5, -0.5) {};
		\node [style=none] (10) at (-1.5, 0.5) {};
		\node [style=none] (11) at (1, 0) {};
		\node [style=X] (12) at (0, 0) {};
		\node [style=none] (13) at (-0.5, -0.5) {};
		\node [style=none] (14) at (-0.5, 0.5) {};
		\node [style=none] (15) at (-1.5, 1.5) {};
		\node [style=none] (16) at (1, 1.5) {};
	\end{pgfonlayer}
	\begin{pgfonlayer}{edgelayer}
		\draw (10.center)
			 to (14.center)
			 to [bend left=45] (12.center)
			 to [bend left=45] (13.center)
			 to (8.center);
		\draw (12) to (11.center);
		\draw (15.center) to (16.center);
	\end{pgfonlayer}
\end{tikzpicture}

%% file: figures/zx/check/cz.tikz
\begin{tikzpicture}[scale=0.6, baseline={([yshift=-0.5ex]current bounding box.center)}]
	\begin{pgfonlayer}{nodelayer}
		\node [style=none] (2) at (-2, 0) {};
		\node [style=none] (5) at (2, 0.5) {};
		\node [style=Z] (6) at (1, 0.5) {};
		\node [style=none] (7) at (0.5, 0) {};
		\node [style=none] (8) at (0.5, 1) {};
		\node [style=none] (10) at (2, 2) {};
		\node [style=none] (11) at (-2, 1.5) {};
		\node [style=Z] (12) at (-1, 1.5) {};
		\node [style=none] (13) at (-0.5, 1) {};
		\node [style=none] (14) at (-0.5, 2) {};
		\node [style=H] (15) at (0, 1) {};
	\end{pgfonlayer}
	\begin{pgfonlayer}{edgelayer}
		\draw (2.center)
			 to (7.center)
			 to [bend right=45] (6.center)
			 to [bend right=45] (8.center)
			 to (13.center)
			 to [bend left=45] (12.center)
			 to [bend left=45] (14.center)
			 to (10.center);
		\draw (6) to (5.center);
		\draw (12) to (11.center);
	\end{pgfonlayer}
\end{tikzpicture}

%% file: figures/zx/check/cz_1.tikz
\begin{tikzpicture}[scale=0.6, baseline={([yshift=-0.5ex]current bounding box.center)}]
	\begin{pgfonlayer}{nodelayer}
		\node [style=none] (8) at (-1.5, -0.5) {};
		\node [style=none] (10) at (-1.5, 0.5) {};
		\node [style=none] (11) at (1, 0) {};
		\node [style=Z] (12) at (0, 0) {};
		\node [style=none] (13) at (-0.5, -0.5) {};
		\node [style=none] (14) at (-0.5, 0.5) {};
		\node [style=none] (15) at (-1.5, 1.5) {};
		\node [style=none] (16) at (1, 1.5) {};
		\node [style=H] (17) at (-1, 0.5) {};
	\end{pgfonlayer}
	\begin{pgfonlayer}{edgelayer}
		\draw (10.center)
			 to (14.center)
			 to [bend left=45] (12.center)
			 to [bend left=45] (13.center)
			 to (8.center);
		\draw (12) to (11.center);
		\draw (15.center) to (16.center);
	\end{pgfonlayer}
\end{tikzpicture}

%% file: figures/zx/h.tikz
\begin{tikzpicture}[scale=0.4, baseline={([yshift=-0.5ex]current bounding box.center)}]
	\begin{pgfonlayer}{nodelayer}
		\node [style=H] (0) at (0, 0) {};
		\node [style=none] (1) at (-1, 0) {};
		\node [style=none] (2) at (1, 0) {};
	\end{pgfonlayer}
	\begin{pgfonlayer}{edgelayer}
		\draw (1.center) to (0);
		\draw (0) to (2.center);
	\end{pgfonlayer}
\end{tikzpicture}

%% file: figures/zx/X0.tikz
\begin{tikzpicture}[baseline={([yshift=-0.5ex]current bounding box.center)}]
	\begin{pgfonlayer}{nodelayer}
		\node [style={X_tri_l}, font={\scriptsize}] (0) at (0, 0) {$0$};
		\node [style=none] (1) at (1, 0) {};
	\end{pgfonlayer}
	\begin{pgfonlayer}{edgelayer}
		\draw (0) to (1.center);
	\end{pgfonlayer}
\end{tikzpicture}

%% file: figures/zx/X0S.tikz
\begin{tikzpicture}[baseline={([yshift=-0.5ex]current bounding box.center)}]
	\begin{pgfonlayer}{nodelayer}
		\node [style=Z, font={\scriptsize}] (0) at (0, 0) {};
		\node [style=none] (1) at (1, 0) {};
	\end{pgfonlayer}
	\begin{pgfonlayer}{edgelayer}
		\draw (0) to (1.center);
	\end{pgfonlayer}
\end{tikzpicture}

%% file: figures/zx/Z0.tikz
\begin{tikzpicture}[baseline={([yshift=-0.5ex]current bounding box.center)}]
	\begin{pgfonlayer}{nodelayer}
		\node [style={Z_tri_l}, font={\scriptsize}] (0) at (0, 0) {$0$};
		\node [style=none] (1) at (1, 0) {};
	\end{pgfonlayer}
	\begin{pgfonlayer}{edgelayer}
		\draw (0) to (1.center);
	\end{pgfonlayer}
\end{tikzpicture}

%% file: figures/zx/Z0S.tikz
\begin{tikzpicture}[baseline={([yshift=-0.5ex]current bounding box.center)}]
	\begin{pgfonlayer}{nodelayer}
		\node [style=X, font={\scriptsize}] (0) at (0, 0) {};
		\node [style=none] (1) at (1, 0) {};
	\end{pgfonlayer}
	\begin{pgfonlayer}{edgelayer}
		\draw (0) to (1.center);
	\end{pgfonlayer}
\end{tikzpicture}

%% file: figures/zx/X1.tikz
\begin{tikzpicture}[baseline={([yshift=-0.5ex]current bounding box.center)}]
	\begin{pgfonlayer}{nodelayer}
		\node [style={X_tri_l}, font={\scriptsize}] (0) at (0, 0) {$1$};
		\node [style=none] (1) at (1, 0) {};
	\end{pgfonlayer}
	\begin{pgfonlayer}{edgelayer}
		\draw (0) to (1.center);
	\end{pgfonlayer}
\end{tikzpicture}

%% file: figures/zx/X1S.tikz
\begin{tikzpicture}[baseline={([yshift=-0.5ex]current bounding box.center)}]
	\begin{pgfonlayer}{nodelayer}
		\node [style={Z_med}, font={\scriptsize}] (0) at (0, 0) {$\pi$};
		\node [style=none] (1) at (1, 0) {};
	\end{pgfonlayer}
	\begin{pgfonlayer}{edgelayer}
		\draw (0) to (1.center);
	\end{pgfonlayer}
\end{tikzpicture}

%% file: figures/zx/Z1.tikz
\begin{tikzpicture}[baseline={([yshift=-0.5ex]current bounding box.center)}]
	\begin{pgfonlayer}{nodelayer}
		\node [style={Z_tri_l}, font={\scriptsize}] (0) at (0, 0) {$1$};
		\node [style=none] (1) at (1, 0) {};
	\end{pgfonlayer}
	\begin{pgfonlayer}{edgelayer}
		\draw (0) to (1.center);
	\end{pgfonlayer}
\end{tikzpicture}

%% file: figures/zx/Z1S.tikz
\begin{tikzpicture}[baseline={([yshift=-0.5ex]current bounding box.center)}]
	\begin{pgfonlayer}{nodelayer}
		\node [style={X_med}, font={\scriptsize}] (0) at (0, 0) {$\pi$};
		\node [style=none] (1) at (1, 0) {};
	\end{pgfonlayer}
	\begin{pgfonlayer}{edgelayer}
		\draw (0) to (1.center);
	\end{pgfonlayer}
\end{tikzpicture}

%% file: figures/zx/rules/i11.tikz
\begin{tikzpicture}[scale=0.4, baseline={([yshift=-0.5ex]current bounding box.center)}]
	\begin{pgfonlayer}{nodelayer}
		\node [style=none] (0) at (0, 0) {};
		\node [style=none] (1) at (2.5, 0) {};
		\node [style=Z] (2) at (1.25, 0) {};
	\end{pgfonlayer}
	\begin{pgfonlayer}{edgelayer}
		\draw (0.center) to (2);
		\draw (2) to (1.center);
	\end{pgfonlayer}
\end{tikzpicture}

%% file: figures/zx/rules/i12.tikz
\begin{tikzpicture}[scale=0.4, baseline={([yshift=-0.5ex]current bounding box.center)}]
	\begin{pgfonlayer}{nodelayer}
		\node [style=none] (0) at (0, 0) {};
		\node [style=none] (1) at (2.5, 0) {};
		\node [style=none] (2) at (1.25, 0) {};
	\end{pgfonlayer}
	\begin{pgfonlayer}{edgelayer}
		\draw (0.center) to (2.center);
		\draw (2.center) to (1.center);
	\end{pgfonlayer}
\end{tikzpicture}

%% file: figures/zx/rules/i22.tikz
\begin{tikzpicture}[scale=0.4, baseline={([yshift=-0.5ex]current bounding box.center)}]
	\begin{pgfonlayer}{nodelayer}
		\node [style=none] (0) at (0, 0) {};
		\node [style=none] (1) at (2.5, 0) {};
		\node [style=none] (2) at (1.25, 0) {};
	\end{pgfonlayer}
	\begin{pgfonlayer}{edgelayer}
		\draw (0.center) to (2.center);
		\draw (2.center) to (1.center);
	\end{pgfonlayer}
\end{tikzpicture}

%% file: figures/zx/xor3.tikz
\begin{tikzpicture}[scale=0.8, baseline={([yshift=-0.5ex]current bounding box.center)}]
	\begin{pgfonlayer}{nodelayer}
		\node [style=none] (12) at (8, 0) {};
		\node [style={X_med}] (15) at (7, 0) {$\sigma\pi$};
	\end{pgfonlayer}
	\begin{pgfonlayer}{edgelayer}
		\draw (15) to (12.center);
	\end{pgfonlayer}
\end{tikzpicture}

%% file: figures/zx/xor4.tikz
\begin{tikzpicture}[scale=0.8, baseline={([yshift=-0.5ex]current bounding box.center)}]
	\begin{pgfonlayer}{nodelayer}
		\node [style=none] (12) at (8, 0) {};
		\node [style={Z_tri_l}] (15) at (7, 0) {$\sigma$};
	\end{pgfonlayer}
	\begin{pgfonlayer}{edgelayer}
		\draw (15) to (12.center);
	\end{pgfonlayer}
\end{tikzpicture}

%% file: figures/yrules/yx.tikz
\begin{tikzpicture}[scale=0.4, baseline={([yshift=-0.5ex]current bounding box.center)}]
	\begin{pgfonlayer}{nodelayer}
		\node [style=none] (4) at (0, 0) {};
		\node [style=none] (6) at (3, 0) {};
		\node [style=YX, rotate=90] (8) at (1.5, 0) {\YX};
	\end{pgfonlayer}
	\begin{pgfonlayer}{edgelayer}
		\draw (4.center) to (8);
		\draw (8) to (6.center);
	\end{pgfonlayer}
\end{tikzpicture}

%% file: figures/yrules/yz.tikz
\begin{tikzpicture}[scale=0.4, baseline={([yshift=-0.5ex]current bounding box.center)}]
	\begin{pgfonlayer}{nodelayer}
		\node [style=none] (4) at (0, 0) {};
		\node [style=none] (6) at (3, 0) {};
		\node [style=YZ, rotate=270] (8) at (1.5, 0) {\YZ};
	\end{pgfonlayer}
	\begin{pgfonlayer}{edgelayer}
		\draw (4.center) to (8);
		\draw (8) to (6.center);
	\end{pgfonlayer}
\end{tikzpicture}

%% file: figures/yrules/xy.tikz
\begin{tikzpicture}[scale=0.4, baseline={([yshift=-0.5ex]current bounding box.center)}]
	\begin{pgfonlayer}{nodelayer}
		\node [style=none] (4) at (0, 0) {};
		\node [style=none] (6) at (3, 0) {};
		\node [style=YX, rotate=270] (8) at (1.5, 0) {\YX};
	\end{pgfonlayer}
	\begin{pgfonlayer}{edgelayer}
		\draw (4.center) to (8);
		\draw (8) to (6.center);
	\end{pgfonlayer}
\end{tikzpicture}

%% file: figures/yrules/zy.tikz
\begin{tikzpicture}[scale=0.4, baseline={([yshift=-0.5ex]current bounding box.center)}]
	\begin{pgfonlayer}{nodelayer}
		\node [style=none] (4) at (0, 0) {};
		\node [style=none] (6) at (3, 0) {};
		\node [style=YZ, rotate=90] (8) at (1.5, 0) {\YZ};
	\end{pgfonlayer}
	\begin{pgfonlayer}{edgelayer}
		\draw (4.center) to (8);
		\draw (8) to (6.center);
	\end{pgfonlayer}
\end{tikzpicture}

%% file: figures/zx/ry.tikz
\begin{tikzpicture}[scale=0.4, baseline={([yshift=-0.5ex]current bounding box.center)}]
	\begin{pgfonlayer}{nodelayer}
		\node [style=none] (0) at (-1, 0) {};
		\node [style=none] (2) at (2, 0) {};
		\node [style={Y_med}] (4) at (0.5, 0) {$\alpha$};
	\end{pgfonlayer}
	\begin{pgfonlayer}{edgelayer}
		\draw (0.center) to (4);
		\draw (4) to (2.center);
	\end{pgfonlayer}
\end{tikzpicture}

%% file: figures/yrules/zp.tikz
\begin{tikzpicture}[scale=0.4, baseline={([yshift=-0.5ex]current bounding box.center)}]
	\begin{pgfonlayer}{nodelayer}
		\node [style=none] (4) at (0, 0) {};
		\node [style=none] (6) at (3, 0) {};
		\node [style={Z_med}] (7) at (1.5, 0) {$+$};
	\end{pgfonlayer}
	\begin{pgfonlayer}{edgelayer}
		\draw (4.center) to (7);
		\draw (7) to (6.center);
	\end{pgfonlayer}
\end{tikzpicture}

%% file: figures/yrules/xp.tikz
\begin{tikzpicture}[scale=0.4, baseline={([yshift=-0.5ex]current bounding box.center)}]
	\begin{pgfonlayer}{nodelayer}
		\node [style=none] (4) at (0, 0) {};
		\node [style=none] (6) at (3, 0) {};
		\node [style={X_med}] (7) at (1.5, 0) {$+$};
	\end{pgfonlayer}
	\begin{pgfonlayer}{edgelayer}
		\draw (4.center) to (7);
		\draw (7) to (6.center);
	\end{pgfonlayer}
\end{tikzpicture}

%% file: figures/yrules/zm.tikz
\begin{tikzpicture}[scale=0.4, baseline={([yshift=-0.5ex]current bounding box.center)}]
	\begin{pgfonlayer}{nodelayer}
		\node [style=none] (4) at (0, 0) {};
		\node [style=none] (6) at (3, 0) {};
		\node [style={Z_med}] (7) at (1.5, 0) {$-$};
	\end{pgfonlayer}
	\begin{pgfonlayer}{edgelayer}
		\draw (4.center) to (7);
		\draw (7) to (6.center);
	\end{pgfonlayer}
\end{tikzpicture}

%% file: figures/yrules/xm.tikz
\begin{tikzpicture}[scale=0.4, baseline={([yshift=-0.5ex]current bounding box.center)}]
	\begin{pgfonlayer}{nodelayer}
		\node [style=none] (4) at (0, 0) {};
		\node [style=none] (6) at (3, 0) {};
		\node [style={X_med}] (7) at (1.5, 0) {$-$};
	\end{pgfonlayer}
	\begin{pgfonlayer}{edgelayer}
		\draw (4.center) to (7);
		\draw (7) to (6.center);
	\end{pgfonlayer}
\end{tikzpicture}

%% file: figures/deriv/x14.tikz
\begin{tikzpicture}[scale=0.4, baseline={([yshift=-0.5ex]current bounding box.center)}]
	\begin{pgfonlayer}{nodelayer}
		\node [style=none] (2) at (0, 0) {};
		\node [style={X_med}] (10) at (3, 0) {$\pi$};
		\node [style={X_med}] (11) at (1.5, 0) {$\alpha$};
	\end{pgfonlayer}
	\begin{pgfonlayer}{edgelayer}
		\draw (2.center) to (11);
		\draw (11) to (10);
	\end{pgfonlayer}
\end{tikzpicture}

%% file: figures/ansatz/aoa.tikz
\begin{tikzpicture}
	\begin{pgfonlayer}{nodelayer}
		\node [style={med_rectv}] (0) at (0, 0) {$A(\alpha)$};
		\node [style={med_rectv}] (1) at (1.5, 0) {$B(\beta)$};
		\node [style=none] (2) at (-1.25, 0) {};
		\node [style=none] (3) at (-1.25, 0.75) {};
		\node [style=none] (4) at (-1.25, -0.75) {};
		\node [style=none] (5) at (2.75, -0.75) {};
		\node [style=none] (6) at (2.75, 0) {};
		\node [style=none] (7) at (2.75, 0.75) {};
		\node [style=none] (8) at (-0.75, 1.375) {};
		\node [style=none] (9) at (-0.75, -1.375) {};
		\node [style=none] (10) at (2.25, -1.375) {};
		\node [style=none] (11) at (2.25, 1.375) {};
	\end{pgfonlayer}
	\begin{pgfonlayer}{edgelayer}
		\draw (7.center) to (3.center);
		\draw (2.center) to (6.center);
		\draw (5.center) to (4.center);
		\draw [style=dashs] (11.center)
			 to (8.center)
			 to (9.center)
			 to (10.center)
			 to cycle;
	\end{pgfonlayer}
\end{tikzpicture}

%% file: figures/ansatz/tna.tikz
\begin{tikzpicture}
	\begin{pgfonlayer}{nodelayer}
		\node [style={smol_box}, font={\scriptsize}, inner sep=1ex] (0) at (0, 0.375) {$U(\{\theta\}_1)$};
		\node [style={smol_box}, font={\scriptsize}, inner sep=1ex] (1) at (0, -1.125) {$U(\{\theta\}_2)$};
		\node [style=none] (2) at (-1.25, 0) {};
		\node [style=none] (3) at (-1.25, 0.75) {};
		\node [style=none] (4) at (-1.25, -0.75) {};
		\node [style=none] (5) at (3, -0.75) {};
		\node [style=none] (6) at (3, 0) {};
		\node [style=none] (7) at (3, 0.75) {};
		\node [style=none] (8) at (-1.25, -1.5) {};
		\node [style=none] (9) at (3, -1.5) {};
		\node [style={smol_box}, font={\scriptsize}, inner sep=1ex] (10) at (1.75, -0.375) {$U(\{\theta\}_3)$};
	\end{pgfonlayer}
	\begin{pgfonlayer}{edgelayer}
		\draw (7.center) to (3.center);
		\draw (2.center) to (6.center);
		\draw (5.center) to (4.center);
		\draw (9.center) to (8.center);
	\end{pgfonlayer}
\end{tikzpicture}

%% file: ansatz/stairs.tikz
\begin{tikzpicture}[scale=0.8, baseline={([yshift=-0.5ex]current bounding box.south)}]
	\begin{pgfonlayer}{nodelayer}
		\node [style={Z_tri_l}] (0) at (0, 3) {0};
		\node [style={Z_tri_l}] (1) at (0, 2) {0};
		\node [style={Z_tri_l}] (2) at (0, 1) {0};
		\node [style={Z_big}] (3) at (5, 3) {$\alpha_1$};
		\node [style={Z_big}] (4) at (5, 2) {$\alpha_2$};
		\node [style={Z_big}] (5) at (5, 1) {$\alpha_3$};
		\node [style={X_big}] (6) at (1, 3) {$\theta_1$};
		\node [style={X_big}] (7) at (1, 2) {$\theta_2$};
		\node [style={X_big}] (8) at (1, 1) {$\theta_3$};
		\node [style={Z_big}] (9) at (2, 3) {$\theta_1$};
		\node [style={Z_big}] (10) at (2, 2) {$\theta_2$};
		\node [style={Z_big}] (11) at (2, 1) {$\theta_3$};
		\node [style=Z] (12) at (3, 3) {};
		\node [style=Z] (13) at (3.5, 2) {};
		\node [style=Z] (14) at (4, 1) {};
		\node [style={Z_tri_l}] (15) at (0, 0) {0};
		\node [style={Z_big}] (16) at (5, 0) {$\alpha_4$};
		\node [style={X_big}] (17) at (1, 0) {$\theta_4$};
		\node [style={Z_big}] (18) at (2, 0) {$\theta_4$};
		\node [style=X] (19) at (3, 2) {};
		\node [style=X] (20) at (3.5, 1) {};
		\node [style=X] (21) at (4, 0) {};
		\node [style={Z_big}] (22) at (9, 3) {$\gamma_1$};
		\node [style={Z_big}] (23) at (9, 2) {$\gamma_2$};
		\node [style={Z_big}] (24) at (9, 1) {$\gamma_3$};
		\node [style={Z_big}] (25) at (9, 0) {$\gamma_4$};
		\node [style=none] (26) at (10, 3) {};
		\node [style=none] (27) at (10, 2) {};
		\node [style=none] (28) at (10, 1) {};
		\node [style=none] (29) at (10, 0) {};
		\node [style={X_big}] (30) at (6, 3) {$\frac{\pi}{2}$};
		\node [style={Z_big}] (38) at (7, 3) {$\beta_1$};
		\node [style={Z_big}] (39) at (7, 2) {$\beta_2$};
		\node [style={Z_big}] (40) at (7, 1) {$\beta_3$};
		\node [style={Z_big}] (41) at (7, 0) {$\beta_4$};
		\node [style={X_big}] (42) at (6, 2) {$\frac{\pi}{2}$};
		\node [style={X_big}] (43) at (6, 1) {$\frac{\pi}{2}$};
		\node [style={X_big}] (44) at (6, 0) {$\frac{\pi}{2}$};
		\node [style={X_big}] (45) at (8, 3) {$\frac{-\pi}{2}$};
		\node [style={X_big}] (46) at (8, 2) {$\frac{-\pi}{2}$};
		\node [style={X_big}] (47) at (8, 1) {$\frac{-\pi}{2}$};
		\node [style={X_big}] (48) at (8, 0) {$\frac{-\pi}{2}$};
		\node [style=none] (49) at (2.5, -0.5) {};
		\node [style=none] (50) at (9.5, -0.5) {};
		\node [style=none] (51) at (2.5, 3.5) {};
		\node [style=none] (52) at (9.5, 3.5) {};
	\end{pgfonlayer}
	\begin{pgfonlayer}{edgelayer}
		\draw (0) to (6);
		\draw (1) to (7);
		\draw (2) to (8);
		\draw (15) to (17);
		\draw (17) to (18);
		\draw (8) to (11);
		\draw (7) to (10);
		\draw (6) to (9);
		\draw (9) to (12);
		\draw (12) to (3);
		\draw (10) to (19);
		\draw (19) to (13);
		\draw (13) to (4);
		\draw (11) to (20);
		\draw (20) to (14);
		\draw (14) to (5);
		\draw (18) to (21);
		\draw (21) to (16);
		\draw (14) to (21);
		\draw (13) to (20);
		\draw (12) to (19);
		\draw (3) to (30);
		\draw (4) to (42);
		\draw (5) to (43);
		\draw (16) to (44);
		\draw (30) to (38);
		\draw (42) to (39);
		\draw (43) to (40);
		\draw (44) to (41);
		\draw (41) to (48);
		\draw (40) to (47);
		\draw (39) to (46);
		\draw (38) to (45);
		\draw (45) to (22);
		\draw (22) to (26.center);
		\draw (46) to (23);
		\draw (47) to (24);
		\draw (48) to (25);
		\draw (25) to (29.center);
		\draw (24) to (28.center);
		\draw (23) to (27.center);
		\draw [style=dashs] (49.center) to (50.center);
		\draw [style=dashs] (51.center) to (49.center);
		\draw [style=dashs] (51.center) to (52.center);
		\draw [style=dashs] (52.center) to (50.center);
	\end{pgfonlayer}
\end{tikzpicture}

%% file: ansatz/stairsy.tikz
\begin{tikzpicture}[scale=0.8, baseline={([yshift=-0.5ex]current bounding box.south)}]
	\begin{pgfonlayer}{nodelayer}
		\node [style={Z_tri_l}] (0) at (0, 3) {0};
		\node [style={Z_tri_l}] (1) at (0, 2) {0};
		\node [style={Z_tri_l}] (2) at (0, 1) {0};
		\node [style={Z_big}] (3) at (5, 3) {$\alpha_1$};
		\node [style={Z_big}] (4) at (5, 2) {$\alpha_2$};
		\node [style={Z_big}] (5) at (5, 1) {$\alpha_3$};
		\node [style={X_big}] (6) at (1, 3) {$\theta_1$};
		\node [style={X_big}] (7) at (1, 2) {$\theta_2$};
		\node [style={X_big}] (8) at (1, 1) {$\theta_3$};
		\node [style={Z_big}] (9) at (2, 3) {$\theta_1$};
		\node [style={Z_big}] (10) at (2, 2) {$\theta_2$};
		\node [style={Z_big}] (11) at (2, 1) {$\theta_3$};
		\node [style=Z] (12) at (3, 3) {};
		\node [style=Z] (13) at (3.5, 2) {};
		\node [style=Z] (14) at (4, 1) {};
		\node [style={Z_tri_l}] (15) at (0, 0) {0};
		\node [style={Z_big}] (16) at (5, 0) {$\alpha_4$};
		\node [style={X_big}] (17) at (1, 0) {$\theta_4$};
		\node [style={Z_big}] (18) at (2, 0) {$\theta_4$};
		\node [style=X] (19) at (3, 2) {};
		\node [style=X] (20) at (3.5, 1) {};
		\node [style=X] (21) at (4, 0) {};
		\node [style={Z_big}] (22) at (7, 3) {$\gamma_1$};
		\node [style={Z_big}] (23) at (7, 2) {$\gamma_2$};
		\node [style={Z_big}] (24) at (7, 1) {$\gamma_3$};
		\node [style={Z_big}] (25) at (7, 0) {$\gamma_4$};
		\node [style=none] (26) at (8, 3) {};
		\node [style=none] (27) at (8, 2) {};
		\node [style=none] (28) at (8, 1) {};
		\node [style=none] (29) at (8, 0) {};
		\node [style={Y_big}] (38) at (6, 3) {$\beta_1$};
		\node [style={Y_big}] (39) at (6, 2) {$\beta_2$};
		\node [style={Y_big}] (40) at (6, 1) {$\beta_3$};
		\node [style={Y_big}] (41) at (6, 0) {$\beta_4$};
		\node [style=none] (49) at (2.5, -0.5) {};
		\node [style=none] (50) at (7.5, -0.5) {};
		\node [style=none] (51) at (2.5, 3.5) {};
		\node [style=none] (52) at (7.5, 3.5) {};
	\end{pgfonlayer}
	\begin{pgfonlayer}{edgelayer}
		\draw (0) to (6);
		\draw (1) to (7);
		\draw (2) to (8);
		\draw (15) to (17);
		\draw (17) to (18);
		\draw (8) to (11);
		\draw (7) to (10);
		\draw (6) to (9);
		\draw (9) to (12);
		\draw (12) to (3);
		\draw (10) to (19);
		\draw (19) to (13);
		\draw (13) to (4);
		\draw (11) to (20);
		\draw (20) to (14);
		\draw (14) to (5);
		\draw (18) to (21);
		\draw (21) to (16);
		\draw (14) to (21);
		\draw (13) to (20);
		\draw (12) to (19);
		\draw (22) to (26.center);
		\draw (25) to (29.center);
		\draw (24) to (28.center);
		\draw (23) to (27.center);
		\draw [style=dashs] (49.center) to (50.center);
		\draw [style=dashs] (51.center) to (49.center);
		\draw [style=dashs] (51.center) to (52.center);
		\draw [style=dashs] (52.center) to (50.center);
		\draw (3) to (38);
		\draw (38) to (22);
		\draw (39) to (4);
		\draw (39) to (23);
		\draw (40) to (24);
		\draw (40) to (5);
		\draw (16) to (41);
		\draw (41) to (25);
	\end{pgfonlayer}
\end{tikzpicture}

%% file: ansatz/bricks.tikz
\begin{tikzpicture}[scale=0.8, baseline={([yshift=-1ex]current bounding box.south)}]
	\begin{pgfonlayer}{nodelayer}
		\node [style={Z_tri_l}] (0) at (0, 3) {0};
		\node [style={Z_tri_l}] (1) at (0, 2) {0};
		\node [style={Z_tri_l}] (2) at (0, 1) {0};
		\node [style={Z_big}] (9) at (2, 3) {$\theta_1$};
		\node [style={Z_big}] (10) at (2, 2) {$\theta_2$};
		\node [style={Z_big}] (11) at (2, 1) {$\theta_3$};
		\node [style=Z] (12) at (4.25, 3) {};
		\node [style=Z] (13) at (4.75, 2) {};
		\node [style={Z_tri_l}] (15) at (0, 0) {0};
		\node [style={Z_big}] (18) at (2, 0) {$\theta_4$};
		\node [style=X] (19) at (4.25, 2) {};
		\node [style=X] (20) at (4.75, 1) {};
		\node [style={X_big}] (30) at (5.75, 3) {$\frac{\pi}{2}$};
		\node [style={Z_big}] (38) at (6.75, 3) {$\alpha_1$};
		\node [style={Z_big}] (39) at (6.75, 2) {$\alpha_2$};
		\node [style={Z_big}] (40) at (6.75, 1) {$\alpha_3$};
		\node [style={Z_big}] (41) at (6.75, 0) {$\alpha_4$};
		\node [style={X_big}] (42) at (5.75, 2) {$\frac{\pi}{2}$};
		\node [style={X_big}] (43) at (5.75, 1) {$\frac{\pi}{2}$};
		\node [style={X_big}] (44) at (5.75, 0) {$\frac{\pi}{2}$};
		\node [style={X_big}] (45) at (7.75, 3) {$\frac{-\pi}{2}$};
		\node [style={X_big}] (46) at (7.75, 2) {$\frac{-\pi}{2}$};
		\node [style={X_big}] (47) at (7.75, 1) {$\frac{-\pi}{2}$};
		\node [style={X_big}] (48) at (7.75, 0) {$\frac{-\pi}{2}$};
		\node [style={X_big}] (53) at (1, 3) {$\frac{\pi}{2}$};
		\node [style={X_big}] (54) at (1, 2) {$\frac{\pi}{2}$};
		\node [style={X_big}] (55) at (1, 1) {$\frac{\pi}{2}$};
		\node [style={X_big}] (56) at (1, 0) {$\frac{\pi}{2}$};
		\node [style={X_big}] (57) at (3, 3) {$\frac{-\pi}{2}$};
		\node [style={X_big}] (58) at (3, 2) {$\frac{-\pi}{2}$};
		\node [style={X_big}] (59) at (3, 1) {$\frac{-\pi}{2}$};
		\node [style={X_big}] (60) at (3, 0) {$\frac{-\pi}{2}$};
		\node [style=Z] (61) at (4.25, 1) {};
		\node [style=X] (62) at (4.25, 0) {};
		\node [style=Z] (63) at (8.75, 3) {};
		\node [style=Z] (64) at (9.25, 2) {};
		\node [style=X] (65) at (8.75, 2) {};
		\node [style=X] (66) at (9.25, 1) {};
		\node [style=Z] (67) at (8.75, 1) {};
		\node [style=X] (68) at (8.75, 0) {};
		\node [style=none] (69) at (10, 3) {};
		\node [style=none] (70) at (10, 2) {};
		\node [style=none] (71) at (10, 1) {};
		\node [style=none] (72) at (10, 0) {};
		\node [style=none] (73) at (3.75, 3.5) {};
		\node [style=none] (74) at (3.75, -0.5) {};
		\node [style=none] (75) at (9.75, -0.5) {};
		\node [style=none] (76) at (9.75, 3.5) {};
	\end{pgfonlayer}
	\begin{pgfonlayer}{edgelayer}
		\draw (19) to (13);
		\draw (13) to (20);
		\draw (12) to (19);
		\draw (30) to (38);
		\draw (42) to (39);
		\draw (43) to (40);
		\draw (44) to (41);
		\draw (41) to (48);
		\draw (40) to (47);
		\draw (39) to (46);
		\draw (38) to (45);
		\draw (61) to (62);
		\draw (61) to (20);
		\draw (57) to (12);
		\draw (60) to (62);
		\draw (59) to (61);
		\draw (58) to (19);
		\draw (9) to (57);
		\draw (10) to (58);
		\draw (11) to (59);
		\draw (18) to (60);
		\draw (56) to (18);
		\draw (55) to (11);
		\draw (54) to (10);
		\draw (53) to (9);
		\draw (0) to (53);
		\draw (1) to (54);
		\draw (2) to (55);
		\draw (15) to (56);
		\draw (12) to (30);
		\draw (13) to (42);
		\draw (20) to (43);
		\draw (62) to (44);
		\draw (65) to (64);
		\draw (64) to (66);
		\draw (63) to (65);
		\draw (67) to (68);
		\draw (67) to (66);
		\draw (63) to (69.center);
		\draw (64) to (70.center);
		\draw (66) to (71.center);
		\draw (68) to (72.center);
		\draw (45) to (63);
		\draw (46) to (65);
		\draw (48) to (68);
		\draw (47) to (67);
		\draw [style=dashs] (73.center) to (76.center);
		\draw [style=dashs] (76.center) to (75.center);
		\draw [style=dashs] (75.center) to (74.center);
		\draw [style=dashs] (74.center) to (73.center);
	\end{pgfonlayer}
\end{tikzpicture}

%% file: ansatz/bricksy.tikz
\begin{tikzpicture}[scale=0.8, baseline={([yshift=-1ex]current bounding box.south)}]
	\begin{pgfonlayer}{nodelayer}
		\node [style={Z_tri_l}] (0) at (0, 3) {0};
		\node [style={Z_tri_l}] (1) at (0, 2) {0};
		\node [style={Z_tri_l}] (2) at (0, 1) {0};
		\node [style={Y_big}] (9) at (1, 3) {$\theta_1$};
		\node [style={Y_big}] (10) at (1, 2) {$\theta_2$};
		\node [style={Y_big}] (11) at (1, 1) {$\theta_3$};
		\node [style=Z] (12) at (2.25, 3) {};
		\node [style=Z] (13) at (2.75, 2) {};
		\node [style={Z_tri_l}] (15) at (0, 0) {0};
		\node [style={Y_big}] (18) at (1, 0) {$\theta_4$};
		\node [style=X] (19) at (2.25, 2) {};
		\node [style=X] (20) at (2.75, 1) {};
		\node [style={Y_big}] (38) at (3.75, 3) {$\alpha_1$};
		\node [style={Y_big}] (39) at (3.75, 2) {$\alpha_2$};
		\node [style={Y_big}] (40) at (3.75, 1) {$\alpha_3$};
		\node [style={Y_big}] (41) at (3.75, 0) {$\alpha_4$};
		\node [style=Z] (61) at (2.25, 1) {};
		\node [style=X] (62) at (2.25, 0) {};
		\node [style=Z] (63) at (4.75, 3) {};
		\node [style=Z] (64) at (5.25, 2) {};
		\node [style=X] (65) at (4.75, 2) {};
		\node [style=X] (66) at (5.25, 1) {};
		\node [style=Z] (67) at (4.75, 1) {};
		\node [style=X] (68) at (4.75, 0) {};
		\node [style=none] (69) at (6.25, 3) {};
		\node [style=none] (70) at (6.25, 2) {};
		\node [style=none] (71) at (6.25, 1) {};
		\node [style=none] (72) at (6.25, 0) {};
		\node [style=none] (73) at (1.75, 3.5) {};
		\node [style=none] (74) at (5.75, 3.5) {};
		\node [style=none] (75) at (5.75, -0.5) {};
		\node [style=none] (76) at (1.75, -0.5) {};
	\end{pgfonlayer}
	\begin{pgfonlayer}{edgelayer}
		\draw (19) to (13);
		\draw (13) to (20);
		\draw (12) to (19);
		\draw (61) to (62);
		\draw (61) to (20);
		\draw (65) to (64);
		\draw (64) to (66);
		\draw (63) to (65);
		\draw (67) to (68);
		\draw (67) to (66);
		\draw (63) to (69.center);
		\draw (64) to (70.center);
		\draw (66) to (71.center);
		\draw (68) to (72.center);
		\draw (68) to (15);
		\draw (2) to (67);
		\draw (65) to (1);
		\draw (0) to (63);
		\draw [style=dashs] (73.center) to (76.center);
		\draw [style=dashs] (76.center) to (75.center);
		\draw [style=dashs] (75.center) to (74.center);
		\draw [style=dashs] (74.center) to (73.center);
	\end{pgfonlayer}
\end{tikzpicture}

%% file: chapter2.tex
\chapter{Developing Diagrammatic QML}
This chapter introduces the phase gadget and the Pauli gadget, and explains why they are the natural way to reason with parameterised quantum circuits. 
We first express phase polynomials in terms of phase gadgets in section \ref{power}, then we characterise the quality of phase polynomial ans\"{a}tze by counting the number of phase gadgets they generate. We also show how to represent a wide class of circuits using only Pauli gadgets in section \ref{convert_ansatz}. 

\section{Phase Gadgets}\label{phase_gadget}
The \textit{phase gadget} \autocite{phasegadget} is diagrammatically defined as 
$$ \Phi_n(\theta) = \input{phase_gadget/pg.tikz} $$
and has the following properties. \\

\noindent\textbf{Fusion rule}:
	$$\input{phase_gadget/pg_2_1.tikz} = \input{phase_gadget/pg_2_2.tikz}$$
\begin{proof}
	The left-hand side can be rearranged to the right-hand side with the following intermediate steps:
    $$\input{phase_gadget/pg_2_5.tikz} \eq{f} \input{phase_gadget/pg_2_3.tikz} \eq{b} \input{phase_gadget/pg_2_4.tikz}$$
\end{proof}	
\noindent\textbf{Commutation rule}:
	$$\input{phase_gadget/pg_3_1.tikz} = \input{phase_gadget/pg_3_2.tikz}$$
\begin{proof}The gadget legs commute through spider fusion and unfusion.\end{proof}
\noindent\textbf{Decomposition rule (multi-legs)}:
	$$\input{phase_gadget/pg_4_1.tikz} = \input{phase_gadget/pg_4_2.tikz}$$
\begin{proof}The 2 legged X spider is the identity map, so the two Z spiders fuse. \end{proof}
\noindent\textbf{Decomposition rule (single-leg)}:
	 $$\input{phase_gadget/pg_1_2.tikz} = \input{phase_gadget/pg_1_1.tikz}$$
\begin{proof}
	The left-hand side can be rearranged to the right-hand side with the following intermediate steps:
    $$\input{phase_gadget/pg_1_3.tikz} = \input{phase_gadget/pg_1_4.tikz} \eq{b} \input{phase_gadget/pg_1_5.tikz}$$
 \end{proof}

\noindent\textbf{Identity rule}:
\[\input{phase_gadget/pg_5_1.tikz} \eq{c} 
\input{phase_gadget/pg_5_3.tikz} \eq{f, id1} 
\input{phase_gadget/pg_5_4.tikz}\]
The identity and addition rules make the phase gadget a one parameter 
unitary group, so we can find its Hamiltonian by differentiating with respect to $\theta$.

$$\input{phase_gadget/pg_diff1.tikz}
\eq{d}\input{phase_gadget/pg_diff2.tikz}
\eq{$\pi$}\input{phase_gadget/pg_diff3.tikz}
\eq{f}\input{phase_gadget/pg_diff4.tikz}$$

By Stone's theorem, $\Phi_n(\theta) = e^{i\frac{\theta}{2} P}$, where the Hamiltonian $P\in\{I,Z\}^{\otimes n}$ is a tensor product of $I$ and $Z$ that depends on the arrangement of the legs. For example, 

$$\left\llbracket\input{phase_gadget/pg_def.tikz}\right\rrbracket = \exp\left(i \frac{\theta}{2} (Z\otimes I\otimes Z)\right)$$

The decomposition rules allow a phase gadget to be decomposed into CNOT and Z rotation gates. The effect of the phase gadget on the circuit is invariant with respect to the wires the phase gadget acts on, and the decomposition of a phase gadget is not unique; we can decompose in a manner that increases the circuit depth by $\lceil\log_2 n\rceil$ instead of $n$.
$$ \input{phase_gadget/pg.tikz} = \input{phase_gadget/pg_decompose.tikz} $$
$$ \input{phase_gadget/pg_ladder.tikz} = \input{phase_gadget/pg_tree.tikz} $$
The decomposition rules of phase gadgets combined with complementarity gives us the following commutation rule
for Z and X phase gadgets
$$ \input{phase_gadget/pg_cnot1.tikz} \eq{com} \input{phase_gadget/pg_cnot2.tikz} $$
$$ \input{phase_gadget/pg_cnot3.tikz} \eq{com} \input{phase_gadget/pg_cnot4.tikz} $$
This commutation rule will be used extensively throughout the project to obtain normal forms and parametric redundancy results.
\section{Phase Polynomials}
Phase polynomials are a class of quantum circuits composed of CNOT and Z rotation gates. 
They were first introduced by Amy et al.\ \autocite{Amy_2018} in the \textit{sum-over-paths} 
form:

$$f(\boldsymbol{x})=\sum_{\boldsymbol{y} \in \mathbb{P}_{2}^{n}} \widehat{f}(\boldsymbol{y}) \cdot\left(x_{1} y_{1} \oplus x_{2} y_{2} \oplus \cdots \oplus x_{n} y_{n}\right)$$

$$U_{C}=\sum_{x \in \mathbb{F}_{2}^{n}} e^{2 \pi i f(\boldsymbol{x})}|A \boldsymbol{x}\rangle\langle\boldsymbol{x}|$$

where $A$ can be viewed as a change of basis in $GF(2)$ and $\hat{f}(\bm y)$ can be viewed as the Fourier coefficients.  Phase polynomials are of interest to quantum computing researchers because they are central to the construction of Instantaneous Quantum Polynomial-time (IQP) circuits, which are believed to be hard to classically simulate; Bremner et al.\ have shown that IQP circuits can sample from a probability distribution that is believed to be classically intractable to sample, 
provided the polynomial hierarchy does not collapse \autocite{IQP}.

\ctikzfig{iqp}

\subsection{Characterising the Power of Phase Polynomials}\label{power}

To understand the power of phase polynomials, we will study the unitaries represented by phase gadgets, which can be used to describe the dimensionality of a parameterised phase polynomial. Diagrammatically, a phase gadget takes its inputs in the Z basis, copies it, XOR's the bits together, and multiplies the state by $\exp(-i\theta/2)$ or $\exp(i\theta/2)$ depending on the parity. One can think of the phase gadget as a many-qubit generalisation of the Z rotation gate.

$$ \left\llbracket\input{pp/interp.tikz}\right\rrbracket = 
%{\tiny
%\begin{pmatrix}
%e^{-i\frac{\theta}{2}} &  &  &  &  &  &  & \\
% & e^{i\frac{\theta}{2}} &  &  &  &  &  & \\
% &  & e^{i\frac{\theta}{2}} &  &  &  &  & \\
% &  &  & e^{-i\frac{\theta}{2}} &  &  &  & \\
% &  &  &  & e^{i\frac{\theta}{2}} &  &  & \\
% &  &  &  &  & e^{-i\frac{\theta}{2}} &  & \\
% &  &  &  &  &  & e^{-i\frac{\theta}{2}} & \\
% &  &  &  &  &  &  & e^{i\frac{\theta}{2}}
% \end{pmatrix}} = 
 \text{diag}\left\{\begin{pmatrix}
1\\
e^{i\theta} \\
e^{i\theta} \\
1 \\
e^{i\theta} \\
1 \\ 1 \\
e^{i\theta} \\
\end{pmatrix}
\cdot e^{-i\frac{\theta}{2}}
\right\}
$$

Each phase gadget corresponds to a diagonal matrix in the Z basis, where the $i$th diagonal entry $e^{-i\theta/2}$ or $e^{i\theta/2}$ depending on the parity of its legs.  
Phase gadgets alone cannot change the distribution of the observed state, but can be powerful when combined with other components, such as a change of basis.
$$
\left\llbracket\input{gs/diagonal_complete.tikz} \right\rrbracket =
\text{diag}\left\{\begin{pmatrix}
1\\
e^{i( a+b+d+e)}\\
e^{i( a+b+c+f)}\\
e^{i( c+d+e+f)}\\
e^{i( a+c+d+g)}\\
e^{i( b+c+e+g)}\\	
e^{i( b+d+f+g)}\\
e^{i( a+e+f+g)}
\end{pmatrix} \cdot s
\right\}
$$
$$ s = e^{-i(a+b+c+d+e+f+g)/2} $$
In general, a complete set of phase gadgets forms a basis for diagonal unitaries, up to a global phase. This means that the diagonal phase polynomials used in IQP circuits can be constructed using phase gadgets alone. In the more general case, phase polynomials combine the phase gadgets along with a layer of CNOTs to create a non-diagonal permutation matrix, where each non-zero entry is a phase rotation of the form $e^{i\theta}$.
\begin{align*}
	\text{Just Phase Gadgets:\quad} &\displaystyle\sum_{\bm{x} \in \mathbb{F}_{2}^{n}} e^{2 \pi i f(\boldsymbol{x})}|\boldsymbol{x}\rangle\langle\boldsymbol{x}| \\
	\text{Just CNOTs:\quad} &\displaystyle\sum_{\bm{x} \in \mathbb{F}_{2}^{n}} |A\boldsymbol{x}\rangle\langle\boldsymbol{x}| \\
	\text{Phase Polynomials:\quad} &\displaystyle\sum_{\bm{x} \in \mathbb{F}_{2}^{n}} e^{2 \pi i f(\boldsymbol{x})}|A\boldsymbol{x}\rangle\langle\boldsymbol{x}|
\end{align*}
When Amy et al.\ \autocite{Amy_2018} introduced the sum-over-paths form of a phase polynomial, they commented that there is a many-to-one correspondence from CNOT+RZ circuits to the sum-over-paths form of a phase polynomial, but the sum-over-paths form has a one-to-one correspondence with the underlying unitary. Diagrammatically, we can obtain a normal form for phase polynomials by ``dragging" all of the phase rotations to the left to obtain a layer of (commuting) phase gadgets followed by a layer of CNOTs. This normal form corresponds exactly to the sum-over-paths form.

By considering the diagonal unitary matrix as an element in a $2^n$ dimensional vector space, the circuit parameterisation with $k$ unique phase gadgets spans a $k$-dimensional subspace. Put simply, the more unique phase gadgets in the phase polynomial, the more expressive the phase polynomial is. This will be a guiding principle when characterising the strength of other variational quantum circuits.

\section{Learning Phase Polynomials}

Although phase polynomials alone are not as expressive as general QML ans\"{a}tze, their desirable commutative properties make them a good starting point to study normal forms and parametric redundancies of QML ans\"{a}tze. 
One aspect of ans\"{a}tz design worth studying is the effect of CNOT arrangement on the circuit. 
Many proposed ans\"{a}tze use only CNOTs and phase rotations, but tend to only have CNOT gates of one orientation. This may be because a misplaced vertically-inverted CNOT can lead to swap gates, which may cause parameters to fuse. 
Therefore it is necessary to introduce the following definitions:
\begin{definition}
	\textit{Parameterised phase polynomials} (PPPs) are circuits constructed using only CNOT and parameterised RZ gates.
\end{definition}

\begin{definition}
    \textit{Monotonic parameterised phase polynomials} (MPPPs) are circuits constructed using only CNOTs and parameterised RZ gates, with the CNOTs arranged in a monotonic fashion. More formally, given a totally ordered set of wires, a collection of CNOTs $\{ CNOT(x_i, y_i) \}_i$ is monotonic if all of the CNOTs are applied in such a way that either 
    $\forall i: \: x_i < y_i$ or $\forall i: \: x_i > y_i$.
\end{definition}

Studying MPPPs could tell us whether avoiding vertically inverted CNOTs entirely is a good idea when designing ans\"{a}tze. Ans\"{a}tze that generate MPPPs can be obtained by removing the RX gates and inverted CNOT gates from ans\"{a}tze that only consist of RZ, RX and CNOT gates. For example, we could remove the RX gates from the ans\"{a}tze in section \ref{convert_ansatz}.

\begin{center}
\input{pp/stairs.tikz}
\hspace{1cm}
\input{pp/bricks.tikz}
\end{center}

Like other phase polynomials, these MPPP ans\"{a}tze yield a normal form by propagating all RZ gates to the right. 
If the repeating unit is repeated five times, the first layer of RZ rotations eventually transforms into one-legged gadgets as it passes through the repeating unit, which by the second decomposition rule become RZ rotations once again; these RZ rotations commute with the gadgets from the second to fourth layers and fuse with the RZs in the fifth layer, making those parameters redundant. This ``periodicity" phenomenon is not unique to the ``brick wall" layout, and actually generalises to all MPPPs. (See theorem \ref{period})

$$\input{pp/brick_drag1.tikz}$$
$$\input{pp/brick_drag2.tikz}$$
$$\input{pp/brick_drag3.tikz}$$
$$\input{pp/brick_drag4.tikz}$$

If we write the positions of the phase gadget's legs as a vector in $GF(2)^n$, we see that the propagation of a circuit of CNOTs through a phase gadget is a linear map in $GL(n, 2)$. By utilising the linearity of the transformation, the entire transformation can be calculated by commuting Z rotations on each qubit, which are represented by the basis vectors in $GF(2)^n$. As an example, the staircase and brick wall layers from section \ref{conv} have mappings
$\displaystyle ( x_{1} ,x_{2} ,x_{3} ,x_{4}) \mapsto ( x_{1} \oplus x_{2} ,x_{2} \oplus x_{3} ,x_{3} \oplus x_{4} ,x_{4})$ and $\displaystyle ( x_{1} ,x_{2} ,x_{3} ,x_{4}) \mapsto ( x_{1} \oplus x_{2} ,x_{2} \oplus x_{3} \oplus x_{4} ,x_{3} \oplus x_{4} ,x_{4})$ respectively and correspond to matrices

$$
\begin{aligned}
\begin{pmatrix}
1 & 1 & 0 & 0\\
 & 1 & 1 & 0\\
 &  & 1 & 1\\
 &  &  & 1
\end{pmatrix} & \ \ \ \ \ \ \text{and} \ \ \ \ \ \ \begin{pmatrix}
1 & 1 & 0 & 0\\
 & 1 & 1 & 1\\
 &  & 1 & 1\\
 &  &  & 1
\end{pmatrix}
\end{aligned}
$$

respectively. We emphasise that these matrices are not the same matrices as the change of basis $A$, but in fact the \emph{inverse transpose} of $A$, $(A^{-1})^T$. The proof of this is given in 
theorem \ref{inv_trans}.

 %As a special case of the sum over paths form, CNOT circuits can be written as $\ket{A \bm x}\bra{\bm x}$ and equivalent circuits that represent the same unitary map will share the same change of basis $A$. Likewise, two phase gadgets that each pass through equivilent CNOT circuits will be transformed in the same manner. 
% Nonetheless, such matrices form a group: define $\mathbb{B}(n)$ to be the group of $n\times n$ invertible binary matrices under multiplication. These matrices act on a binary vector which represent the legs of a Z phase gadget. 
%\subsection{Periodicity}\label{period}

As seen in section \ref{PQC architectures}, layered gate ans\"{a}tze repeat their layers as that is believed to be a good technique to generate good ans\"{a}tze. But is this the case for MPPP circuits? During the conversion to the normal form, the Z rotations on the $i$th layer pass through $L-i$ layers of CNOTs, where $L$ is the number of repeated layers. This process is captured by repeatedly applying the matrix $B\in GL(n, 2)$ representing the CNOT block to the vector in $GF(2)^n$ representing the legs of a phase gadget. Overall the action on the phase gadget after passing through $l$ repeating layers is given by the matrix $B^l$. 

Seeing as how the 4 qubit brick ans\"{a}tz transforms the Z rotations back into Z rotations after 4 layers of CNOTs, we may wish to find the \textit{group order} of such matrices in general.  Any MPPP circuit has a lower triangular change of basis $A$, so the action of its CNOTs on a Z phase gadget is given by upper triangular matrix $B=(A^{-1})^T$. If the group order of such matrices is low, then the phase polynomial ans\"{a}tz produces fewer unique phase gadgets and so the overall phase polynomial is not very expressive.
\begin{theorem}\label{period}
	The maximum order of an $n\times n$ triangular matrix $B$ in $GL(n,2)$ is $m=2^{\lceil \log_2n \rceil}$. That is: $B^m = I$. 
\end{theorem}
\begin{proof}
	The triangular matrices form a group of order $2^{n(n-1)/2}$ under multiplication, so the order of $B$ must be a power of 2 by Lagrange's theorem. Let $B = I + L$. By the freshman's dream lemma: $B^m = I + L^m$. Triangular matrices with a zero diagonal are nilpotent with $L^n = 0$, as the size of the triangle of non-zero entries decreases per left multiplication of $L$. Since $m \geq n$, $L^m = 0$ so $B^m = I$. 
\end{proof}

This theorem immediately generalises to $GF(p)$.

\begin{corollary}
	 An $n$-qubit MPPP circuit generated from layered gate ans\"{a}tz has no more than $2n^2$ non-redundant parameters.
\end{corollary}

\begin{figure}
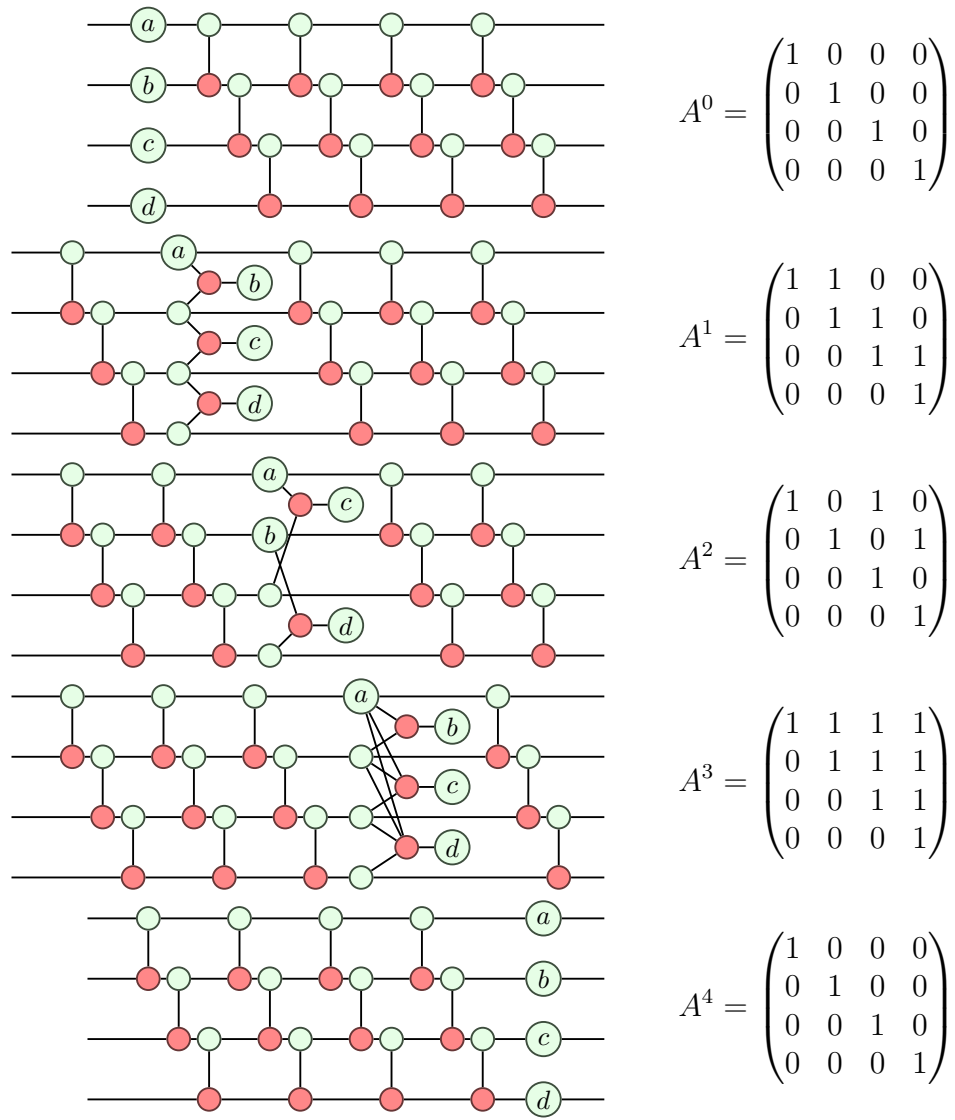

\begin{align*}
\input{pp/big/stair1.tikz} &
\qquad A^0 =
\begin{pmatrix}
1 & 0 & 0 & 0\\
0 & 1 & 0 & 0\\
0 & 0 & 1 & 0\\
0 & 0 & 0 & 1\\
\end{pmatrix}\\
\input{pp/big/stair2.tikz} &
\qquad A^1 =
\begin{pmatrix}
1 & 1 & 0 & 0\\
0 & 1 & 1 & 0\\
0 & 0 & 1 & 1\\
0 & 0 & 0 & 1\\
\end{pmatrix}\\
\input{pp/big/stair3.tikz} &
\qquad A^2 = 
\begin{pmatrix}
1 & 0 & 1 & 0\\
0 & 1 & 0 & 1\\
0 & 0 & 1 & 0\\
0 & 0 & 0 & 1\\
\end{pmatrix}\\
\input{pp/big/stair4.tikz} &
\qquad A^3 = 
\begin{pmatrix}
1 & 1 & 1 & 1\\
0 & 1 & 1 & 1\\
0 & 0 & 1 & 1\\
0 & 0 & 0 & 1\\
\end{pmatrix}\\
\input{pp/big/stair5.tikz} &
\qquad A^4 = 
\begin{pmatrix}
1 & 0 & 0 & 0\\
0 & 1 & 0 & 0\\
0 & 0 & 1 & 0\\
0 & 0 & 0 & 1\\
\end{pmatrix}\\
\end{align*}

\caption{How a layer of RZ gates are transformed when passing through a CNOT layer. The legs of each phase gadget can be read off vertically.}
\end{figure}

When the CNOTs are not restricted to be arranged monotonically, then the order is no longer linear in $n$. Since a swap gate can be implemented using 3 CNOTs, all permutations can be generated.
%\st{The order of a permutation is given by the least common multiple of its canonical cycles, and}
The greatest order of an element in the permutation group $S_n$ is given by Landau's function, which asymptotically grows as $e^{\sqrt{n\log n}}$ 
\autocite{landau}. This provides evidence that general parametric quantum circuits should not arrange their CNOTs monotonically, but instead have a more symmetric ans\"{a}tz by adding CNOTs from the last qubit to the first qubit.

\section{Normal Form}
We see that a normal form of phase polynomials can be achieved by dragging the phase gadgets across the layers of CNOTs.
This can be generalised to ans\"{a}tze consisting of CNOT, RZ and RX rotation gates. Although this form is less illuminating than its application to phase polynomials, it is easier to simplify than the original circuit as it is an architecture-agnostic description of the circuit. Phase gadgets are also a higher level, more concise description of the circuit than interleaving CNOTs and rotation gates.
In fact, we suggest that it is clearer and more convenient to directly use Z and X phase gadgets when describing QML circuits in many scenarios; this is analogous to writing code in a high level language instead of machine code. 

\ctikzfig{nf/bricks_nf}
\ctikzfig{nf/bricks_nf1}
\ctikzfig{nf/bricks_nf2}
%\ctikzfig{nf/bricks_nf3}

\section{Initialisation Strategy}\label{init_strat}
As discussed in section \ref{barren}, randomly generated ans\"{a}tz have landscapes full of barren plateaus 
where their gradients are close to 0, making it difficult for gradient descent to find a 
good minimum. 

Grant et al.\ \autocite{init} shows that the randomly generated ans\"{a}tz can be used for learning, provided 
the right initial parameters are used. Geometrically, this means finding a starting point that is 
not in a barren plateau, so it can descend towards a good minimum. Of course, finding such a point 
in an arbitrary circuit is difficult, just like finding a starting point close to the global minimum 
is difficult. Instead, Grant et al.\ presents an initialisation strategy for a wide class of randomly 
generated ans\"{a}tz that avoids the barren plateaus as a starting point. 

First, the ans\"{a}tz is arranged in layers of unitary blocks.
$$ \input{cnotfree/d1.tikz} $$
Each unitary block is composed of smaller sub-blocks 
$$ \input{cnotfree/d2.tikz} $$
which are arranged such that $U^j_i(\theta^j_i) = U^{L-j}_i(\theta^{L-j}_i)^\dag$: 
this is done so that each unitary block evaluates to the identity, for those parameters. 
Concretely, each unitary sub-block can be generated randomly using CNOT, RZ and RX gates yielding a normal form:
$$\input{cnotfree/c1.tikz}$$
$$\input{cnotfree/c2.tikz}$$
A consequence of this normal form is that the overall ans\"{a}tz is ``CNOT-free", as the layer of CNOTs 
in each unitary sub-block necessarily cancel out:
$$\input{cnotfree/c3.tikz}$$
$$\input{cnotfree/c4.tikz}$$
$$\input{cnotfree/c5.tikz}$$

This provides more evidence that phase gadgets are a natural, architecture-agnostic way of designing and analysing QML ans\"{a}tze.
%This normal form is also useful for investigating network parameter redundancy because we have moved from investigating whether 
%phases fuse within the circuit to investigating whether the phase gadgets fuse.

\subsection{Commutativity Relation}
\begin{theorem} \label{zxswap}
	Z phase gadgets and X phase gadgets commute if and only if they share an even number of legs. 	
	$$ 2n \: \{ \input{commute/commute1.tikz} = \input{commute/commute2.tikz} \} \: 2n $$
\end{theorem}
\begin{proof}
$$ \input{gs/rg_swap_proof1.tikz} = \input{gs/rg_swap_proof2.tikz} = \input{gs/rg_swap_proof3.tikz} $$
If an even number of legs are swapped, the wires created between the ``body" of the phase gadgets disappear by complementarity.
$$ \input{gs/rg_swap_proof4.tikz} = \input{gs/rg_swap_proof5.tikz} $$
\end{proof}

This allows us to work directly with phase gadgets, with an understanding of which ones commute with each other; it does not allow us to commute rotations gates that do not 
originally commute --- if there are 
two non-commutative rotation gates adjacent on a wire, they will transform to phase gadgets sharing an odd number of legs after passing through a layer of CNOTs. If this were not the case, then the gadgets would commute: inverting the layer of CNOTs would leave the original rotation gates commuted, which would be a contradiction. This is not surprising; after all, rotation gates are special cases of phase gadgets, so the theorem should still apply.

\vspace{2mm}
\resizebox{\textwidth}{!}{\input{gs/ex1.tikz}=\input{gs/ex2.tikz}}
\begin{center}\scalebox{0.8}{=\input{gs/ex3.tikz}}\end{center}

The manipulations used in the proof of theorem \ref{zxswap} can be used to forcefully commute the gadgets 
such that the Z and X gadgets are separated. The \textit{graph state} can be thought of as a controller for the 
gadgets, and is equivalent to a superposition of separated phase gadgets. Graph states are originally used in 
measurement-based quantum computing \autocite{backens2014zx, duncan2009graph, duncan2012graphical}.

\section{Pauli Gadgets} \label{pauli_gadget}
In the paper \textit{Phase Gadget Synthesis for Shallow Circuits}, Cowtan et al.\ presents the theory behind Pauli gadgets \autocite{pauligadget}, which are a generalisation of phase gadgets with a rich set of properties. The Pauli gadgets are diagrammatically defined as 
$$\input{pauli_gadget/pg.tikz} = \input{pauli_gadget/pg1.tikz}$$
In section \ref{phase_gadget}, we that the Z phase gadget is generated by a tensor product of $\{Z, I\}$, two of the four Pauli matrices. By differentiating the Pauli gadget, we see that it is generated using all four Pauli matrices $\{X,Y,Z,I\}$.
$$\input{pauli_gadget/pg2.tikz} = \input{pauli_gadget/pg3.tikz}$$
$$\input{pauli_gadget/pg4.tikz} = \input{pauli_gadget/pg5.tikz}$$

$$\left\llbracket\input{pauli_gadget/pg1.tikz}\right\rrbracket = \exp\left(i \frac{\theta}{2} (I\otimes X\otimes Y\otimes Z)\right)$$

We have already encountered Pauli gadgets before --- the X phase gadget is a special case of the Pauli gadget where all of its legs are X legs.
$$ \input{gs/r_is_pauli1.tikz} = \input{gs/r_is_pauli2.tikz} =  \input{gs/r_is_pauli3.tikz} = \input{gs/r_is_pauli4.tikz} $$
Like phase gadgets, the Pauli gadgets interact with the CNOT and Clifford gates (rotations of the form $R(k\pi/2)$) in a nice way, shown in see figures 3.2 and 3.3. These properties have been used in the latest simplification routines of \tket \autocite{pauligadget, cowtan2020generic}.
The coloured legs of the Pauli gadgets allow us to see which gadgets commute by applying the fusion rule. This is much simpler than using traditional circuit diagram notation to analyse circuit ans\"{a}tze.

\input{pauli.tex}

\section{Commutation Relations}

\begin{theorem}
Pauli Gadgets commute when they mismatch on an even number of legs.
Non-diagrammatically, they commute when their Hamiltonians commute.
\end{theorem}
\begin{proof}
There are three possible types of mismatches for the legs of the Pauli gadget: ZX, XY, and YZ. Every time a mismatched leg passes through each other, a Hadamard wire appears 
between the bodies of the Pauli gadget.
\[ \input{gs/zx_comm1.tikz} \]
\[ \input{gs/zx_comm2.tikz} =  \input{gs/zx_comm3.tikz} = \input{gs/zx_comm4.tikz} \]
\[ \input{gs/xy_comm1.tikz} \]
\[ \input{gs/xy_comm2.tikz} =  \input{gs/xy_comm3.tikz} = \input{gs/xy_comm4.tikz} \]
\[ \input{gs/yz_comm1.tikz} \]
\[ \input{gs/yz_comm2.tikz} =  \input{gs/yz_comm3.tikz} = \input{gs/yz_comm4.tikz} \]
\[ \input{gs/yz_comm5.tikz} =  \input{gs/yz_comm6.tikz} \]
When there are an even number of mismatches, the wires between the bodies of the Pauli gadget 
disappear by complementarity.
\[ \input{gs/had_comp1.tikz} = \input{gs/had_comp2.tikz} =  \input{gs/had_comp3.tikz} = \input{gs/had_comp4.tikz} \]
\end{proof}

\section{Euler Decomposition} \label{eu}

We have already defined the Hadamard gate in terms of $\pm\frac{\pi}{2}$ rotations. The decomposition of the Hadamard gate in terms of rotations of alternating bases has been studied geometrically in $SO(3)$ by Euler, and the definition of the Hadamard decomposition was introduced by Duncan and Pedrix \autocite{duncan2009graph} as an early attempt to complete the ZX calculus.  More generally, alternating rotation gates can be rewritten in terms of rotation gates with a different sequence of bases. 
There is an isomorphism between $SU(2)$ and $SO(3)/\mathbb{Z}_2$, so Euler decomposition in $SU(2)$ can be visualised using the Bloch sphere akin to how Euler decomposition can be visualised using $S^3$.
\begin{theorem}
For $a_1, a_2, a_3 \in (-\pi, \pi]$, there exists $b_1, b_2, b_3 \in (-\pi, \pi]$ such that
$$R_{x}(a_{3}) R_{z}(a_{2}) R_{x}(a_{1}) = R_{z}(b_{3}) R_{x}(b_{2}) R_{z}(b_{1})$$
\begin{center} and \end{center}
$$R_{z}(b_{3}) R_{x}(b_{2}) R_{z}(b_{1}) = R_{x}(a_{3}) R_{z}(a_{2}) R_{x}(a_{1}),$$
where $b_1, b_2, b_3$ is defined using the auxiliary variables $z_1, z_2$: 
$$z_{1} =\cos \frac{a_2}{2}\cos\left(\frac{a_{1}+a_{3}}{2}\right) +i\sin \frac{a_2}{2}\cos\left(\frac{a_{1}-a_{3}}{2}\right)$$
$$z_{2} =\cos \frac{a_2}{2}\sin\left(\frac{a_{1}+a_{3}}{2}\right) -i\sin \frac{a_2}{2}\sin\left(\frac{a_{1}-a_{3}}{2}\right)$$
\begin{align*}
    b_{1} &=\arg z_{1} +\arg z_{2}\\
    b_{2} &=\arctan\left(\frac{|z_{2} |}{|z_{1} |}\right)\\
          &=\arg( |z_{1} |+i|z_{2} |)\\
    b_{3} &=\arg z_{1} -\arg z_{2}.
\end{align*}
\end{theorem}
Diagrammatically the equations can be written as
$$ \input{eu/eq1.tikz} = \input{eu/eq2.tikz} $$
Euler decomposition is not applicable to all Pauli gadgets --- Duncan et al.\ \autocite{pauligadget} showed that three Pauli gadgets $e^{iP_1}e^{iP_2}e^{iP_3}$ are only admissible to Euler decomposition when $P_1 = P_3$, i.e. when the legs of the first and third Pauli gadgets match.

Euler decomposition allows us to study a partial redundancy/dependence between alternating phase rotations and gadgets. For example, a circuit made 
from alternating rotations can be simplified as follows:

$$ \input{eu/eu1.tikz} $$
$$ = \input{eu/eu2.tikz} = \input{eu/eu3.tikz} $$

\section{Converting Everything to Gadgets}\label{convert_ansatz}

\subsection{Universality}
So far we have only analysed circuits constructed exclusively using CNOT, RZ and RX gates and how they give rise to a normal form 
that consists of a layer of CNOTs followed by a sequence of Z and X phase gadgets. To bolster the claim that phase gadgets 
are sufficient for constructing QML ans\"{a}tze, we shall convert commonly used gates into phase gadgets. In particular, we first show that the CNOT, Hadamard and Clifford-T gates can be written using gadgets, as these gates are approximately universal and so can be used to construct any unitary, to any arbitary precision \autocite{universal1, universal2}. We then show that all three sets: phase gadgets + Hadamards, phase gadgets + X rotations and Pauli gadgets form a universal set of gates for quantum computation.

As defined in section \ref{zx def}, a Hadamard gate can be expressed as three $\pm\frac{\pi}{2}$ rotations of the same sign with alternating bases.
\begin{align*}
	\input{zx/h.tikz}=\input{zx/h_def_1.tikz} &= \input{zx/h_def_2.tikz} \\
                  =\input{zx/h_def_3.tikz} &= \input{zx/h_def_4.tikz}
\end{align*}
 This allows us to convert any circuit consisting of only phase gadgets and Hadamards to circuits consisting of only phase gadgets and X rotations. Furthermore, the X rotation gate is a special case of a Pauli gadget so the conversion from phase gadgets and X rotations to Pauli gadgets is also done. Now it suffices to show that phase gadgets + Hadamard circuits are indeed universal.

The controlled-Z (CZ) gate can be constructed using two Hadamard gates with a CNOT gate in between. Using properties of the Hadamard, it can be rewritten as 
$$  \input{zx/cz.tikz} \eq{h} \input{convert_all/cz1.tikz} $$
$$ \eq{h def} \input{convert_all/cz2.tikz}  \eq{$\star$} \input{convert_all/cz3.tikz} \eq{f} \input{convert_all/cz4.tikz} $$ 
where the equality $\star$ holds because 
$$ \input{convert_all/cz_lem1.tikz}  \eq{h} \input{convert_all/cz_lem2.tikz} \eq{h def, f} \input{convert_all/cz_lem3.tikz} $$ 
$$ \eq{$\pi$} \input{convert_all/cz_lem4.tikz}  \eq{f, scalar} \input{convert_all/cz_lem5.tikz}. $$ 
Conversely, a CNOT gate can also be constructed using two Hadamard gates with a CZ gate in between, which are both expressible using phase gadgets.
$$ \input{zx/cnot.tikz} \eq{h} \input{convert_all/cx1.tikz} $$
$$ = \input{convert_all/cx2.tikz} \eq{h def} \input{convert_all/cx3.tikz} $$ 
A more compact way to write the CNOT gate is to use a Pauli gadget, a Z rotation and a X rotation. 
$$ \input{convert_all/cx4.tikz} $$ 
\subsection{Other Gates}
Although we have shown universality, it would be useful to find compact representations of commonly used gates in quantum machine learning. Three such gates are $CU_1$, $CR_z$ and $CR_x$, 2-qubit gates used for entanglement.

\begin{center}
\includegraphics[scale=1]{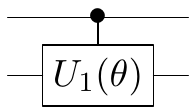}
\includegraphics[scale=1]{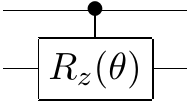}
\includegraphics[scale=1]{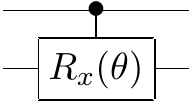}  
\end{center}
\vspace{5mm}

$$
CU_1(\theta) = 
\begin{pmatrix}
1 &  &  & \\
 & 1 &  & \\
 &  & 1 & \\
 &  &  & e^{i\theta }
\end{pmatrix} \qquad
CR_z(\theta) = 
\displaystyle \begin{pmatrix}
1 &  &  & \\
 & 1 &  & \\
 &  & e^{-i\theta/2 }& \\
 &  &  & e^{i\theta/2 }
\end{pmatrix}
$$
$$
CR_x(\theta) = 
\footnotesize
\begin{pmatrix}
1 &  &  & \\
 & 1 &  & \\
 &  & \cos(\theta/2) & -i\sin(\theta/2)\\
 &  & -i\sin(\theta/2) & \cos(\theta/2)
\end{pmatrix}
$$
Although the $U_1$ and the $R_z$ gates are considered the same up to global phase, the $CU_1$ and $CR_z$ are not.  
Fortunately, they are both diagonal matrices, so they can be expressed as a combination of Z phase gadgets. 

\vspace{5mm}
$$CU_1(\theta) = \left\llbracket \input{convert_all/cu1.tikz} \right\rrbracket$$
\vspace{2mm}
\[
\left.
\begin{aligned}
\ket{00} \rightarrow e^{-i\theta/2}e^{+i\theta/2}e^{-i\theta/2} \ket{00} \\
\ket{01} \rightarrow e^{-i\theta/2}e^{-i\theta/2}e^{+i\theta/2} \ket{01} \\
\ket{10} \rightarrow e^{+i\theta/2}e^{-i\theta/2}e^{-i\theta/2} \ket{10} \\
\ket{11} \rightarrow e^{+i\theta/2}e^{+i\theta/2}e^{+i\theta/2} \ket{11}
\end{aligned}
\right\}
= e^{-i\theta/2} 
\begin{pmatrix}
1 &  &  & \\
 & 1 &  & \\
 &  & 1 & \\
 &  &  & e^{i\theta }
\end{pmatrix}\]

\vspace{5mm}
$$CR_z(2\theta) = \left\llbracket \input{convert_all/crz.tikz} \right\rrbracket$$
\vspace{2mm}
\[
\left.
\begin{aligned}
\ket{00} \rightarrow e^{+i\theta/2}e^{-i\theta/2} \ket{00} \\
\ket{01} \rightarrow e^{-i\theta/2}e^{+i\theta/2} \ket{01} \\
\ket{10} \rightarrow e^{-i\theta/2}e^{-i\theta/2} \ket{10} \\
\ket{11} \rightarrow e^{+i\theta/2}e^{+i\theta/2} \ket{11}
\end{aligned}
\right\}
= 
\begin{pmatrix}
1 &  &  & \\
 & 1 &  & \\
 &  & e^{-i\theta }& \\
 &  &  & e^{i\theta }
\end{pmatrix}\]
%Analogous to the CNOT/CX, the Controlled RX gate can be expressed by prepending and appending a Hadamard gate on the controlled RZ gate.   
\vspace{5mm}
$$ CR_x(2\theta) = \left\llbracket \input{convert_all/crx.tikz} \right\rrbracket $$
\vspace{2mm}
\[
\left.
\begin{aligned}
\ket{00} &\rightarrow \ket{0+} \rightarrow \ket{00} \\
\ket{01} &\rightarrow \ket{0-} \rightarrow \ket{01} \\
\ket{10} &\rightarrow \ket{10} + \ket{11} \rightarrow e^{-i\theta}\ket{10} + e^{i\theta}\ket{11} \\ 
	 &\rightarrow (e^{-i\theta}+e^{i\theta})\ket{10} + (e^{-i\theta}-e^{i\theta})\ket{11} \\
	 &= \cos(\theta)\ket{10} - i\sin(\theta)\ket{11} \\
\ket{11} &\rightarrow \ket{10} - \ket{11} \rightarrow e^{-i\theta}\ket{10} - e^{i\theta}\ket{11} \\ 
	 &\rightarrow (e^{-i\theta}-e^{i\theta})\ket{10} + (e^{-i\theta}+e^{i\theta})\ket{11} \\
	 &= -i\sin(\theta)\ket{10} + \cos(\theta)\ket{11} \\
\end{aligned}
\right\}
= 
\begin{pmatrix}
1 &  &  & \\
 & 1 &  & \\
 &  & \cos(\theta) & -i\sin(\theta)\\
 &  & -i\sin(\theta) & \cos(\theta)
\end{pmatrix}
\]
The $CR_x$ gate can be expressed by prepending and appending a Hadamard gate on the $CR_z$ gate. 
Curiously, the $CZ$ gate is a special case of the $CU_1$ gate rather than the $CR_z$ gate, 
and the $CNOT$ gate (also known as the $CX$ gate) does not quite correspond to the $CR_x$ gate; 
this is because the X gate and $R_x(\pi)$ are only equal up to a scalar. 

\begin{align*}
  X \triangleq \begin{pmatrix}0&1\\1&0\end{pmatrix} \neq 
  \begin{pmatrix}0&-i\\-i&0\end{pmatrix} \triangleq R_x(\pi) \\
  Y \triangleq \begin{pmatrix}0&-i\\i&0\end{pmatrix} \neq 
  \begin{pmatrix}0&-1\\1&0\end{pmatrix} \triangleq R_y(\pi) \\  
  Z \triangleq \begin{pmatrix}1&0\\0&-1\end{pmatrix} \neq 
  \begin{pmatrix}-i&0\\0&i\end{pmatrix} \triangleq R_z(\pi) 
\end{align*}
 
\newpage
\vspace*{\fill}
%\begin{center}
\begin{center}\subsection*{Sussman Attains Enlightenment}\end{center}
\noindent
{ \itshape
In the days when Sussman was a novice, Minsky once came to him as he sat hacking at the PDP-6.

\vspace{3mm}
\noindent
Minsky: ``What are you doing?" \\
Sussman: ``I am training a randomly wired neural net to play Tic-tac-toe." \\
Minsky: ``Why is the net wired randomly?" \\
Sussman: ``I do not want it to have any preconceptions of how to play." 

\vspace{3mm}
\noindent
Minsky then shut his eyes. 

\vspace{3mm}
\noindent
Sussman: ``Why do you close your eyes?" \\
Minsky: ``So that the room will be empty."

\vspace{3mm}
\noindent
At that moment, Sussman was enlightened.
}

\vspace{10mm}
What Minsky means is that wiring the circuit randomly doesn't eliminate all preconceptions, we just don't know what the preconceptions are. This philosophy also applies to quantum circuits.
Wiring your circuit randomly does not make it better. You just have less understanding of it.

In the same way that we cannot wire up our neural network randomly, we cannot escape the question of ``what makes a VQC good for learning" by simply using circuits we do not know how to analyse, or generating the circuits randomly.
%This project attempts to explore and study as wide of class of variational quantum circuits as possible; this is a trade-off between the generality of our results and the calculi available to us.

\vspace*{\fill}

%% file: figures/phase_gadget/pg_4_2.tikz
\begin{tikzpicture}[scale=0.6, baseline={([yshift=-0.5ex]current bounding box.center)}]
	\begin{pgfonlayer}{nodelayer}
		\node [style=none] (0) at (0, 1) {};
		\node [style=none] (8) at (4, 1) {};
		\node [style={Z_med}] (9) at (2, 1) {$\theta$};
	\end{pgfonlayer}
	\begin{pgfonlayer}{edgelayer}
		\draw (0.center) to (9);
		\draw (9) to (8.center);
	\end{pgfonlayer}
\end{tikzpicture}

%% file: pauli.tex
\def\YZ{\textcolor{zxblue}{-}\nodepart{two}\textcolor{zxgreen}{-}}
\def\YX{\textcolor{zxblue}{-}\nodepart{two}\textcolor{zxred}{-}}
\begin{figure}
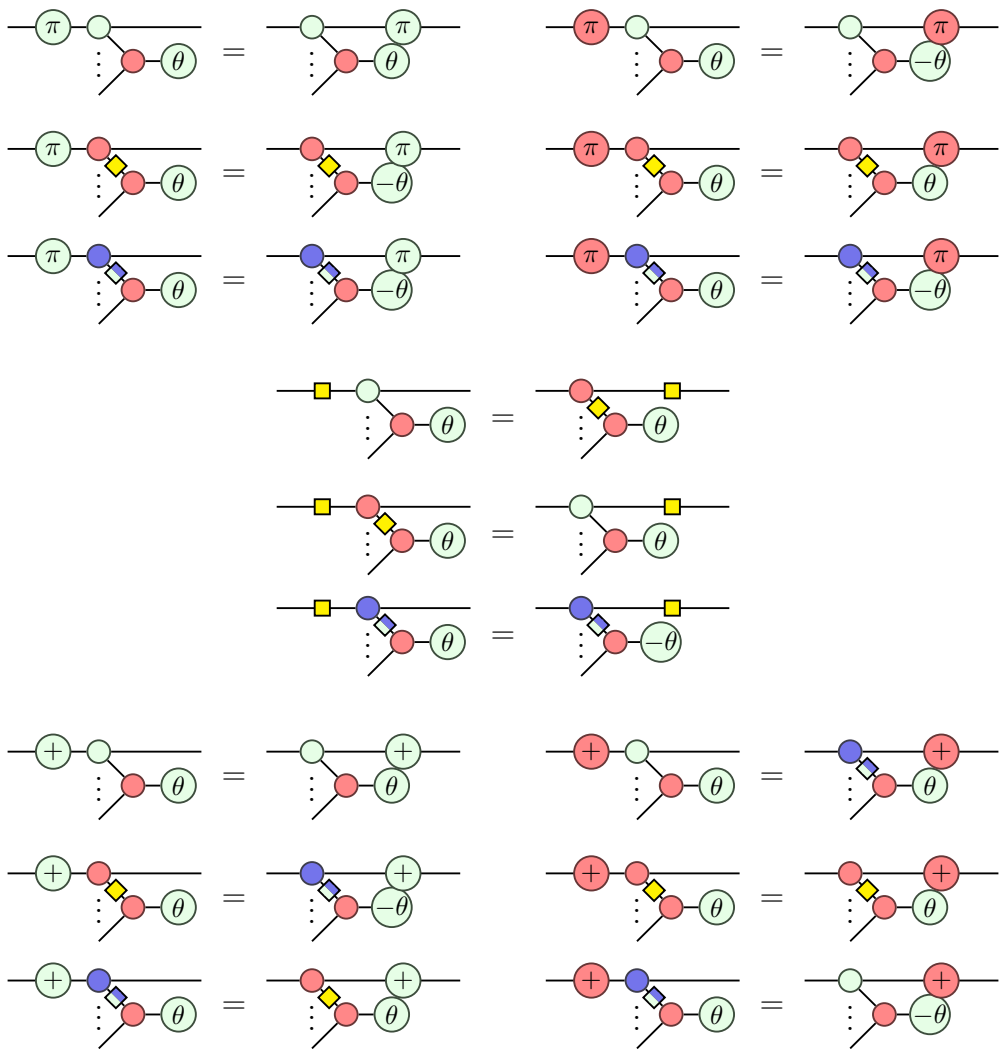

\begin{center}
$$ \input{pauli_comm/p11.tikz} = \input{pauli_comm/p12.tikz} \qquad
 \input{pauli_comm/p21.tikz} = \input{pauli_comm/p22.tikz} $$
$$ \input{pauli_comm/p13.tikz} = \input{pauli_comm/p14.tikz} \qquad
 \input{pauli_comm/p23.tikz} = \input{pauli_comm/p24.tikz} $$
$$ \input{pauli_comm/p15.tikz} = \input{pauli_comm/p16.tikz} \qquad
 \input{pauli_comm/p25.tikz} = \input{pauli_comm/p26.tikz} $$

$$ \input{pauli_comm/p31.tikz} = \input{pauli_comm/p32.tikz} $$
$$ \input{pauli_comm/p33.tikz} = \input{pauli_comm/p34.tikz} $$
$$ \input{pauli_comm/p35.tikz} = \input{pauli_comm/p36.tikz} $$

$$ \input{pauli_comm/p41.tikz} = \input{pauli_comm/p42.tikz} \qquad
 \input{pauli_comm/p51.tikz} = \input{pauli_comm/p52.tikz} $$
$$ \input{pauli_comm/p43.tikz} = \input{pauli_comm/p44.tikz} \qquad
 \input{pauli_comm/p53.tikz} = \input{pauli_comm/p54.tikz} $$
$$ \input{pauli_comm/p45.tikz} = \input{pauli_comm/p46.tikz} \qquad
 \input{pauli_comm/p55.tikz} = \input{pauli_comm/p56.tikz} $$
\end{center}
\caption{Rules for passing Clifford gates through Pauli gadgets.}
\end{figure}
\begin{figure}
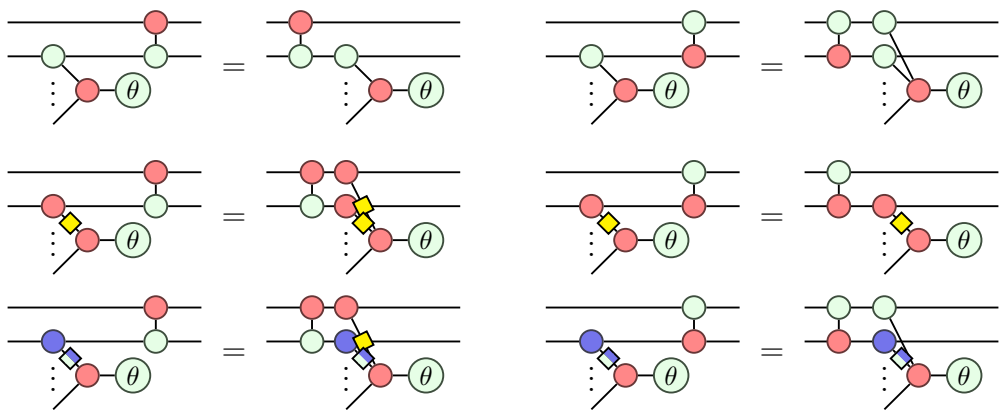

\begin{center}
$$ \input{pauli_comm/q11.tikz} = \input{pauli_comm/q12.tikz} \qquad
 \input{pauli_comm/q13.tikz} = \input{pauli_comm/q14.tikz} $$
$$ \input{pauli_comm/q21.tikz} = \input{pauli_comm/q22.tikz} \qquad
 \input{pauli_comm/q23.tikz} = \input{pauli_comm/q24.tikz} $$
$$ \input{pauli_comm/q31.tikz} = \input{pauli_comm/q32.tikz} \qquad
 \input{pauli_comm/q33.tikz} = \input{pauli_comm/q34.tikz} $$
\end{center}
\caption{Rules for passing CNOT gates through Pauli gadgets.}
\end{figure}

%% file: figures/eu/eq1.tikz
\begin{tikzpicture}[scale=0.4, baseline={([yshift=-0.5ex]current bounding box.center)}]
	\begin{pgfonlayer}{nodelayer}
		\node [style=none] (0) at (5.5, 0) {};
		\node [style=none] (4) at (0, 0) {};
		\node [style={X_med}] (16) at (1.25, 0) {$a_1$};
		\node [style={X_med}] (17) at (4.25, 0) {$a_3$};
		\node [style={Z_med}] (18) at (2.75, 0) {$a_2$};
	\end{pgfonlayer}
	\begin{pgfonlayer}{edgelayer}
		\draw (4.center) to (16);
		\draw (17) to (0.center);
		\draw (16) to (18);
		\draw (18) to (17);
	\end{pgfonlayer}
\end{tikzpicture}

%% file: figures/eu/eq2.tikz
\begin{tikzpicture}[scale=0.4, baseline={([yshift=-0.5ex]current bounding box.center)}]
	\begin{pgfonlayer}{nodelayer}
		\node [style=none] (0) at (5.5, 0) {};
		\node [style=none] (4) at (0, 0) {};
		\node [style={Z_med}] (16) at (1.25, 0) {$b_1$};
		\node [style={Z_med}] (17) at (4.25, 0) {$b_3$};
		\node [style={X_med}] (18) at (2.75, 0) {$b_2$};
	\end{pgfonlayer}
	\begin{pgfonlayer}{edgelayer}
		\draw (4.center) to (16);
		\draw (17) to (0.center);
		\draw (16) to (18);
		\draw (18) to (17);
	\end{pgfonlayer}
\end{tikzpicture}

%% file: figures/eu/eu1.tikz
\begin{tikzpicture}[baseline={([yshift=-0.5ex]current bounding box.center)}]
	\begin{pgfonlayer}{nodelayer}
		\node [style={X_med}] (0) at (4, 21) {$a$};
		\node [style={X_med}] (1) at (4, 20) {$b$};
		\node [style={X_med}] (2) at (4, 19) {$c$};
		\node [style={Z_med}] (4) at (5, 21) {$d$};
		\node [style={Z_med}] (5) at (5, 20) {$e$};
		\node [style={Z_med}] (6) at (5, 19) {$f$};
		\node [style=none] (8) at (3, 21) {};
		\node [style=none] (9) at (3, 20) {};
		\node [style=none] (10) at (3, 19) {};
		\node [style=none] (12) at (8, 21) {};
		\node [style=none] (13) at (8, 20) {};
		\node [style=none] (14) at (8, 19) {};
		\node [style={X_med}] (15) at (6, 21) {$g$};
		\node [style={X_med}] (16) at (6, 20) {$h$};
		\node [style={X_med}] (17) at (6, 19) {$i$};
		\node [style={Z_med}] (18) at (7, 21) {$j$};
		\node [style={Z_med}] (19) at (7, 20) {$k$};
		\node [style={Z_med}] (20) at (7, 19) {$l$};
	\end{pgfonlayer}
	\begin{pgfonlayer}{edgelayer}
		\draw (8.center) to (0);
		\draw (0) to (4);
		\draw (4) to (12.center);
		\draw (9.center) to (1);
		\draw (1) to (5);
		\draw (5) to (13.center);
		\draw (10.center) to (2);
		\draw (6) to (2);
		\draw (6) to (14.center);
	\end{pgfonlayer}
\end{tikzpicture}

%% file: figures/eu/eu2.tikz
\begin{tikzpicture}[baseline={([yshift=-0.5ex]current bounding box.center)}]
	\begin{pgfonlayer}{nodelayer}
		\node [style={Z_med}] (0) at (4, 21) {$a'$};
		\node [style={Z_med}] (1) at (4, 20) {$b'$};
		\node [style={Z_med}] (2) at (4, 19) {$c'$};
		\node [style={X_med}] (4) at (5, 21) {$d'$};
		\node [style={X_med}] (5) at (5, 20) {$e'$};
		\node [style={X_med}] (6) at (5, 19) {$f'$};
		\node [style=none] (8) at (3, 21) {};
		\node [style=none] (9) at (3, 20) {};
		\node [style=none] (10) at (3, 19) {};
		\node [style=none] (12) at (8, 21) {};
		\node [style=none] (13) at (8, 20) {};
		\node [style=none] (14) at (8, 19) {};
		\node [style={Z_med}] (15) at (6, 21) {$g'$};
		\node [style={Z_med}] (16) at (6, 20) {$h'$};
		\node [style={Z_med}] (17) at (6, 19) {$i'$};
		\node [style={Z_med}] (18) at (7, 21) {$j$};
		\node [style={Z_med}] (19) at (7, 20) {$k$};
		\node [style={Z_med}] (20) at (7, 19) {$l$};
	\end{pgfonlayer}
	\begin{pgfonlayer}{edgelayer}
		\draw (8.center) to (0);
		\draw (0) to (4);
		\draw (4) to (12.center);
		\draw (9.center) to (1);
		\draw (1) to (5);
		\draw (5) to (13.center);
		\draw (10.center) to (2);
		\draw (6) to (2);
		\draw (6) to (14.center);
	\end{pgfonlayer}
\end{tikzpicture}

%% file: figures/eu/eu3.tikz
\begin{tikzpicture}[baseline={([yshift=-0.5ex]current bounding box.center)}]
	\begin{pgfonlayer}{nodelayer}
		\node [style={Z_med}] (0) at (4, 21) {$a'$};
		\node [style={Z_med}] (1) at (4, 20) {$b'$};
		\node [style={Z_med}] (2) at (4, 19) {$c'$};
		\node [style={X_med}] (4) at (5, 21) {$d'$};
		\node [style={X_med}] (5) at (5, 20) {$e'$};
		\node [style={X_med}] (6) at (5, 19) {$f'$};
		\node [style=none] (8) at (3, 21) {};
		\node [style=none] (9) at (3, 20) {};
		\node [style=none] (10) at (3, 19) {};
		\node [style=none] (12) at (7.5, 21) {};
		\node [style=none] (13) at (7.5, 20) {};
		\node [style=none] (14) at (7.5, 19) {};
		\node [style={Z_long}, font={\tiny}, minimum width=1cm] (15) at (6.5, 21) {$g'+j$};
		\node [style={Z_long}, font={\tiny}, minimum width=1cm] (16) at (6.5, 20) {$h'+k$};
		\node [style={Z_long}, font={\tiny}, minimum width=1cm] (17) at (6.5, 19) {$i'+l$};
	\end{pgfonlayer}
	\begin{pgfonlayer}{edgelayer}
		\draw (8.center) to (0);
		\draw (0) to (4);
		\draw (4) to (12.center);
		\draw (9.center) to (1);
		\draw (1) to (5);
		\draw (5) to (13.center);
		\draw (10.center) to (2);
		\draw (6) to (2);
		\draw (6) to (14.center);
	\end{pgfonlayer}
\end{tikzpicture}

%% file: zoo.tex
\chapter{The Zoo of QML Ans\"{a}tze}
Using the techniques we have developed in the previous chapters, we are now ready to simplify 
and analyse QML ans\"{a}tze. We have selected the nineteen circuits used 
%in \textit{Expressibility and entangling capability of parameterized quantum circuits for hybrid quantum-classical algorithms} 
by Sim et al.\ \autocite{sim2019expressibility}, many of which were from or inspired by past studies.
The four main techniques used for simplification are: spider fusion, phase gadget addition, 
Euler decomposition and single rotation decomposition. These four techniques are applied 
during both inter and intra layer simplification.

\begin{enumerate}
	\item \textbf{Spider fusion}: by writing multi-qubit gates using only X, Y and Z spiders, it is clear which 
gates commute with another. For example, the control Z rotation gates commute with each other, 
whereas the control X rotation gates and CNOT gates commute with each other only when their 
control or target legs are aligned. (See section \ref{zx}) 
	\item \textbf{Phase gadget addition}: by writing multi-qubit gates such as the controlled rotation gates
in terms of phase gadgets and rotations, the gates can partially or completely simplified. 
(See section \ref{phase_gadget})
	\item \textbf{Euler Decomposition}: Any combination of rotations gates on a qubit can be written in terms 
of three rotation gates. (See section \ref{eu}) 
	\item \textbf{Single rotation decomposition}: in practice, Z rotations are easier to perform than X rotations and Y rotations on a quantum computer \autocite{mckay2017efficient}. By strategically decomposing 
the X and Y rotation gates, we can both reduce the number of gates and use cheaper rotations. (See section \ref{y})
\end{enumerate}

The following analysis takes the original circuit, converts it to ZX calculus then applies progressive 
simplifications. A \textit{dependency graph} is included to show precisely how to obtain the new parameters 
from the original parameters. Writing out the dependency graph allows automatic differentiation on the new circuit, in terms of the old parameters; this allows to simulate quantum machine learning with the original circuit, 
using the new, simplified circuit.

Although the commutative properties of the Pauli gadget have not been used in this analysis, they can 
be useful for analysing future ans\"{a}tze that are designed directly using phase and Pauli gadgets.

\newpage
{\tabulinesep=1.2mm
\section*{Circuit 1}
\begin{center}
\raisebox{-0.5\height}{\includegraphics[scale=1.5]{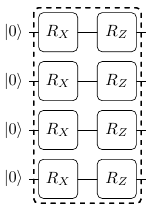}} 
\input{ansatz/new/circ1.tikz}
\input{ansatz/new/circ1_simp.tikz}
\scalebox{0.8}{\input{ansatz/new/circ1_graph.tikz}}
\end{center}
\textbf{Note}: This circuit has already been studied in Euler Decomposition section.
\section*{Circuit 2}
\begin{center}
\raisebox{-0.5\height}{\includegraphics[scale=1.5]{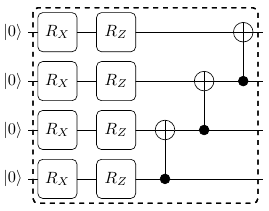}} 
\input{ansatz/new/circ2.tikz}
\end{center}
\textbf{Note}: This circuit has already been studied in the Normal Form section.
\section*{Circuit 3}
\begin{center}
{\renewcommand{\arraystretch}{1.8}
\begin{tabu}{cc}
	\raisebox{-0.5\height}{\includegraphics[scale=1.5]{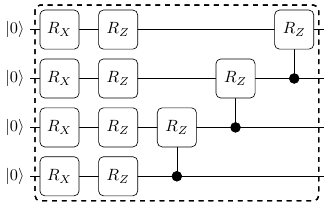}} &
	\input{ansatz/new/circ3_graph.tikz} \\
	\input{ansatz/new/circ3.tikz} &
	\input{ansatz/new/circ3_simp.tikz} \\ 
\end{tabu}} 
\end{center}

\newpage
\section*{Circuit 4}
\begin{center}
\begin{tabu}{cc}
	\raisebox{-0.5\height}{\includegraphics[scale=1.5]{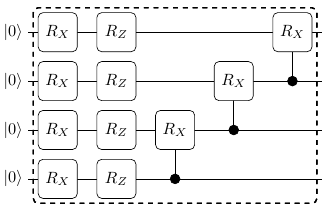}} &
	\input{ansatz/new/circ4_graph.tikz} \\
	\input{ansatz/new/circ4.tikz} &
	\input{ansatz/new/circ4_simp.tikz} \\
\end{tabu}
\end{center}
\ctikzfig{ansatz/new/circ4_simp1}

\section*{Circuit 5}
\begin{center}\includegraphics[scale=1.5]{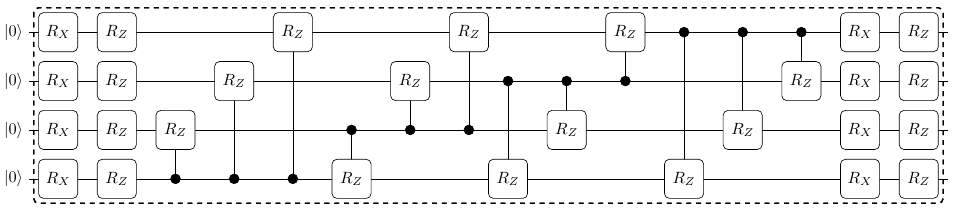}\end{center}
\ctikzfig{ansatz/new/circ5} 
\ctikzfig{ansatz/new/circ5_simp}  
\ctikzfig{ansatz/new/circ5_simp1}  
\ctikzfig{ansatz/new/circ5_simp2}  
\ctikzfig{ansatz/new/circ5_simp3}  
\ctikzfig{ansatz/new/circ5_simp4} 
\textbf{Note}: The new circuit has only 14 parameters in the repeating layer in comparison to 28 parameters in the original circuit.
\begin{center}\scalebox{0.9}{\input{ansatz/new/circ5_graph.tikz}}\end{center}
\section*{Circuit 6}
\begin{center}\includegraphics[scale=1.5]{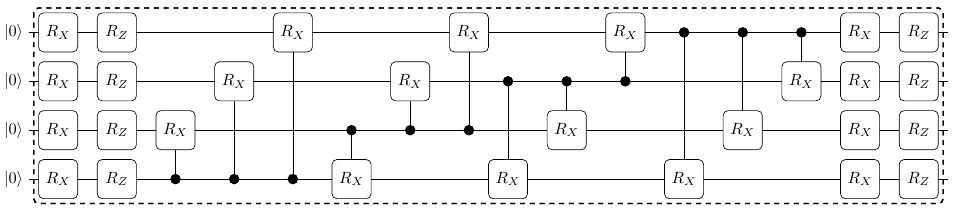}\end{center}
\ctikzfig{ansatz/new/circ6} 
\ctikzfig{ansatz/new/circ6_simp} 
\ctikzfig{ansatz/new/circ6_simp1} 
\ctikzfig{ansatz/new/circ6_simp2} 
\newpage
\vspace*{\fill}
\ctikzfig{ansatz/new/circ6_graph}
\vspace*{\fill}
\newpage
\section*{Circuit 7}
\resizebox{\textwidth}{!}{
\begin{tabu}{cc}
	\raisebox{-0.5\height}{\includegraphics[scale=1.5]{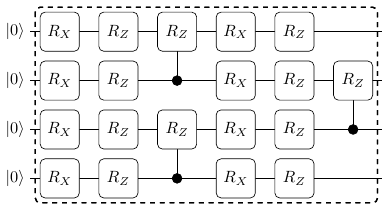}} &
	\input{ansatz/new/circ7_graph.tikz} \\
	\input{ansatz/new/circ7.tikz} &
	\input{ansatz/new/circ7_simp.tikz} \\ 
\end{tabu}}
\section*{Circuit 8}
\resizebox{\textwidth}{!}{
\begin{tabu}{cc}
	\raisebox{-0.5\height}{\includegraphics[scale=1.5]{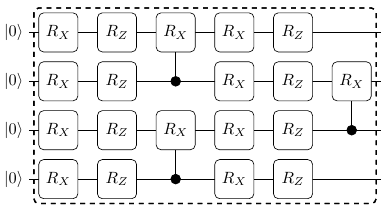}} &
	\input{ansatz/new/circ8.tikz} \\
	\input{ansatz/new/circ8_simp.tikz} & 
	\input{ansatz/new/circ8_simp1.tikz}
\end{tabu}}\\
\vspace{0.5cm}
\ctikzfig{ansatz/new/circ8_graph}
\section*{Circuit 9}
\begin{center}
\includegraphics[scale=1.5]{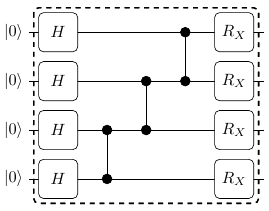}
\input{ansatz/new/circ9_simp1.tikz}
\input{ansatz/new/circ9_simp2.tikz}
\end{center}
\section*{Circuit 10}
\begin{center}
\includegraphics[scale=1.5]{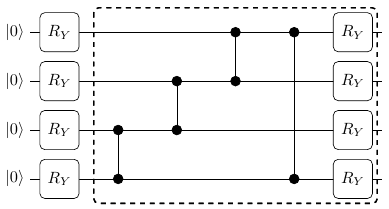} \\
\vspace{1.5mm}
\input{ansatz/new/circ10.tikz}
\input{ansatz/new/circ10_simp.tikz}
\end{center}
\section*{Circuit 11}
\begin{center}
\begin{tabu}{cc}
	\raisebox{-0.5\height}{\includegraphics[scale=1.5]{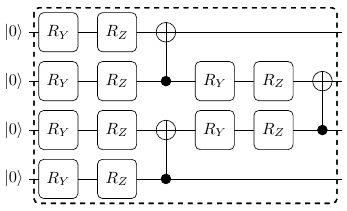}} &
	\input{ansatz/new/circ11_graph.tikz} \\
	\input{ansatz/new/circ11.tikz} &
	\input{ansatz/new/circ11_simp.tikz} \\
\end{tabu}
\end{center}
\section*{Circuit 12}
\begin{center}
\includegraphics[scale=1.5]{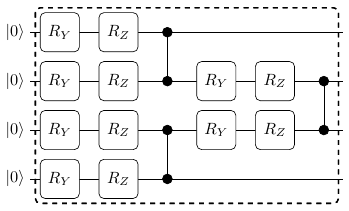} \\
\vspace{1.5mm}
\input{ansatz/new/circ12.tikz}
\input{ansatz/new/circ12_simp.tikz}
\end{center}
\section*{Circuit 13}
\begin{center}
\includegraphics[scale=1.5]{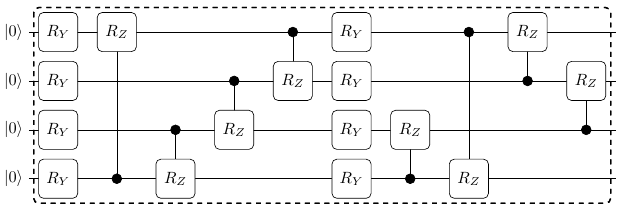}
\end{center}
\ctikzfig{ansatz/new/circ13} 
\ctikzfig{ansatz/new/circ13_simp} 
\ctikzfig{ansatz/new/circ13_simp1} 
\ctikzfig{ansatz/new/circ13_graph}
\textbf{Note}: The repeating layer of this circuit can be split into two structurally identical halves.
% Repeating layers
\section*{Circuit 14}
\begin{center}
\includegraphics[scale=1.5]{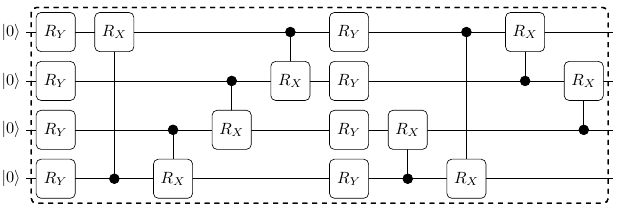}\end{center}
\ctikzfig{ansatz/new/circ14} 
\ctikzfig{ansatz/new/circ14_simp} 
\begin{center}\scalebox{0.9}{\input{ansatz/new/circ14_graph.tikz}}\end{center}
\section*{Circuit 15}
\begin{center}
\raisebox{-0.5\height}{\includegraphics[scale=1.5]{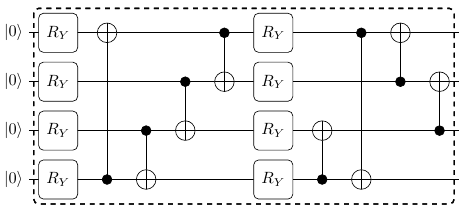}}
\input{ansatz/new/circ15.tikz} 
\end{center}
\textbf{Note}: The effect of CNOT layers on Y spiders cannot be described using $GF(2)$, 
and its periodicity has yet to be studied.
\section*{Circuit 16}
\begin{center}
\begin{tabu}{cc}
	\raisebox{-0.5\height}{\includegraphics[scale=1.45]{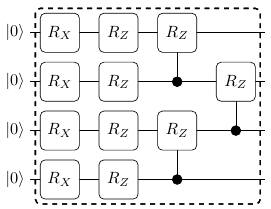}} & 
	\input{ansatz/new/circ16_graph.tikz} \\
	\input{ansatz/new/circ16.tikz} &
	\input{ansatz/new/circ16_simp.tikz} \\
\end{tabu}
\end{center}
\textbf{Note}: Circuit 16 is identical as circuit 13 after simplification.
\section*{Circuit 17}
\begin{center}
\begin{tabu}{cc}
	\raisebox{-0.5\height}{\includegraphics[scale=1.45]{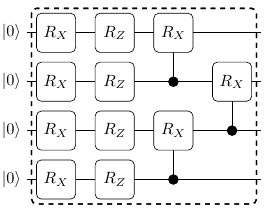}} & 
	\input{ansatz/new/circ17_graph.tikz} \\
	\input{ansatz/new/circ17.tikz} &
	\input{ansatz/new/circ17_simp.tikz} \\
\end{tabu}
\end{center}
\section*{Circuit 18}
\resizebox{\textwidth}{!}{
\begin{tabu}{cc}
	\raisebox{-0.5\height}{\includegraphics[scale=1.5]{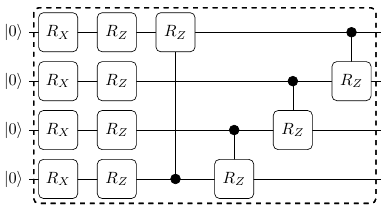}} & 
	\input{ansatz/new/circ18_graph.tikz} \\
	\input{ansatz/new/circ18.tikz} &
	\input{ansatz/new/circ18_simp.tikz} \\
\end{tabu}}
\textbf{Note}: Circuit 3 and circuit 16 are special cases of circuit 18 where $i' = 0$.
\section*{Circuit 19}
\resizebox{\textwidth}{!}{
\begin{tabu}{cc}
	\raisebox{-0.5\height}{\includegraphics[scale=1.5]{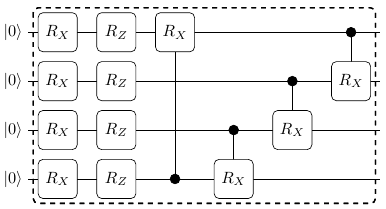}} &
	\input{ansatz/new/circ19_graph.tikz} \\
	\input{ansatz/new/circ19.tikz} &
	\input{ansatz/new/circ19_simp.tikz} \\
\end{tabu}}
}

%% file: chapter3.tex
\chapter{Optimising QML Ans\"{a}tze}
In this chapter we describe a simplification algorithm for ansätze constructed from repeating layers
of mixed Z and X phase gadgets. Our approach leverages a combinatorial description (in terms of $GL(n, 2)$
matrices) of the effect of composing such ansätze with CNOT circuits, and uses simulated annealing to obtain
approximate solutions to a corresponding optimisation problem.
\section{Background}
In the second chapter we explored how most quantum circuits used as ans\"{a}tze for quantum machine learning can be expressed using phase gadgets and Hadamard gates alone, and ans\"{a}tze constructed from CNOT gates, Z rotation gates and X rotation gates yield a normal form consisting of repeating layers of Z and X phase gadgets. In this chapter, our algorithm synthesises circuits designed using repeating Z and X phase gadgets, which can be efficiently synthesised using CNOT gates, Z rotation gates and X rotation gates. 

One way to simplify phase gadgets in ZX calculus is to apply the commutation properties between the phase gadget and the CNOT. The work by Meijer-van de Griend and Duncan \autocite{meijer2020architecture} applied this commutation property for Z phase gadgets to synthesise phase polynomials. This is done by first describing the phase polynomial using Z phase gadgets, which can be stored as a binary matrix where each column represents a phase gadget. The CNOT placements are determined by the bitwise operations used to eliminate 1 entries in the matrix.

The same idea can be applied to circuits consisting of Z and X gadgets, as shown in the following example:
\[\input{routine/demo1.tikz}\]
\[ =\input{routine/demo2.tikz} \]
\section{Abstraction}
Circuits on $n$ qubits composed of Z and X phase gadgets can be represented as a sequence of entries, each entry detailing the basis of the gadget (Z or X), the phase (angle), and the positions of its legs (an $n$-dimensional binary vector, i.e. a vector in $GF(2)^n$).
Our simplification algorithm does not alter the relative ordering of the gadgets or their angles, so a circuit on $n$ qubit with $d = d_z+d_x$ gadgets can be modelled by 3 pieces of data:
\begin{itemize}
    \item An $n \times d_z$ binary matrix $L_Z \in GF(2)^{n \times d_z}$ with columns encoding the leg positions of the $d_z$ Z phase gadgets.
    \item An $n \times d_x$ binary matrix $L_X \in GF(2)^{n \times d_x}$ with columns encoding the leg positions of the $d_x$ X phase gadgets.
    \item A sequence $S \in \left([0,2\pi)\times\{\text{Z},\text{X}\}\right)^{d}$ (with $d_z$ Z entries and $d_x$ X entries) encoding the basis and angles for the gadgets (not changed by the simplification algorithm).
\end{itemize}
The following example (with $n=5$, $d=5$, $d_x=2$ and $d_y=3$) shows how the $L_Z$ and $L_X$ matrices used by the algorithm are derived from the sequence of gadgets (angles omitted for clarity):
\begin{center}
    \texttt{[['Z', (1,1,0)], ['X', (1,1,1)], ['X', (1,1,0)], ['Z', (1,0,0)], ['Z', (1,1,0)]]} \\
    \vspace{0.5cm}
    $L_Z = \begin{pmatrix}1&1&1\\1&0&1\\0&0&0\end{pmatrix}$ \quad
    $L_X = \begin{pmatrix}1&1\\1&1\\1&0\end{pmatrix}$ \\
\end{center}
The algorithm is based on the action of CNOTs on circuits composed entirely of one kind of phase gadgets (either Z or X), as specified by the commutation rules above.
Specifically, it is based on the way in which a circuit of CNOTs changes the legs of a phase gadget as the CNOTs are commuted from the right of the phase gadget to its left (noting that the circuit of CNOTs is left unchanged by the commutation process).

Because the commutation of a circuit $\mathcal{C}$ of CNOTs through a phase gadget is an invertible procedure (inverted by using $\mathcal{C}^\dagger$), its action on the legs of a gadget must correspond to an invertible binary matrix in $GL(n,2)$: we write $h_z(\mathcal{C}) \in GL(n,2)$ for the matrix defining the action on Z phase gadgets and $h_x(\mathcal{C}) \in GL(n,2)$ for the matrix defining the action on X phase gadgets.
Given a circuit $\mathcal{C}$, the matrices $h_z(\mathcal{C})$ and $h_x(\mathcal{C})$ can be constructed by looking at the action of the CNOTs on binary vectors $GF(2)^n$ encoded by X and Z basis vectors respectively:

\vspace{.5cm}
\[
    h_z\left(\left\llbracket\input{cnots/cnot2.tikz}\right\rrbracket\right)
    =
    \begin{pmatrix}1&1\\0&1\end{pmatrix}_z
\]\[
    h_x\left(\left\llbracket\input{cnots/cnot1.tikz}\right\rrbracket\right)
    =
    \begin{pmatrix}1&0\\1&1\end{pmatrix}_x
\]
\vspace{.5cm}
\[
    h_z\left(\left\llbracket\input{cnots/cnot2big.tikz}\right\rrbracket\right)
    =
    \begin{pmatrix}	1&1&0\\0&1&1\\1&1&1 \end{pmatrix}_z
\]
\[
    h_x\left(\left\llbracket\input{cnots/cnot2bigflip.tikz}\right\rrbracket\right)
    =
    \begin{pmatrix}	0&1&1\\1&1&0\\1&1&1 \end{pmatrix}_x
    \vspace{.5cm}
\]

\begin{theorem} \label{inv_trans}
    The actions $h_z(\mathcal{C})$ and $h_x(\mathcal{C})$ of a CNOT circuit $\mathcal{C}$ are related by inverse transpose in $GL(n,2)$:
    $$
        h_x(\mathcal{C}) = \left(h_z(\mathcal{C})^T\right)^{-1}
    $$
\end{theorem}
\begin{proof}
	Let $\mathcal{C}=\prod^n_{i=1}\mathcal{C}_i$ be a CNOT circuit and $\{\mathcal{C}_i\}$ be 
	the individual CNOT gates. Since $h_z$ and $h_x$ are group homomorphisms, 
	the binary matrices $L_z$ and $L_x$ can be decomposed as 
	a product of the binary matrices, each representing a CNOT.
	
	$$	
	h_x(\mathsf{CNOT}(x,y))_{ij} = 
	\begin{cases}
	    1 \qquad i=j \\
	    1 \qquad i=x \text{ and } j=y \\
	    0 \qquad \text{otherwise}
	\end{cases}
	$$
	$$
	h_z(\mathsf{CNOT}(x,y))_{ij} = 
	\begin{cases}
	    1 \qquad i=j \\
	    1 \qquad i=y \text{ and } j=x \\
	    0 \qquad \text{otherwise}
	\end{cases}
	$$

	$$ L_z = h_z(\mathcal{C}) = h_z\left(\prod_i \mathcal{C}_i\right) = \prod_k h_z(\mathcal{C}_k) $$
	$$ L_x = h_x(\mathcal{C}) = h_x\left(\prod_i \mathcal{C}_i\right) = \prod_k h_x(\mathcal{C}_k) $$

	Looking at the definition of $h_z$ and $h_x$, it is evident that $h_z(\mathsf{CNOT}(x,y)) = h_x(\mathsf{CNOT}(y,x)) = h_x(\mathsf{CNOT}(x,y))^T$.	Since the CNOT gate is self inverting, so are the matrices resultant from applying $h_z$ and $h_x$: 
	$$ h_z(\mathsf{CNOT}(x,y))^2 = h_z(\mathsf{CNOT}(x,y)^2) = h_z(\mathsf{id}) = I $$
	$$ h_x(\mathsf{CNOT}(x,y))^2 = h_x(\mathsf{CNOT}(x,y)^2) = h_x(\mathsf{id}) = I $$

	Therefore we can conclude that  $h_z(G_i) = h_x(G_i)^{{-1}^T}$. 
	Now both the transpose and the inverse operations are contravariant, 
	but combine to make a covariant operation on matrices.
	Therefore
\begin{align*}
	\left(L_z^T\right)^{-1} &= \left(\left(L_{z_1}L_{z_2}\ldots L_{z_n}\right)^{T}\right)^{-1}\\ 
		     &= (L_{x_n}\ldots L_{x_2}L_{x_1})^{-1} \\
		     &= (L_{x_1}L_{x_2}\ldots L_{x_n}) = L_x.
\end{align*}
\end{proof}
\noindent Luckily for us, there are no restrictions on which $GL(n,2)$ matrices correspond to the action of CNOT circuits: for any invertible binary matrix in $GL(n, 2)$, we can construct a CNOT circuit with that action using row operations. Kissinger and Meijer-Van de Griend \autocite{kissinger2019cnot} proposed the Steiner-Gauss algorithm, which is an architecture-aware method of obtaining such a CNOT circuit. 

\begin{figure}
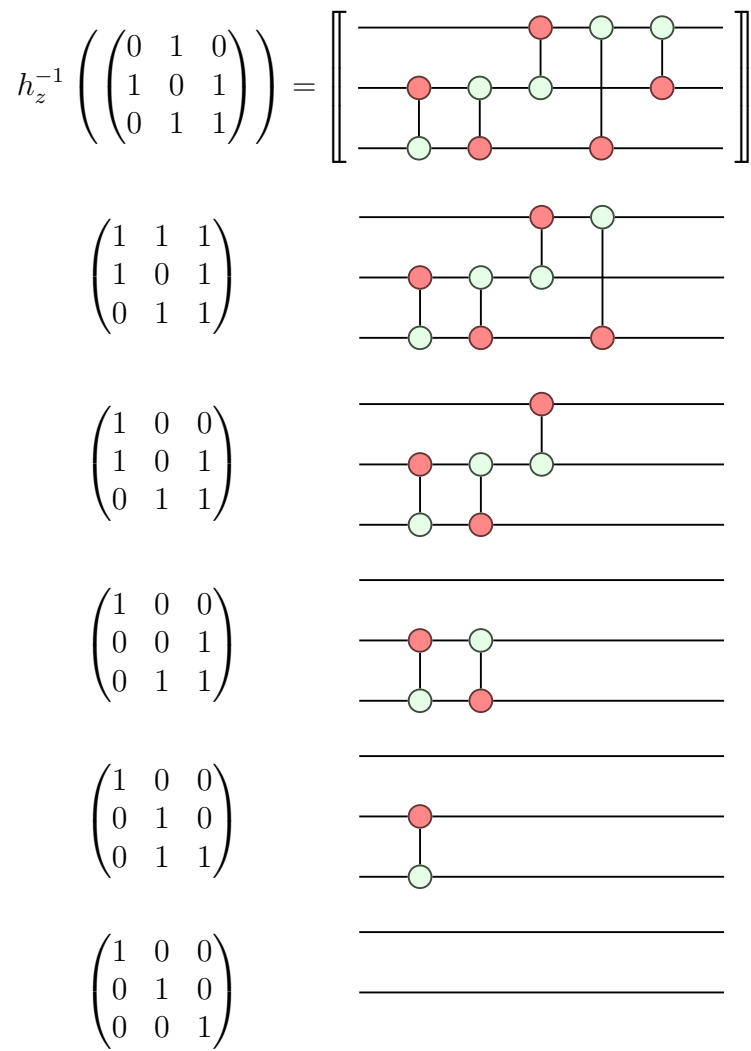

	\centering
\begin{align*}
	h^{-1}_z\left(\begin{pmatrix} 0&1&0\\1&0&1\\0&1&1 \end{pmatrix}\right) = 
	&\left\llbracket\input{cnots/c6.tikz}\right\rrbracket \\[1em]
	\begin{pmatrix} 1&1&1\\1&0&1\\0&1&1 \end{pmatrix} \hspace{1cm}& \quad\input{cnots/c5.tikz}\\[1em]
	\begin{pmatrix} 1&0&0\\1&0&1\\0&1&1 \end{pmatrix} \hspace{1cm}& \quad\input{cnots/c4.tikz}\\[1em]
	\begin{pmatrix} 1&0&0\\0&0&1\\0&1&1 \end{pmatrix} \hspace{1cm}& \quad\input{cnots/c3.tikz}\\[1em]
	\begin{pmatrix} 1&0&0\\0&1&0\\0&1&1 \end{pmatrix} \hspace{1cm}& \quad\input{cnots/c2.tikz}\\[1em]
	\begin{pmatrix} 1&0&0\\0&1&0\\0&0&1 \end{pmatrix} \hspace{1cm}& \quad\input{cnots/c1.tikz}
\end{align*}

\caption{An example of how an invertible matrix in $GF(2)$ can be converted to a CNOT circuit. The algorithm uses row operations to obtain the identity matrix, then adding the CNOTs as we unwind the recursion stack.}
\end{figure}

This bijective correspondence between $GL(n,2)$ matrices and CNOT circuits allows for a neat formulation of the action of CNOT circuits on $(n, d_z, d_x)$ Z/X phase gadget circuits as the following representation of $GL(n,2)$:
$$
\begin{matrix}
    h_z \oplus h_x: & GL(n,2) & \longrightarrow & GL\left(GF(2)^{n\times d_z} \oplus GF(2)^{n\times d_x}, 2\right)\\
    & C & \mapsto & C\oplus \left(C^T\right)^{-1}
\end{matrix}
$$

the Z and X phase gadgets are acted on independently, so we have omitted the interleaving and angle information $S$.
In this formulation, circuit simplification with the aim of reducing the total number of phase gadget legs results in the following optimization problem:
\begin{problem*}
    Given an $n\times d_z$ binary matrix $L_Z \in GF(2)^{n \times d_z}$ and an $n\times d_x$ binary matrix $L_X \in GF(2)^{n \times d_x}$, minimize the sum of the number of 1 entries in the two matrices $C \cdot L_Z$ and $\left(C^T\right)^{-1} \cdot L_X$ across all invertible binary matrices $C \in GL(n,2)$.
\end{problem*}
\noindent This is a discrete (non-differentiable) problem, which we believe to be hard in the general case.
As a consequence, we decided to employ simulated annealing \autocite{SA1,SA2} to obtain approximately optimal solutions.
Below is an example solution $C$ for the parameters used before:
\begin{center}
    \texttt{[['Z', (1,1,0)], ['X', (1,1,1)], ['X', (1,1,0)], ['Z', (1,0,0)], ['Z', (1,1,0)]]} \\
    \vspace{0.5cm}
    $L_Z = \begin{pmatrix}1&1&1\\1&0&1\\0&0&0\end{pmatrix}$ \quad
    $L_X = \begin{pmatrix}1&1\\1&1\\1&0\end{pmatrix}$ \\
    \vspace{0.5cm}
    $C = \begin{pmatrix}1&1&0\\0&1&0\\0&0&1\end{pmatrix}$ \qquad
    $C\cdot L_Z = \begin{pmatrix}0&1&0\\1&0&1\\0&0&0\end{pmatrix}$ \quad
    $\left(C^T\right)^{-1} \cdot L_X = \begin{pmatrix}1&1\\0&0\\1&0\end{pmatrix}$ \\
\end{center}
Given one such solution $C$, the optimized circuit is obtained as follows:
\begin{enumerate}
    \item Synthesise a CNOT circuit from $C$ using the Steiner-Gauss algorithm.
    \item Append the Z and X phase gadgets with legs specified by $C \cdot L_Z$ and $\left(C^T\right)^{-1} \cdot L_X$ respectively, using the original interleaving and angles stored in the sequence $S$.
    \item Synthesise and append a CNOT circuit from $C^\dagger$ using the Steiner-Gauss algorithm.
\end{enumerate}
The intermediate phase-gadget layer can be synthesised in an architecture-aware fashion using the procedure from \autocite{architecture}, or converting CNOT ladders into tree form to reduce depth.
The intermediate layer can also be repeated arbitrarily many times without repeating the initial and final CNOT blocks: this yields (linearly) compounded gains when dealing with ans\"{a}tze with several repeated layers, as common in QML \autocite{expressive,reinforcement,transfer}.

\section{Simulated Annealing}
Simulated annealing is a non--gradient-based, probabilistic optimisation method, inspired by the spontaneous emergence of ordered, energy minimizing configurations during annealing processes in metallurgy.
The essence of the algorithm is as follows:
\begin{enumerate}
    \item Start with a random initial point $C_0$ in the parameter space. In our case, start with a random matrix $C_0$ in parameter space $GL(n, 2)$.
    \item For each iteration $k=0,...,K-1$:
    \begin{enumerate}
        \item Compute a temperature $T$ according to a specified \emph{temperature schedule}.
        \item Randomly select a neighbour of the current point in configuration space. In our case, randomly select a matrix $C_{new}$ in $GL(n, 2)$ obtained by flipping one entry of the current matrix $C$.
        \item Transition from $C$ to $C_{new}$ with probability dependent on the change in energy $e(C_{new})-e(C)$, according to a specified \emph{energy function}. In our case, $\mathrm{energy}(C)$ is the sum of 1 entries in  $C \cdot L_Z$ and $\left(C^T\right)^{-1} \cdot L_X$.
    \end{enumerate}
    \item Return the current point $C$ in parameter space at the end of the iterations.
\end{enumerate}

\begin{algorithm}
\caption{Simulated Annealing}
\begin{algorithmic}[1]

\Procedure{Simulated\_Annealing}{$C_0$}       %\Comment{This is a test}
    \State $C \leftarrow C_0$
    \For{$k = 0, \ldots, K-1$}
	\State $T \leftarrow \mathrm{temperature\_at}(k)$
	\State $C_{new} \leftarrow \mathrm{random\_neighbour\_of}(C)$
	\State $p \leftarrow \mathbb{P}(\text{transition} \:|\: \mathrm{energy}(C), \mathrm{energy}(C_{new}), T)$
	\If{$p \leq \mathrm{random}(0,1)$}
		\State $C \leftarrow C_{new}$
	\EndIf
    \EndFor
    \State \textbf{output} $C$
\EndProcedure

\end{algorithmic}
\end{algorithm}
\noindent For our problem, we chose the following temperature schedule and transition probability:
$$
    T := T_0 \left(1-\frac{k}{K}\right)
    \hspace{2cm}
    \mathbb{P}(\text{transition})
    :=
    \min\left(1, \exp\left(\frac{e(C)-e(C_{new})}{T}\right)\right)
$$
The neighbouring matrix $C_{new}$ is obtained by flipping a randomly chosen entry in $C$, but care must be taken to ensure that $C_{new}$ is invertible: this is done by rejection sampling: (The asymptotic probability that a uniformly sampled binary matrix is invertible is 0.288 \autocite{invertible}.) 
\section{Implementation Details}
Unfortunately, \texttt{numpy} does not support $GF(2)$ as a data type, so we work with the \texttt{numpy.uint64} data type and take the remainder mod 2 after matrix multiplication or matrix inversion. To check whether a matrix is invertible, we just need to check that its determinant is odd. From the \emph{Hadarmard determinant problem}, we know that the maximum determinant achievable by a $n \times n$ binary matrix grows exponentially \autocite{brenner1972hadamard}:
$$
    \det(C) \leq \frac{(n+1)^{(n+1)/2}}{2^n}
$$
Unsigned integer overflow is not undefined behaviour in C, the underlying language used in \texttt{numpy}, so it is perfectly acceptable to compute the determinant of large matrix modulo $2^{64}$ this way. Unfortunately, \texttt{numpy}'s optimised determinant method \texttt{np.linalg.det} computes the determinant via LU factorisation using the LAPACK routine \texttt{z/dgetrf} \autocite{det}, which uses doubled precision floating point numbers instead of integers \autocite{lapack}. IEEE-754 gives double precision floating point numbers 53 bits of mantissa, so matrices with $n \leq 33$ will not have determinant overflow problems.
Furthermore, because binary matrix rank in $\mathbb{R}$ is not the same as binary matrix rank in $GF(2)$, we cannot use \texttt{np.linalg.matrix\_rank} to check for invertibility.  
Instead, we used Gaussian elimination to compute the rank and invert the matrix in $GF(2)$. 

In practice, quantum circuit synthesis libraries such as \tket \autocite{tket}, Qiskit \autocite{Qiskit} and PyZX \autocite{pyzx} only allow users to add primitive gates such as CNOT and rotation gates. To apply our algorithm, we must be able 
to detect phase gadgets and find repeating layers of them when they exist. The detection of phase gadgets can be done by pushing the position of the CNOTs onto a stack, then popping them off after reaching the rotation gate. Every CNOT removed from the stack becomes a leg in the phase gadget. The detection of repeating layers of phase gadgets can be done using the Knuth-Morris-Pratt string matching algorithm \autocite{kmp}. 

\section{Evaluation Results}
First we evaluate the effectiveness of simulated annealing on the proposed problem. To do this we generated binary matrices $L_z$ and $L_x$ uniformly at random, then used simulated annealing to find a good $C$ to minimise the sum. We perform univariate analysis on the four following parameters:

\begin{enumerate}
    \item Attempts: The number of times simulated annealing is run, taking the result with the lowest objective score. 
    \item Iterations: The number of steps taken in one simulated annealing attempt.
    \item Width: The width of matrices $L_z$ and $L_x$, which corresponds to the number of gadgets. 
    \item Height: The height of matrices $L_z$, $L_x$ and $C$. This corresponds to the number of qubits in the circuit.
\end{enumerate}
\begin{figure}
\resizebox{\textwidth}{!}{
\begin{tabular}{|c|c|c|c|}
\hline
Attempts & Iterations & Width & Height \\
\hline
(20, 20, 20, 10, 6000) & (20, 20, 20, 20, 5000) & 
(20, 40, 40, 10, 5000) & (40, 20, 20, 10, 5000) \\
(20, 20, 20, 10, 5000) & (20, 20, 20, 10, 5000) & 
(20, 20, 20, 10, 5000) & (20, 20, 20, 10, 5000) \\
(20, 20, 20, 10, 4000) & (20, 20, 20,  5, 5000) & 
(20, 10, 10, 10, 5000) & (10, 20, 20, 10, 5000) \\
\hline
\end{tabular}}

\resizebox{\textwidth}{!}{
\begin{tabular}{cc}
\includegraphics[width=\textwidth/2]{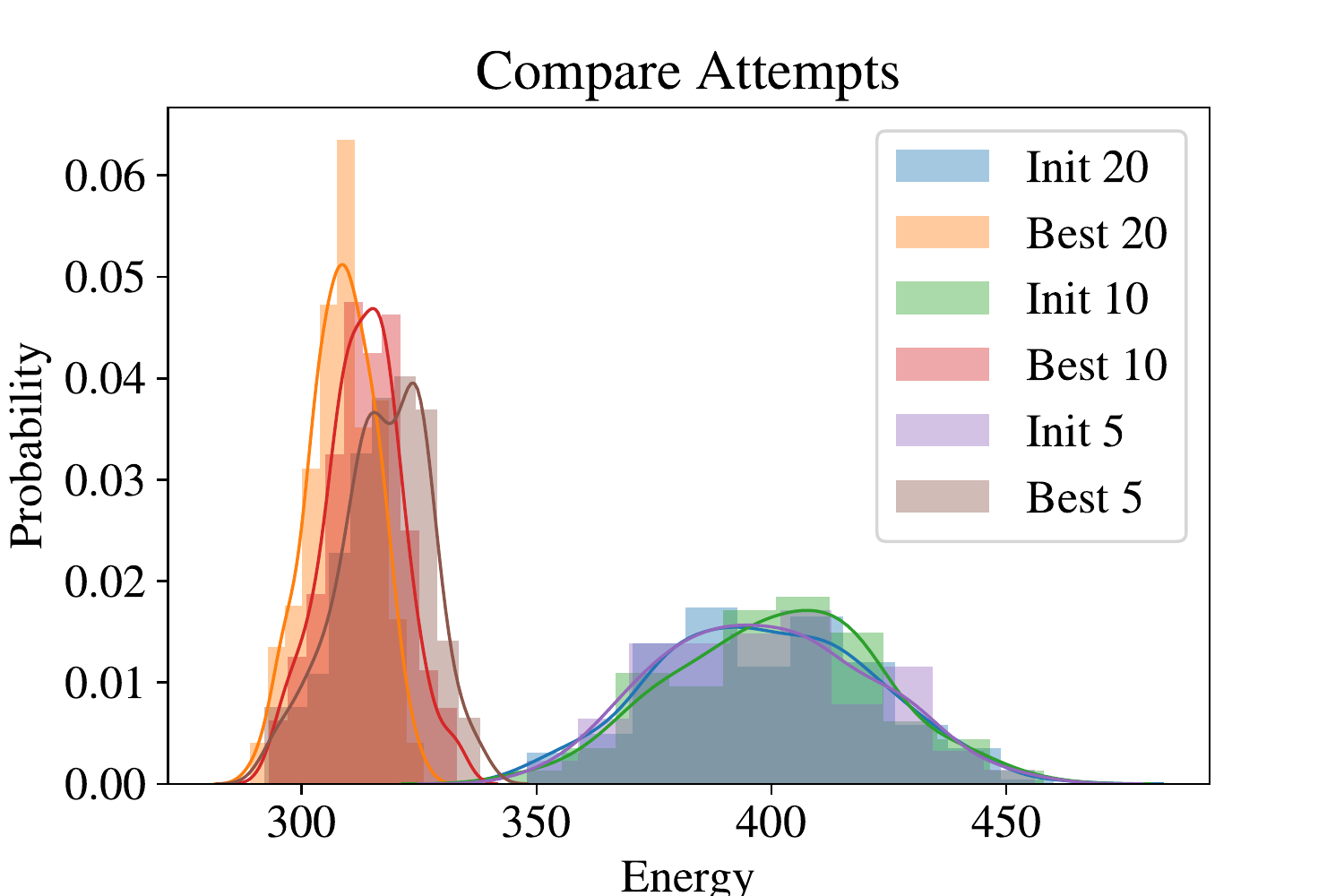} & 
\includegraphics[width=\textwidth/2]{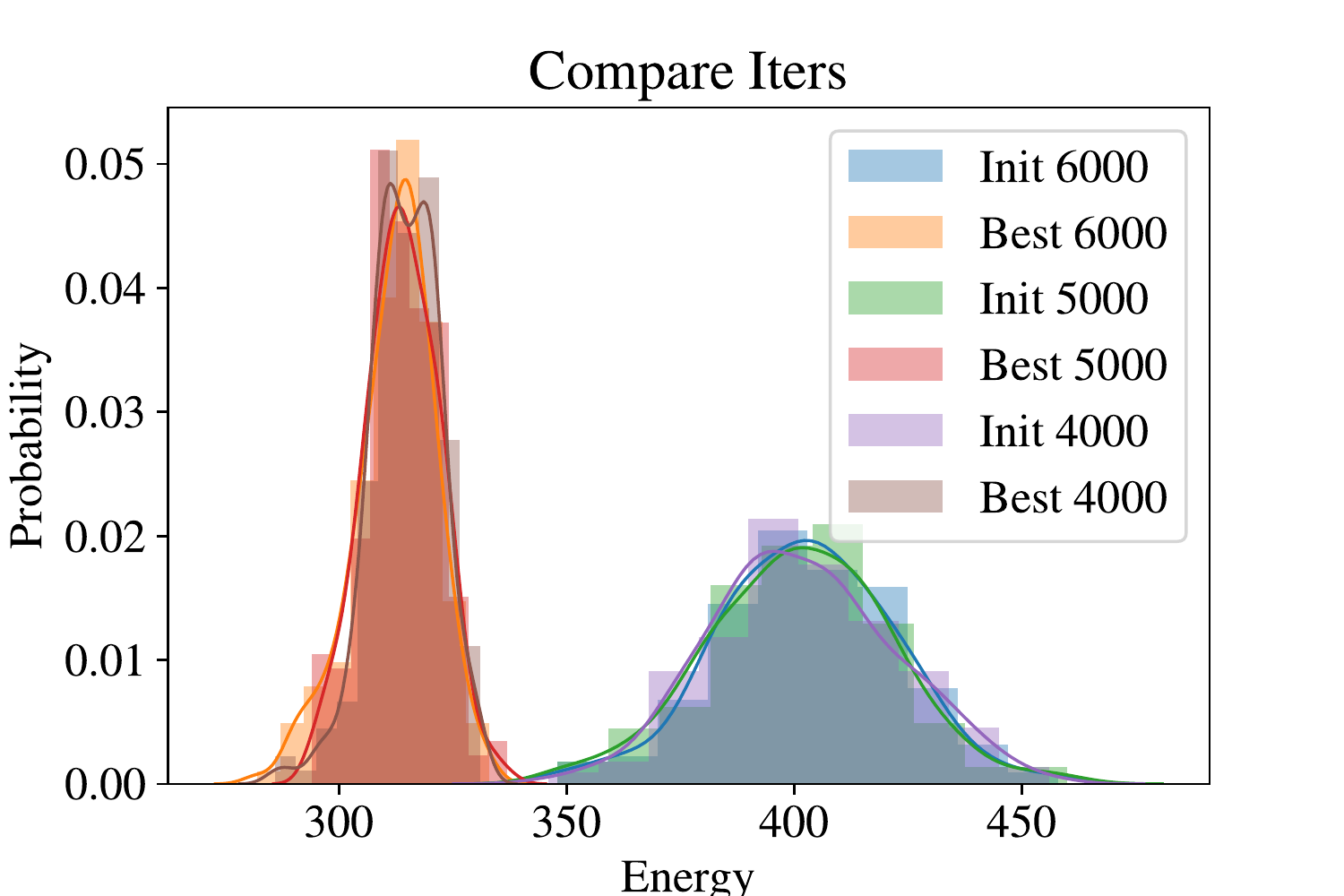} \\
\includegraphics[width=\textwidth/2]{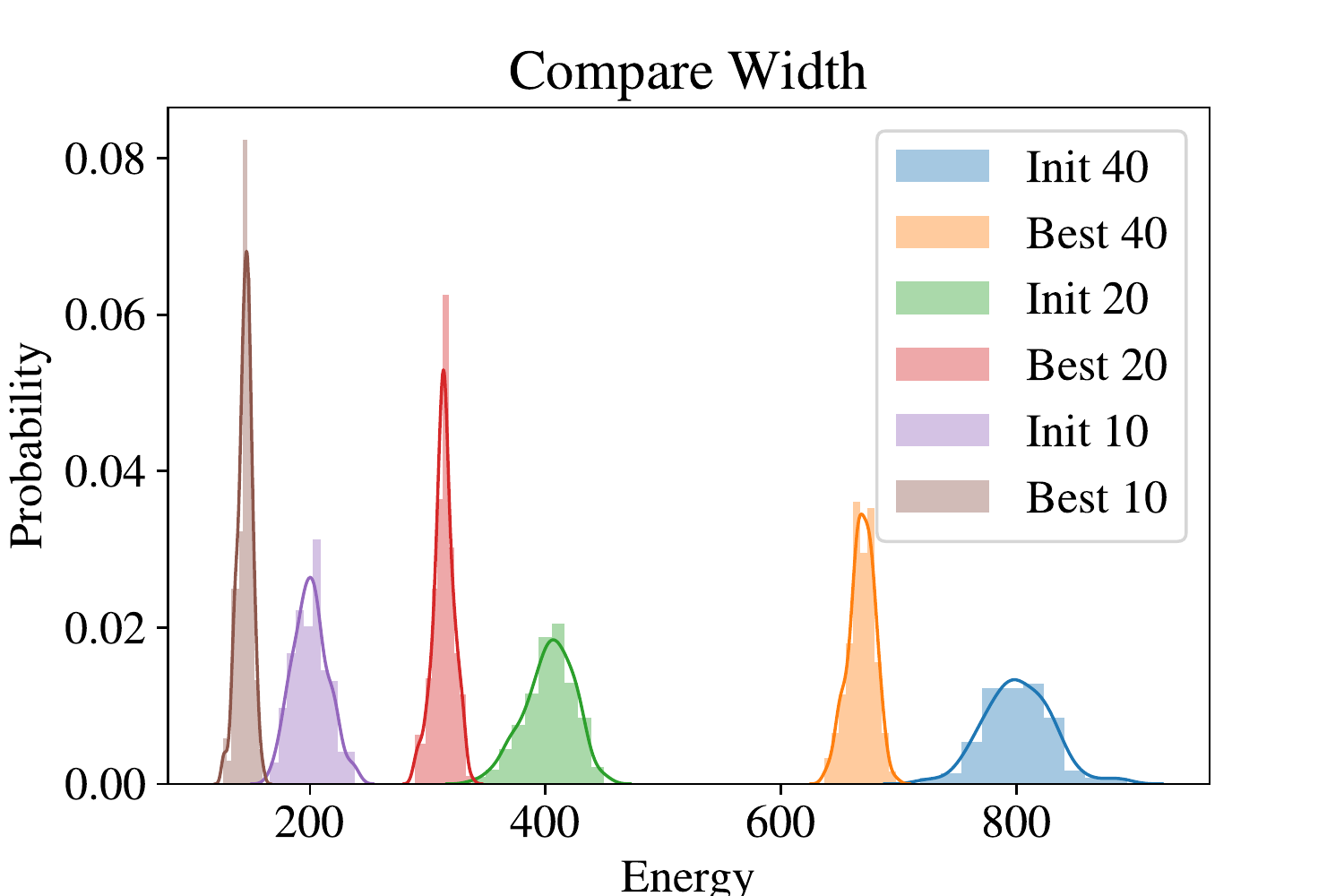} & 
\includegraphics[width=\textwidth/2]{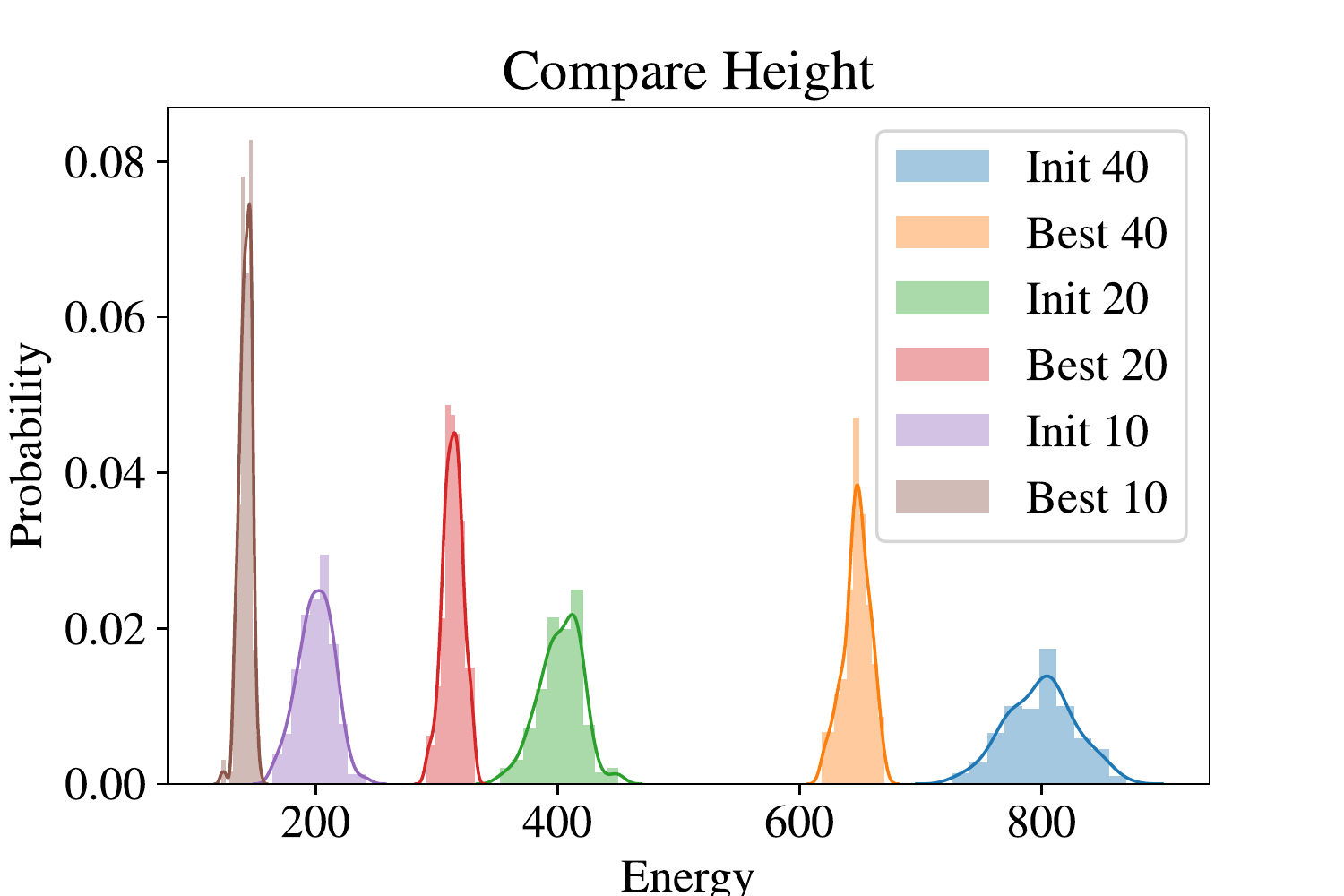} 
\end{tabular}}

\caption{The parameters used to benchmark the performance of simulated annealing (top), and the distribution of scores (bottom).}
\end{figure}

As expected, the mean of the scores increase as the number of attempts increase. The percentage improvement of simulated annealing decreases as the size of the problem increases, but perhaps this can be ameliorated by changing the hyperparameters.

To evaluate the performance of our simplification algorithm we have applied it to random ans\"{a}tze with a varying a number of Z and X gadgets per layer, as well as a varying number of repeating layers. For each \texttt{(n\_gadgets, n\_layers)} configuration pair, we sampled 10 ans\"{a}tze on 8 qubits with uniform probability and we computed the mean CNOT depth and count, for both the original circuits and the simplified circuits:

\begin{figure}
    \centering
    \resizebox{\textwidth}{!}{
    \begin{tabular}{|c||ccc|ccc|ccc|ccc|}
     \hline
     \multicolumn{13}{|c|}{CNOT depth} \\
     \hline
     \# gadgets& \multicolumn{3}{c|}{1 layer}
               & \multicolumn{3}{c|}{2 layers}
               & \multicolumn{3}{c|}{5 layers}
               & \multicolumn{3}{c|}{10 layers} \\
    per layer & before & after & $\Delta$\% & before & after & $\Delta$\% & before & after & $\Delta$\% & before & after & $\Delta$\%\\
    \hline
    10 & 125 & 111 & 11\% & 258 & 146 & 44\% & 597 & 228 & 62\% & 1146 & 389 & 66\% \\
    20 & 261 & 192 & 26\% & 486 & 276 & 43\% & 1159 & 595 & 49\% & 2388 & 1121 & 53\% \\
    40 & 498 & 367 & 26\% & 1014 & 647 & 36\% & 2296 & 1518 & 34\% & 4818 & 2946 & 39\% \\
    80 & 988 & 763 & 23\% & 1924 & 1463 & 24\% & 4986 & 3479 & 30\% & 9521 & 6999 & 26\% \\
    160 & 1974 & 1619 & 18\% & 3906 & 3200 & 18\% & 9476 & 7726 & 18\% & 19259 & 15707 & 18\% \\    
    \hline
    \end{tabular}}

    \vspace{5mm}

    \resizebox{\textwidth}{!}{
    \begin{tabular}{|c||ccc|ccc|ccc|ccc|}
     \hline
     \multicolumn{13}{|c|}{CNOT count} \\
     \hline
     \# gadgets& \multicolumn{3}{c|}{1 layer}
               & \multicolumn{3}{c|}{2 layers}
               & \multicolumn{3}{c|}{5 layers}
               & \multicolumn{3}{c|}{10 layers} \\
    per layer & before & after & $\Delta$\% & before & after & $\Delta$\% & before & after & $\Delta$\% & before & after & $\Delta$\%\\
    \hline
    10 & 120 & 149 & worse & 249 & 185 & 26\% & 580 & 261 & 55\% & 1136 & 414 & 64\% \\
    20 & 252 & 233 & 8\% & 472 & 324 & 31\% & 1154 & 663 & 43\% & 2416 & 1172 & 51\% \\
    40 & 490 & 414 & 15\% & 988 & 692 & 30\% & 2292 & 1580 & 31\% & 4768 & 3056 & 36\% \\
    80 & 976 & 804 & 18\% & 1908 & 1526 & 20\% & 4950 & 3598 & 27\% & 9496 & 7072 & 26\% \\
    160 & 1953 & 1672 & 14\% & 3868 & 3249 & 16\% & 9382 & 7825 & 17\% & 19100 & 15750 & 18\% \\
    \hline
    \end{tabular}}
    \caption{A summary for the performance of the optimisation algorithm on randomly sampled circuits.}
\end{figure}

\noindent When layers are not repeated, the results indicate that savings in CNOT count and depth from phase gadget simplification are mostly offset by the introduction of the CNOT blocks.
However, the savings increase up to 60\% as the number of repeating layers is increased, because the CNOT blocks are not repeated.
The optimisation problem becomes progressively more difficult as the number of gadgets per layer increases, resulting in worse optimization results (all other things held equal).
However, tweaking the parameters of simulated annealing improves the results.

\begin{landscape}

\begin{figure}
\centering
\begin{varwidth}{\linewidth}
\begin{verbatim}
Initial Circuit:
Circuit on 6 qubits with 64 gates.
12 is the T-count
52 Cliffords among which 
52 2-qubit gates and 0 Hadamard gates.
\end{verbatim}
\end{varwidth}
$$\input{routine/real1.tikz}$$
\begin{varwidth}{\linewidth}
\begin{verbatim}
Output Circuit:
Circuit on 6 qubits with 54 gates.
12 is the T-count
42 Cliffords among which 
42 2-qubit gates and 0 Hadamard gates.
\end{verbatim}
\end{varwidth}
$$\input{routine/real2.tikz}$$

\caption{Example execution with \texttt{n\_qubits = 6, n\_gadgets = 6, n\_layers = 2}}
\end{figure}
\end{landscape}